\documentclass[USenglish]{lmcs}
\pdfoutput=1

\usepackage{lastpage}
\renewcommand{\dOi}{14(4:22)2018}
\lmcsheading{}{1--\pageref{LastPage}}{}{}%
{Jun.~07,~2017}{Dec.~07,~2018}{}

\usepackage[tikz,textwidth=2.9cm,mathtools=false,theorems=false,
            mnsymbol=false,paralist=compact,todo=hide,colors=false,
            newdocumentcommand=true]{generic} 
\usepackage{mathabx}
\usepackage[amsfonts]{logic-thm}
\usepackage{wrapfig}
\usepackage{capt-of}
\usepackage[square,numbers,sort]{natbib}
\usetikzlibrary{automata,arrows,decorations.pathreplacing,patterns}

\newcommand{\anca}[2][]{\todo[size=\scriptsize,color=blue!30,#1]{A. #2}}
\newcommand{\gabriele}[2][]{\todo[size=\scriptsize,color=orange!30,#1]{G. #2}}
\newcommand{\olivier}[2][]{\todo[size=\scriptsize,color=green!30,#1]{O. #2}}
\newcommand{\felix}[2][]{\todo[size=\scriptsize,color=red!30,#1]{F. #2}}
\newcommand{\reviewOne}[2][]{\todo[size=\scriptsize,color=brown!30,#1]{R1. #2}}
\newcommand{\reviewTwo}[2][]{\todo[size=\scriptsize,color=cyan!30,#1]{R2. #2}}

\newcommand{\cmax}{{\boldsymbol{C}}} 
\newcommand{\hmax}{{\boldsymbol{H}}} 
\newcommand{\hmaxsq}{{\hmax^2}} 
\newcommand{\emax}{{\boldsymbol{E}}} 
\newcommand{\bound}{\boldsymbol{B}} 

\renewcommand{\cT}{T}
\renewcommand{\Tt}{\typeout{Use $\cT$ instead of $\Tt$}}
\renewcommand{\cA}{A}
\renewcommand{\cB}{B}
\renewcommand{\cC}{C}
\renewcommand{\cD}{D}
\renewcommand{\Aa}{\typeout{Use $\cA$ instead of $\Aa$}}
\renewcommand{\Bb}{\typeout{Use $\cB$ instead of $\Bb$}}
\renewcommand{\Cc}{\typeout{Use $\cC$ instead of $\Cc$}}
\renewcommand{\Dd}{\typeout{Use $\cD$ instead of $\Dd$}}

\newcommand{\PR}[1]{{\normalfont\bfseries P#1}}

\let\oldvdash\vdash
\let\olddashv\dashv
\renewcommand{\vdash}{\mathop{\oldvdash}}
\renewcommand{\dashv}{\mathop{\olddashv}}
\newcommand{\Left}{\ensuremath{\mathsf{left}}}
\newcommand{\Right}{\ensuremath{\mathsf{right}}}

\newcommand{\tr}[1]{\ensuremath{\mathsf{tr}}(#1)}
\newcommand{\an}[1]{\ensuremath{\mathsf{an}}(#1)}
\newcommand{\out}[1]{\ensuremath{\mathsf{out}}(#1)}
\newcommand{\outb}[1]{\ensuremath{\mathsf{out}}\big(#1\big)}
\newcommand{\block}[1]{\ensuremath{\mathsf{block}}(#1)}

\renewcommand{\dom}{\ensuremath{\mathsf{dom}}}

\newcommand{\pump}{\ensuremath{\mathsf{pump}}}
\newcommand{\lbinom}[2]{\left(\begin{array}{l}#1 \\
      #2\end{array}\right)}

\newcommand{\LL}{{\ensuremath{\mathsf{LL}}}}
\renewcommand{\RR}{{\ensuremath{\mathsf{RR}}}}
\newcommand{\LR}{{\ensuremath{\mathsf{LR}}}}
\newcommand{\RL}{{\ensuremath{\mathsf{RL}}}}

\let\oldlhd\lhd 
\let\oldrhd\rhd
\renewcommand{\lhd}{{\oldlhd}} 
\renewcommand{\rhd}{{\oldrhd}}

\newcommand{\crossrel}{\ensuremath{\mathrel{\text{\sf S}}}}

\renewcommand{\simeq}{\crossrel^*}

\newcommand{\lesstime}{\mathrel{\lhd}}
\newcommand{\leqtime}{\mathrel{\unlhd}}

\newcommand{\geqtime}{\mathrel{\unrhd}}

\newcommand{\lesspair}{\sqsubset}

\providecommand{\underbracket}[2][]{\underbrace{#2}}
\providecommand{\overbracket}[2][]{\overbrace{#2}}

\makeatletter
\def\shortrightarrowfill@{\arrowfill@\relbar\relbar\shortrightarrow}
\def\shortleftarrowfill@{\arrowfill@\shortleftarrow\relbar\relbar}
\def\shortleftrightarrowfill@{\arrowfill@\leftrightarrow}
\newcommand{\ort}{\mathpalette{\overarrow@\shortrightarrowfill@}}
\newcommand{\olft}{\mathpalette{\overarrow@\shortleftarrowfill@}}
\newcommand{\olftrt}{\mathpalette{\overarrow@\shortleftrightarrowfill@}}
\makeatother

\makeatletter
\newcommand{\fixed@sra}{$\vrule height 2\fontdimen22\textfont2 width 0pt\shortrightarrow$}
\newcommand{\upperleft}{{\mspace{-2mu}\text{\rotatebox[origin=c]{\numexpr135}{\fixed@sra}}\mspace{-2mu}}}
\newcommand{\lowerright}{{\mspace{-2mu}\text{\rotatebox[origin=c]{\numexpr315}{\fixed@sra}}\mspace{-2mu}}}
\newcommand{\leftshort}{{\shortleftarrow}}
\newcommand{\rightshort}{{\shortrightarrow}}
\makeatother

\makeatletter
\let\pmb\boldsymbol
\makeatother

\makeatletter 
\newcommand\noparbreak{\par\nobreak\@afterheading} 
\makeatother

\keywords{Regular word transductions, Two-way transducers, Sweeping transducers, One-way definability}

\begin{document}
\title{One-way definability of two-way word transducers}
\titlecomment{{\lsuper*}The results presented here were published
  in~\cite{bgmp15} and \cite{bgmp17}. The work was partially supported by the projects ExStream  (ANR-13-JS02-0010)
          and DeLTA (ANR-16-CE40-0007).}
\author[F.~Baschenis]{F\'elix Baschenis\rsuper{1}}	
\address{\lsuper{1}Department of Mathematics and Informatics, University Bremen}	
\author[O.~Gauwin]{Olivier Gauwin\rsuper{2}}	
\address{\lsuper{2}LaBRI, University of Bordeaux, CNRS}	
\author[A.~Muscholl]{Anca Muscholl\rsuper{2}}	
\address{\vskip-7pt} 
\author[G.~Puppis]{Gabriele Puppis\rsuper{3}}	
\address{\lsuper{3}CNRS, LaBRI}	
\thanks{}	

\begin{abstract}
Functional transductions realized by two-way transducers 
(or, equally, by streaming transducers or MSO transductions) 
are the natural and standard notion of ``regular'' mappings from words to words.
It was shown in 2013 that it is decidable if such a 
transduction can be implemented by some one-way transducer, 
but the given algorithm has non-elementary complexity. 
We provide an algorithm of different flavor solving the 
above question, that has doubly exponential space complexity. In the
special case of sweeping transducers the complexity is one exponential
less.  We also show how to
construct an equivalent one-way transducer, whenever
it exists, in doubly or triply exponential time, again depending on
whether the input transducer is sweeping or two-way. In the sweeping
case our construction is shown to be optimal. 


\end{abstract}

\maketitle

\reviewOne[inline]{One can deplore that said figures are not always present to help understand some notation heavy proofs, but as it stands, the paper is nicely illustrated.}%

\reviewOne[inline]{The fact that some long, intricate proofs do not benefit the same degree of structure and illustrations is the only blemish.}%

\reviewOne[inline]{It is announced that those results are already published in two other papers [3,4]. If some results are new to this version, it should be made explicit. If not, disregard this comment.
\olivier[inline]{done, added paragraph in intro}%
}%

\section{Introduction}\label{sec:introduction}

Since the early times of computer science, transducers have been
identified as a fundamental computational model. Numerous
fields of computer science are ultimately concerned with
transformations, ranging from databases to image
processing, and an important challenge is to perform transformations with
low costs, whenever possible. 

The most basic form of transformations is obtained using devices that
process an input with finite memory
and produce outputs during the processing. Such
devices are called finite-state transducers. Word-to-word finite-state
transducers were considered in very early work in formal language
theory~\cite{sch61,ahu69,eilenberg1974automata}, and it was soon clear that they are much more
challenging than finite-state word acceptors (classical finite-state
automata). One essential difference between transducers and automata
over words is that the capability to process the
input in both directions strictly increases the expressive power in the case
of transducers, whereas this does not for
automata~\cite{RS59,she59}. In other words, two-way word transducers
are strictly more expressive than one-way word transducers.

We consider in this paper  functional transducers, that compute
functions from words to words. Two-way word transducers capture very
nicely the notion of regularity in this setting. Regular word
functions, in other words,  functions computed by functional two-way transducers%
\footnote{We know from \cite{EH98,de2013uniformisation} that deterministic 
          and non-deterministic functional two-way transducers have the 
          same expressive power, though the non-deterministic variants are 
          usually more succinct.}, 
inherit many of the characterizations and algorithmic properties of the robust
class of regular languages. Engelfriet and Hoogeboom~\cite{EH98} 
showed that monadic second-order definable graph transductions,
restricted to words, are equivalent to two-way transducers --- this
justifies the 
notation ``regular'' word functions, in the spirit of classical
results in automata theory and logic by B\"uchi, Elgot, Rabin and
others. Recently, Alur and Cern\'{y}~\cite{AlurC10} proposed an enhanced
version of one-way transducers called streaming transducers, and
showed that they are equivalent to the two previous models. A streaming
transducer processes the input word from left to right, and stores
(partial) output words in some given write-only registers, that are
updated through concatenation and constant updates.

\reviewOne[inline]{Few reasons are given for a specific focus towards two-way transducers as opposed to equivalent models.
  \olivier[inline]{I answered to the editor that the main reason is that we are more familiar with it, so it seems natural. But we shouldn't mention it here.}
}%
\anca[inline]{I agree not to say it here, but I would not say this
  is the real reason - and I changed the answer. The real reason (for
  me) is that pumping SST is much less intuitive.}%

Two-way transducers raise challenging questions about resource
requirements. One crucial resource is the number of times the
transducer needs to re-process the input word. In particular, the case
where the input can be processed in a single pass, from left to right,
is very attractive as it corresponds to the setting of \emph{streaming},
where the (possibly very large) inputs do not need to be stored in order
to be processed. The~\emph{one-way definability} of a functional
two-way transducer, that is, the question 
whether the transducer is equivalent to some one-way transducer, was
considered quite recently: \cite{fgrs13} shows that one-way
definability of string transducers is a decidable property. However, the decision procedure
of~\cite{fgrs13} has non-elementary complexity, which raises the
question about the intrinsic complexity of this problem.

In this paper we provide an algorithm of elementary complexity solving
the one-way definability problem. 
Our decision algorithm has single or doubly exponential space complexity,
depending on whether the input transducer is sweeping (namely, it performs reversals
only at the extremities of the input word) or genuinely two-way.
We also describe an algorithm that constructs an equivalent one-way
transducer, whenever it exists, in doubly or triply exponential time,
again depending on whether the input transducer is sweeping or
two-way.   For the
construction of an equivalent one-way transducer we obtain a doubly
exponential lower bound, which is tight for sweeping transducers. Note that for the decision problem,
the best lower bound known is only polynomial space~\cite{fgrs13}.
\reviewOne[inline]{The subclass of sweeping transducers is central to
  the paper: not only is the algorithm of better complexity in this
  particular case, it is also used as an easy case to prepare for the
  general proof. However, beyond this property, it is unclear whether
  this class has other interesting properties, or has garnered
  interest in the part. This information should, ideally, be made
  available in the introduction.  \olivier[inline]{added a paragraph
    below.}%
}%
\reviewOne[inline]{Were deterministic two-way word transducers
  considered at all (Introduction)? If not, do you plan to consider
  them (Conclusion)?  \olivier[inline]{added the footnote in intro,
    and a paragraph in conclusion}%
  \felix[inline]{I think the question is "do you have a better
    complexity if the given transducer is deterministic?" like
    usually. We maybe also need to add something like : One can notice
    that if we want to decide definability by \emph{deterministic}
    one-way transducers, an efficient procedure is known since
    Choffrut78 \olivier[inline]{you can complete the conclusion if you
      want. To me it is sufficient.}%
  }%
}%

\olivier{added this paragraph}%
\gabriele{Changed a bit (without the reviewer comment the paragraph was a bit isolated)}%
\anca{changed a bit more}%
Our initial interest in sweeping transducers was the fact that they
provide a simpler setting for characterizing one-way
definability. Later it turned out that sweeping transducers enjoy interesting and
useful connections with streaming transducers: they have the same
expressiveness as streaming transducers where concatenation of
registers is disallowed. The connection goes even further, since the number of
sweeps corresponds exactly to the number of
registers~\cite{bgmp16}. The results of this paper were refined
in~\cite{bgmp16}, and used to determine the minimal number of
registers required by functional streaming transducers without
register concatenation.


\subsection*{Related work.} 
Besides the papers mentioned above, there are
several recent results around the expressivity and the resources of
two-way transducers, or 
equivalently, streaming transducers. 
First-order definable transductions were shown to be equivalent to
transductions defined by aperiodic
streaming transducers~\cite{FiliotKT14} and to aperiodic two-way
transducers~\cite{CartonDartois15}. An effective characterization of
aperiodicity for one-way transducers was obtained in~\cite{FGL16}.

Register minimization for right-appending deterministic streaming
transducers 
was shown to be
decidable in~\cite{DRT16}. An algebraic characterization of (not necessarily
functional) two-way transducers over unary alphabets was 
provided in~\cite{CG14mfcs}. It was shown that in this case 
sweeping transducers have the same expressivity. 
The expressivity of non-deterministic input-unary
or output-unary two-way
transducers was investigated in~\cite{Gui15}. 

In \cite{Smith14} a pumping lemma for two-way transducers is proposed,
and used to investigate properties of the output languages of  two-way
transducers. In this paper we also rely on pumping arguments over runs
of two-way transducers, but we require loops of a particular form,
that allows to identify periodicities in the output.
\olivier{TODO ref Ismael to be improved before resubmission}%

The present paper unifies results on one-way definability obtained
in~\cite{bgmp15} and~\cite{bgmp17}. 
Compared to the conference versions, some combinatorial proofs have
\gabriele{I have removed the reference to Saarela}%
been simplified,
and the complexity of the procedure presented in~\cite{bgmp15} has been 
improved by one exponential. 
\anca{please check. I deleted the sentence on the
      under-approximation, because I don't find it understandable at this point.}%
\olivier{added this paragraph}%


\subsection*{Overview.} 
Section~\ref{sec:preliminaries} introduces basic
notations for two-way and sweeping transducers, and Section~\ref{sec:overview}
states the main result and provides a roadmap of the proofs. For
better readability our paper
is divided in two parts: in the first part we consider the easier case of
sweeping transducers, and in the second part the general case. The two
proofs are similar, but the general case is more involved since we
need to deal with special loops (see
Section~\ref{sec:loops-twoway}).  Since the high-level ideas of the
proofs are the same and the sweeping case illustrates them in a
simpler way, the proof in that setting is a preparation for the
general case.  Both proofs have
the same structure: first we introduce some combinatorial arguments
(see Section~\ref{sec:combinatorics-sweeping} for the sweeping case
and Section~\ref{sec:combinatorics-twoway} for the general case), then we provide the
characterization of one-way definability (see
Section~\ref{sec:characterization-sweeping} for the sweeping case and
Section~\ref{sec:characterization-twoway} for the general
case). Finally, Section~\ref{sec:complexity} establishes the complexity
of our algorithms.



\section{Preliminaries}\label{sec:preliminaries}

We start with some basic notations and definitions for two-way
automata and transducers. We assume that every input word $w=a_1\cdots a_n$ 
has two special delimiting symbols $a_1 = \vdash$
and $a_n = \dashv$ that do not occur elsewhere: $a_i \notin \{\vdash,\dashv\}$ 
for all $i=2,\dots,n-1$. 

A \emph{two-way automaton} is a tuple 
$\cA=(Q,\Sigma,\vdash,\dashv,\Delta,I,F)$, where 
\begin{itemize}
  \item $Q$ is a finite set of states, 
  \item $\Sigma$ is a finite alphabet (including $\vdash, \dashv$), 
  \item $\Delta \subseteq Q \times \Sigma \times Q \times \set{\Left,\Right}$ 
        is a transition relation, 
  \item $I,F\subseteq Q$ are sets of initial and final states, respectively.
\end{itemize}
By convention, left transitions on $\vdash$ are not
allowed; on the other hand, right transitions on $\dashv$ are allowed, but, as we will see, 
they will necessarily appear as last transitions of successful runs.
A \emph{configuration} of $\cA$ has the form $u\,q\,v$, 
with $uv \in \vdash\,  \S^* \, \dashv$ 
and $q \in Q$. A configuration $u\,q\,v$ represents the 
situation where the current state of $\cA$ is $q$ and its head reads the first 
symbol of $v$ (on input $uv$). If $(q,a,q',\Right) \in \Delta$,
then there is a transition from any configuration of the form
$u\,q\,av$ to the configuration $ua\,q'\,v$; we denote such a transition
by $u\,q\,av \trans{a,\Right} ua\,q'\,v$. 
Similarly, if $(q,a,q',\Left) \in \Delta$, 
then there is a transition from any configuration of the form
$ub\,q\,av$ to the configuration $u\,q'\,bav$, 
denoted as $ub\,q\,av \trans{a,\Left} u\,q'\,bav$.
A \emph{run} on $w$ 
is a sequence of transitions. 
It is \emph{successful} if it starts in an initial configuration 
$q\, w$, with $q\in I$, and ends in a final configuration $w\,q'$,
with $q' \in F$ --- note that this latter configuration does not allow 
additional transitions. The \emph{language} of $\cA$ is the set of words
that admit a successful run of $\cA$.

The definition of \emph{two-way transducers} is similar to that of two-way automata,
with the only difference that now there is an additional output alphabet $\Gamma$ and the transition 
relation is a finite subset of $Q \times \Sigma \times \Gamma^* \times  Q \times \{\Left,\Right\}$,
which associates an output over $\Gamma$ with each transition of the underlying two-way automaton.
For a two-way transducer $\cT=(Q,\Sigma,\vdash,\dashv,\Gamma,\Delta,I,F)$,
we have a transition of the form $ub\,q\,av \trans{a,d \, \mid w} u'\,q'\,v'$, outputting $w$, 
whenever $(q,a,w,q',d)\in\Delta$ and either $u'=uba ~\wedge~ v'=v$ or 
$u'=u ~\wedge~ v'=bav$, depending on whether $d=\Right$ or $d=\Left$.
The \emph{output} associated with a  run 
$\rho = u_1\,q_1\,v_1 \trans{a_1,d_1 \mid w_1}<\qquad> \dots
\trans{a_n,d_n \mid w_n}<\qquad> u_{n+1}\,q_{n+1}\,v_{n+1}$
of $\cT$ is the word $\out{\rho} = w_1\cdots w_n$. A transducer $\cT$
defines a relation $\sL(\cT)$ 
\anca{I was wondering if we ever use this notation, but 
      I found it on page 10... Maybe it is the only place...}%
consisting of all pairs $(u,w)$ such that
$w=\out{\rho}$, for some successful run $\rho$ on $u$. 
The \emph{domain} of $\cT$, denoted $\dom(\cT)$, 
is the set of input words that have a successful run. 
For transducers $\cT,\cT'$, we write $\cT' \subseteq \cT$ 
to mean that $\dom(\cT') \subseteq \dom(\cT)$ and the transductions 
computed by $\cT,\cT'$ coincide on $\dom(\cT)$.


A transducer is called \emph{one-way} if it does not contain transition
rules of the form $(q,a,w,q',\Left)$. It is called \emph{sweeping} if
it can perform reversals only at the borders of the input word.
A transducer that is equivalent to some one-way
(resp.~sweeping) transducer is called \emph{one-way definable}
(resp.~\emph{sweeping definable}).

The \emph{size} of a transducer takes into
account both the state space and the transition relation, and thus
includes the length of the output of each transition.

\medskip
\subsection*{Crossing sequences.}
The first notion that we use throughout the paper is that of crossing sequence. 
We follow the convenient presentation from \cite{HU79}, which appeals to a
graphical representation of runs of a two-way transducer, where each configuration
is seen as a point (location) in a two-dimensional space. 
Let $u=a_1\cdots a_n$ be an input word (recall that $a_1=\vdash$ and $a_n=\dashv$) 
and let $\rho$ be a run of a two-way automaton (or transducer) on $u$.
The \emph{positions} of $\rho$ are the numbers from $0$ to $n$, corresponding
to ``cuts'' between two consecutive letters of the input. For example, 
position $0$ is just before the first letter $a_1$,
position $n$ is just after the last letter $a_n$,
and any other position $x$, with $1\le x<n$, is between 
the letters $a_x$ and $a_{x+1}$.
We will denote by $u[x_1,x_2]$ the factor of $u$ between 
the positions $x_1$ and $x_2$ (both included).
\olivier{added}%

\gabriele{Please check this reformulation in terms of configurations}%
Each configuration $u\,q\,v$ of a two-way run $\r$
has a specific position associated with it. For technical reasons 
we need to distinguish leftward and rightward transitions. If the configuration $u\,q\,v$ 
is the target of a rightward transition 
then the position associated with $u\,q\,v$ is  $x=|u|$.  
The same definition also applies when $u\,q\,v$ is the initial configuration,
for which we have $u=\emptystr$ and $x=|u|=0$.
Otherwise, if $u\,q\,v$ is the target of a leftward transition 
then the position associated with $u\,q\,v$ is $x=|u|+1$. Note that in
both cases, the letter read by the transition leading to $u\,q\,v$ is
$a_x$. 
\anca{rewritten a bit}%
A \emph{location} of $\rho$ is any pair $(x,y)$, 
where $x$ is the position of some configuration of $\rho$ and
$y$ is any non-negative integer for which there are at least 
$y+1$ configurations in $\rho$ with the same position $x$.
The second component $y$ of a location is called \emph{level}.
For example, in Figure \ref{fig:run} we represent a possible
run of a two-way automaton together with its locations $(0,0)$,
$(1,0)$, $(2,0)$, $(2,1)$, etc. 
Each location is naturally associated with a configuration, and thus a state. 
Formally, we say that $q$ is the \emph{state at location $\ell=(x,y)$} in $\rho$, 
and we denote this by writing $\rho(\ell)=q$, 
if the $(y+1)$-th configuration of $\rho$ with position $x$ has state $q$.
Finally, we define the \emph{crossing sequence} of $\rho$ at position $x$ 
as the tuple $\rho|x=(q_0,\dots,q_h)$, where the $q_y$'s are all the states 
at locations of the form $(x,y)$, for $y=0,\dots,h$.

\begin{figure}[!t]
\centering
\scalebox{0.85}{
\begin{tikzpicture}[->,>=stealth',shorten >=1pt,auto,node distance=3.2cm,semithick,xscale=2,yscale=1]

  \node[state,minimum size=2mm,fill=gray!50] (A)       at (0,2)             {$q_0  $};
  \node[state,minimum size=2mm,fill=gray!50] (B)   [right of=A]                 {$  q_1  $};
  \node[state,minimum size=2mm,fill=gray!50] (C)   [right of=B]                 {$  q_2   $};
  \node[state,minimum size=2mm,fill=gray!50] (D)   [above=1cm of C]                 {$  q_3  $};
   \node[state,minimum size=2mm,fill=gray!50] (E)   [left of=D]                 {$ q_4   $};
 
  \node[state,minimum size=2mm,fill=gray!50] (F)   [above=1cm of E]                 {$  q_5    $};
  \node[state,minimum size=2mm,fill=gray!50] (G)   [right of=F]                 {$ q_6    $};
   \node[state,minimum size=2mm,fill=gray!50] (H)   [right of=G]                 {$ q_7  $};
   \node[state,minimum size=2mm,fill=gray!50] (I) [right of=H] { $q_8$} ;
 \coordinate[below=1cm of A] (d1);   
   
  \path (A) edge              node {\small $a_1 , \Right $ } (B) ;
   \path (B) edge              node {\small $a_2 , \Right $ } (C) ;
    \path (C) edge [out=0,in=0]       node [swap] {\small $a_3 , \Left $ } (D) ;
    \path (D) edge              node [above] {\small $a_2 , \Left $ } (E) ;
     \path (E) edge [out=180,in=180]       node  {\small $a_1 , \Right $ } (F) ;
       \path (F) edge              node {\small $a_2 , \Right $ } (G) ;
         \path (G) edge              node {\small $a_3 , \Right $ } (H) ;
          \path (H) edge              node {\small $a_4 , \Right $ } (I) ;
    
    \node[state,rectangle,minimum size=3mm] (X) at (0.8,-0.6) {$a_1$};
    \node[state,rectangle,minimum size=3mm] (Y) [right of=X] {$a_2$};
    \node[state,rectangle,minimum size=3mm] (Z) [right of=Y] {$a_3$};
    \node[state,rectangle,minimum size=3mm] (W) [right of=Z] {$a_4$};
    
    \node at (-1.5,-0.7) {Input word:};
    \node at (-1.5,0.5) {Positions:};
    \node at (-1.5,2) {Run:};
    
    \node (M) at (0,0.5) {$0$};
    \node (N) [right of=M] {$1$};
    \node (O) [right of=N] {$2$};
    \node (P) [right of=O] {$3$};
    \node (Q) [right of=P] {$4$};

  \node (AA) [above=0cm of A] {$(0,0)$} ;
  \node (AB) [right of=AA] {$(1,0)$};
  \node (AC) [right of=AB] {$(2,0)$};
  \node (AD) [above=0cm of D] {$(2,1)$};
  \node (AE) [left of=AD] {$(1,1)$};
  \node (AF) [above=0cm of F] {$(1,2)$};
  \node (AG) [right of=AF] {$(2,2)$};
  \node (AH) [right of=AG] {$(3,0)$};
  \node (AI) [right of=AH] {$(4,0)$};
\end{tikzpicture}
}
\caption{Graphical presentation of a run by means of crossing sequences.}\label{fig:run}
\end{figure}
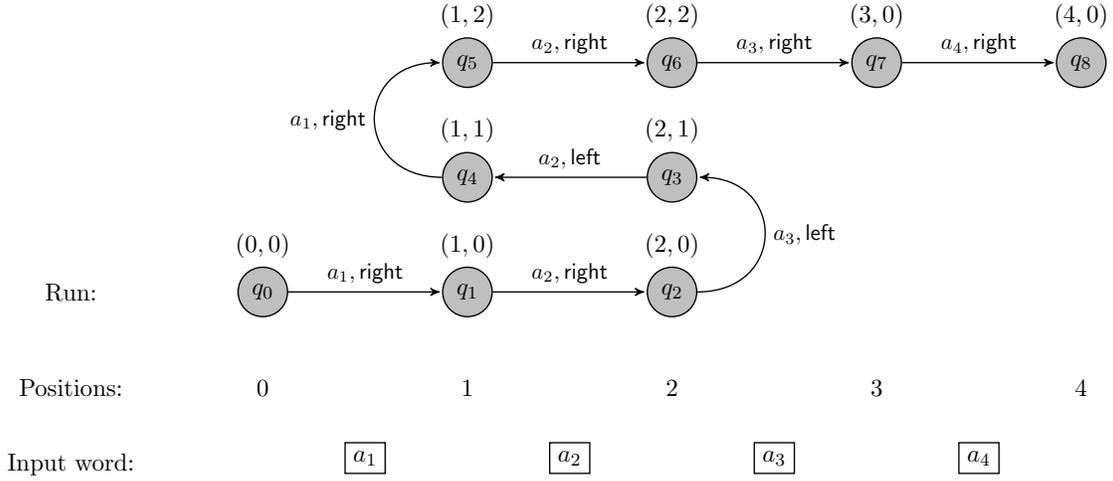


As shown in Figure~\ref{fig:run}, a two-way run can be represented 
as an path between locations annotated by the associated states.
We observe in particular that if a location $(x,y)$ is the target
of a rightward transition, then this transition has read the symbol $a_x$;
similarly, if $(x,y)$ is the target of a leftward transition, then 
the transition has read the symbol $a_{x+1}$.  
We also observe that, in any successful run $\rho$, every crossing 
sequence has odd length and every rightward (resp.~leftward) transition 
reaches a location with even (resp.~odd) level.
In particular, we can identify four types of transitions between locations, 
depending on the parities of the levels (the reader may refer again to 
Figure~\ref{fig:run}):
\begin{center}
\begin{tikzpicture}[baseline=0,scale=0.9]
  \draw (-1,1.75) node (node1) {$\phantom{\!+\!1}~(x,2y)$};
  \draw (3,1.75) node (node2) {$(x\!+\!1,2y')~\phantom{\!+\!1}$};
  \draw (-1,0) node (node3) {$(x,2y\!+\!1)$};
  \draw (3,0) node (node4) {$(x\!+\!1,2y'\!+\!1)$};
  \draw (9,1.75) node (node5) {$\phantom{\!+\!1}~(x,2y)$};
  \draw (9,2.5) node (node6) {$(x,2y\!+\!1)$};
  \draw (9,0) node (node7) {$(x,2y\!+\!1)$};
  \draw (9,0.75) node (node8) {$(x,2y\!+\!2)$};
  \draw (node1) edge [->] node [above=-0.05, scale=0.9] {\small $a_{x+1},\Right$} (node2);
  \draw (node4) edge [->] node [above=-0.05, scale=0.9] {\small $a_{x+1},\Left$} (node3);
  \draw (node5.east) edge [->, out=0, in=0, looseness=2] 
        node [right=0, scale=0.9] {\small $a_{x+1},\Left$} (node6.east);
  \draw (node7.west) edge [->, out=180, in=180, looseness=2] 
        node [left=0, scale=0.9] {\small $a_x,\Right$} (node8.west);
\end{tikzpicture}
\vspace{1mm}
\end{center}
Hereafter, we will identify runs with the corresponding annotated paths between locations. 

It is also convenient to define a total order $\leqtime$ on the locations of a run $\rho$ 
by letting $\ell_1 \leqtime \ell_2$ if $\ell_2$ is reachable from $\ell_1$ by following the 
path described by $\rho$ --- the order $\leqtime$ on locations is called 
\emph{run order}. 
Given two locations $\ell_1 \leqtime \ell_2$ of a run $\rho$, we write $\rho[\ell_1,\ell_2]$ 
for the factor of the run that starts in $\ell_1$ and ends in $\ell_2$. Note that the latter 
is also a run and hence the notation $\outb{\rho[\ell_1,\ell_2]}$ is permitted.
\gabriele{Parity of levels needs to be preserved}%
We will often reason with factors of runs up to isomorphism, that is, modulo 
shifting the coordinates of their locations while preserving the parity of the levels.
Of course, when the last location of (a factor of) a run $\rho_1$ 
coincides with the first location of (another factor of) a run $\rho_2$, 
then $\rho_1$ and $\rho_2$ can be concatenated to form a longer run,
denoted by $\rho_1 \rho_2$.
\olivier{added, see review 2 below}%
\gabriele{rephrased a bit, and I think corrected a problem with negative numbers.
          But do we really need to say this?}%
This operation can be performed even if the two locations, 
say $(x_1,y_1)$ and $(x_2,y_2)$, are different, provided that
$y_1=y_2\bmod 2$: in this case it suffices to shift the positions
(resp.~levels) of the locations of the first run $\rho_1$ by $x_2$ (resp.~$y_2$)
and, similarly, the positions (resp.~levels) of the locations of the second
run $\rho_2$ by $x_1$ (resp.~$y_1$).
\felix{it's not really enough, with our current definition the first location of $\rho_2$ is $(x,0)$ is $x$ is the position of the cut. So we cannot even concatenate $\rho_1$ and $\rho_2$. I have an alternate paragraph where the locations of the subruns use the levels of the original run.}%
\gabriele{I don't understand your comment Felix. Does it still apply?}%

\reviewTwo[inline]{
There seems to be some problems in you definition of crossing sequences.

First, according to your definition, the state at location $(x,y)$ is the state reached by the $(y+1)$-th transition that crosses x. In order for this to be consistent with Figure 1, you should swap u and u' in your definition of a position crossed by a transition.
\olivier[inline]{Correct, there was a mistake there! Check my fix above.}%

Second, since, following your definition, the state at a given location always is the target of a transition, the first state of the run has no location (it should be added explicitly).
\olivier[inline]{Correct, also fixed}%

Third, not allowing more freedom in the second components of the locations seems to cause some problems while considering subruns.
Let me explain what I mean with an example.
Let us suppose that we split the run $\rho$ of your Figure 1 into two runs $\rho_1$ and $\rho_2$, where $\rho_2$ is composed of the last three transitions.
Then the locations of the states of the run $\rho_1$ (considered as a whole run) are equal to the locations of the same states considered as part of $\rho$.
However, the locations of the states of the run $\rho_2$ (considered as a whole run) are (1,0), (2,0), (3,0), (4,0), whereas in $\rho$, the locations of the corresponding states are  (1,2), (2,2), (3,0), (4,0).
As a consequence, $\rho_1$ and $\rho_2$ can not be concatenated to reform $\rho$, since to do so we need that “the last location of [\dots] $\rho_1$ coincides with the first location of [\dots] $\rho_2$” (p.5 l.15).
\olivier[inline]{Added a sentence to explicit}%

Finally (and this one is not a real problem), you introduce precise notions, but sometimes do not use them in the following parts of the paper: for example, p.4, “never uses a left transition from position x” corresponds to “no transition crosses position x”. Why not use the latter ?
\olivier[inline]{``no left transition on $\vdash$'' comes before the definition of ``crosses'' and is more explicit, so I would not change.}%
}%
\anca{there is no ``crosses'' anymore, right?}%

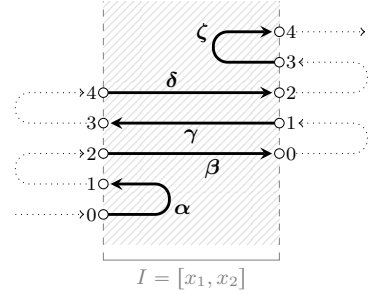
\begin{wrapfigure}{r}{5.1cm}
\centering
\scalebox{0.9}{
\begin{tikzpicture}[baseline=0, inner sep=0, outer sep=0, minimum size=0pt, xscale=0.65, yscale=0.45]
  \tikzstyle{dot} = [draw, circle, fill=white, minimum size=4pt]
  \tikzstyle{fulldot} = [draw, circle, fill=black, minimum size=4pt]
  \tikzstyle{factor} = [->, shorten >=1pt, dotted, rounded corners=6]
  \tikzstyle{fullfactor} = [->, >=stealth, shorten >=1pt, very thick, rounded corners=6]

  \fill [pattern=north east lines, pattern color=gray!25]
        (4,-1) rectangle (8,7);
  \draw [dashed, thin, gray] (4,-1) -- (4,7);
  \draw [dashed, thin, gray] (8,-1) -- (8,7);
  \draw [gray] (4,-1.25) -- (4,-1.5) -- (8,-1.5) -- (8,-1.25);
  \draw [gray] (6,-1.75) node [below] {\footnotesize $I=[x_1,x_2]$};

  \draw (2,0) node (node0) {};
  \draw (4,0) node [dot] (node1) {};
  \draw (node1) node [left=0.1] (node1') {\scriptsize $0$};
  \draw (5.5,0) node (node2) {};
  \draw (5.5,1) node (node3) {};
  \draw (4,1) node [dot] (node4) {};
  \draw (node4) node [left=0.1] (node4') {\scriptsize $1$};
  \draw (2,1) node (node5) {};
  \draw (2,2) node (node6) {};
  \draw (4,2) node [dot] (node7) {};
  \draw (node7) node [left=0.1] (node7') {\scriptsize $2$};
  \draw (8,2) node [dot] (node8) {};
  \draw (node8) node [right=0.1] (node8') {\scriptsize $0$};
  \draw (10,2) node (node9) {};
  \draw (10,3) node (node10) {};
  \draw (8,3) node [dot] (node11) {};
  \draw (node11) node [right=0.1] (node11') {\scriptsize $1$};
  \draw (4,3) node [dot] (node12) {};
  \draw (node12) node [left=0.1] (node12') {\scriptsize $3$};
  \draw (2,3) node (node13) {};
  \draw (2,4) node (node14) {};
  \draw (4,4) node [dot] (node15) {};
  \draw (node15) node [left=0.1] (node15') {\scriptsize $4$};
  \draw (8,4) node [dot] (node16) {};
  \draw (node16) node [right=0.1] (node16') {\scriptsize $2$};
  \draw (10,4) node (node17) {};
  \draw (10,5) node (node18) {};
  \draw (8,5) node [dot] (node19) {};
  \draw (node19) node [right=0.1] (node19') {\scriptsize $3$};
  \draw (6.5,5) node (node20) {};
  \draw (6.5,6) node (node21) {};
  \draw (8,6) node [dot] (node22) {};
  \draw (node22) node [right=0.1] (node22') {\scriptsize $4$};
  \draw (10,6) node (node23) {};

  \draw [factor] (node0) -- (node1');
  \draw [fullfactor] (node1) -- (node2.center) -- node [below right=1mm] {\footnotesize $\pmb{\alpha}$}
                     (node3.center) -- (node4); 
  \draw [factor] (node4') -- (node5.center) -- (node6.center) -- (node7'); 
  \draw [fullfactor] (node7) -- node [below right=1.25mm] {\footnotesize $~~\pmb{\beta}$} (node8);
  \draw [factor] (node8') -- (node9.center) -- (node10.center) -- (node11');
  \draw [fullfactor] (node11) -- node [below=1mm] {\footnotesize $\pmb{\gamma}$} (node12);
  \draw [factor] (node12') -- (node13.center) -- (node14.center) -- (node15');
  \draw [fullfactor] (node15) -- node [above left=1.25mm] {\footnotesize $\pmb{\delta}~~$} (node16);
  \draw [factor] (node16') -- (node17.center) -- (node18.center) -- (node19');
  \draw [fullfactor] (node19) -- (node20.center) -- node [above left=1mm] {\footnotesize $\pmb{\zeta}$}
                     (node21.center) -- (node22);
  \draw [factor] (node22') -- (node23);
\end{tikzpicture}
}
\caption{Intercepted factors.}\label{fig:intercepted-factors}

\end{wrapfigure}
%
\medskip
\subsection*{Intercepted factors.}
For simplicity, we will denote by $\omega$ 
the maximal position of the input word.
We will consider \emph{intervals of positions} of 
the form $I=[x_1,x_2]$, with $0 \le x_1<x_2 \le \omega$. 
The \emph{containment} relation $\subseteq$ on intervals is 
defined expected, as $[x_3,x_4] \subseteq [x_1,x_2]$ if $x_1\le x_3 < x_4\le x_2$.
%
%
A \emph{factor} of a run $\rho$ is a contiguous subsequence of $\rho$. 
A factor of $\rho$ \emph{intercepted} by an interval 
$I=[x_1,x_2]$ is a maximal factor  that visits only 
positions $x\in I$, and never uses a left transition from 
position $x_1$ or a right transition from position $x_2$.  
Figure~\ref{fig:intercepted-factors} on the right shows 
the factors $\alpha,\beta,\gamma,\delta,\zeta$ intercepted
by an interval $I$.
The numbers that annotate the endpoints of the factors 
represent their levels.

\medskip
\subsection*{Functionality.}
We say that a transducer is \emph{functional}  (equivalently, one-valued, or single-valued)
if for each input $u$, at most one output $w$ can be
produced by any possible successful run on $u$.
Of course, every deterministic transducer is functional, while the opposite 
implication fails in general. 
To the best of our knowledge determining the precise complexity of the determinization 
of a two-way transducer (whenever an equivalent deterministic one-way transducer exists)
is still open. From classical bounds on determinization of finite automata, we only know
that the size of a determinized transducer may be exponential in the worst-case. 
One solution to this question, which is probably
not the most efficient one,  is to
check one-way definability: if an equivalent one-way transducer is
constructed, one can check in $\ptime$ if it can be
determinized~\cite{cho77,weberklemm95,bealcarton02}. 

The following result,  proven in Section \ref{sec:complexity}, 
is the  reason to consider only functional transducers:

{
\renewcommand{\thethm}{\ref{prop:undecidability}}
\begin{prop}[]
The one-way definability problem for \emph{non-functional} 
sweeping transducers is undecidable.
\end{prop}
}

Unless otherwise stated, hereafter we tacitly assume that all transducers 
are functional. Note that functionality is a decidable property, as shown below. 
The proof of this result is similar to the decidability proof for the equivalence problem
of deterministic two-way transducers~\cite{Gurari80}, as it reduces the functionality
problem to the reachability problem of a $1$-counter automaton of exponential size. 
A matching $\pspace$ lower bound follows by a reduction of the emptiness problem 
for the intersection of finite-state automata \cite{Kozen77}.

\begin{prop}\label{prop:functionality-pspace}
Functionality of two-way transducers can be decided in polynomial
space. This problem is $\pspace$-hard already for sweeping transducers.
\end{prop}

A (successful) run of a two-way transducer is called \emph{normalized}
if it never visits two locations with the same position, the same
state, and both either at even or at odd level. It is easy to see that
if a successful run $\rho$ of a \emph{functional} transducer visits two locations
$\ell_1=(x,y)$ and $\ell_2=(x,y')$ with the same state
$\rho(\ell_1)=\rho(\ell_2)$ and with $y=y' \mod 2$,  then the output
produced by $\rho$ between $\ell_1$ and $\ell_2$ is empty:
otherwise, by repeating the non-empty factor $\rho[\ell_1,\ell_2]$, we would
 contradict functionality.  So, by deleting the factor
 $\rho[\ell_1,\ell_2]$ we obtain a successful run that produces the
 same output. Iterating this operation leads to an equivalent,
 normalized run.

 Normalized runs are interesting because their crossing sequences have
 bounded length (at most $\hmax = 2|Q|-1$). Throughout the paper we
 will implicitly assume that successful runs are normalized. The
 latter property can be easily checked on crossing sequences.


\section{One-way definability: overview}\label{sec:overview}

In this section we state our main result, which is the existence of an
elementary algorithm for checking whether a two-way transducer is
equivalent to some one-way transducer. We call such transducers
\emph{one-way definable}. Before stating our result, we
start with  a few examples illustrating the reasons that may prevent a
transducer to be one-way definable.

\begin{exa}\label{ex:one-way-definability}
We consider two-way transducers that accept any input $u$ 
from a given regular language $R$ and produce as output the word $u\,u$. 
We will argue how, depending on $R$, these transducers may or may not be one-way definable.
\begin{enumerate}
\item If $R=(a+b)^*$, then there is no equivalent one-way transducer, 
      as the output language is not regular. 
      If $R$ is finite, however, then the transduction mapping $u\in R$ to $u\,u$ 
      can be implemented by a one-way transducer that stores the input $u$ 
      (this requires at least as many states as the cardinality of $R$),
      and outputs $u\,u$ at the end of the computation.
\item A special case of transduction with finite domain is obtained from the language
      $R_n = \{ a_0 \, w_0 \, \cdots a_{2^n-1} \, w_{2^n-1} \::\: a_i\in\{a,b\} \}$,
      where $n$ is a fixed natural number, the input alphabet is $\{a,b,0,1\}$, 
      and each $w_i$ is the binary encoding of the index $i=0,\dots,2^n-1$ 
      (hence $w_i\in\{0,1\}^n$).
      According to Proposition~\ref{prop:lower-bound} below, the transduction 
      mapping $u\in R_n$ to $u\,u$ can be implemented by a two-way transducer 
      of size $\cO(n^2)$, but every equivalent one-way transducer 
      has size (at least) doubly exponential in $n$.
\item Consider now the periodic language $R=(abc)^*$. 
      The function that maps $u\in R$ to $u\,u$ can be easily implemented by a 
      one-way transducer: it suffices to output alternatively $ab$, $ca$, $bc$
      for each input letter, while checking that the input is in $R$.
\end{enumerate}
\end{exa}

\begin{exa}\label{ex:running}
We consider now a slightly more complicated transduction
that is defined on input words of the form $u_1 \:\#\: \cdots \:\#\: u_n$,
where each factor $u_i$ is over the alphabet $\Sigma=\{a,b,c\}$. 
The associated output has the form
$w_1 \:\#\: \cdots \:\#\: w_n$, where each $w_i$
is either $u_i \: u_i$ or just $u_i$, depending on whether or not
$u_i\in (abc)^*$ and $u_{i+1}$ has even length, with $u_{n+1}=\emptystr$
by convention.

\noindent
The natural way to implement this transduction is by means
of a two-way transducer that performs multiple passes 
on the factors of the input:
a first left-to-right pass is performed on 
$u_i \,\#\, u_{i+1}$ to produce the first copy of $u_i$ 
and to check whether $u_i\in (abc)^*$ and $|u_{i+1}|$ is even; if so,
a second pass on $u_i$ is performed to produce 
another copy of $u_i$.

\noindent
Observe however that the above transduction can also be implemented by a one-way
transducer, using non-determinism: when entering a factor $u_i$, the transducer
guesses whether or not $u_i\in (abc)^*$ and $|u_{i+1}|$ is even; 
depending on this it outputs either $(abc\,abc)^{\frac{|u_i|}{3}}$ 
or $u_i$, and checks that the guess is correct while proceeding to
read the input.
\end{exa}

\noindent

\medskip

The  main result of our paper is an elementary algorithm that decides
whether a functional transducer is one-way definable:

\begin{thm}\label{thm:main}
There is an algorithm that takes as input a functional two-way
transducer $\cT$ and outputs in $3\exptime$ a \emph{one-way} transducer 
$\cT'$ satisfying the following properties:
\begin{enumerate}
  \item $\cT'\subseteq\cT$,
  \item $\dom(\cT')=\dom(\cT)$ if and only if $\cT$ is one-way definable.
  \item $\dom(\cT')=\dom(\cT)$ can be checked in $2\expspace$.
\end{enumerate}
Moreover, if $\cT$ is a sweeping transducer, then $\cT'$ can be
constructed in $2\exptime$
and $\dom(\cT')=\dom(\cT)$ is decidable in $\expspace$.
\end{thm}

\begin{rem}
The transducer $\cT'$ constructed in the above theorem is
in a certain sense maximal: for every 
$v \in \dom(\cT) ~\setminus~ \dom(\cT')$ 
and every one-way transducer $\cT''$ with 
$\dom(\cT') \subseteq \dom(\cT'') \subseteq \dom(\cT)$ there exists 
some witness input $v'$ obtained from $v$ such that 
$v' \in \dom(\cT) ~\setminus~ \dom(\cT'')$. We will make this more
precise at the end of Section~\ref{sec:characterization-twoway}.
\end{rem}

We also provide a two-exponential lower bound for the size of the equivalent transducer. 
As the lower bound is achieved by a sweeping transduction (even a deterministic one), 
this gives a tight lower bound on the size of any one-way transducer equivalent to 
some sweeping transducer.

\begin{prop}\label{prop:lower-bound} 
\gabriele{I have added the fact that the sweeping transducers are even deterministic}%
There is a family $(f_n)_{n\in\bbN}$ of transductions such that
\begin{enumerate}
  \item $f_n$ can be implemented by a deterministic sweeping transducer of size $\cO(n^2)$,
  \item $f_n$ can be implemented by a one-way transducer,
  \item every one-way transducer that implements $f_n$ 
        has size $\Omega(2^{2^n})$. 
\end{enumerate}
\end{prop}

\begin{proof}
The family of transformations is precisely the one described in 
Example~\ref{ex:one-way-definability}~(2), where $f_n$ maps inputs of the form
$u = a_0 \, w_0 \, \cdots \, a_{2^n-1} ~ w_{2^n-1}$ to outputs of the form $u\,u$,
where $a_i\in\{a,b\}$ and $w_i\in\{0,1\}^n$ is the binary encoding of $i$.
A deterministic sweeping transducer implementing $f_n$ first checks that the 
binary encodings $w_i$, for $i=0,\dots,2^n-1$, are correct. 
This can be done with $n$ passes: 
the $j$-th pass uses $\cO(n)$ states to check the correctness of 
the $j$-th bits of the binary encodings.
Then, the sweeping transducer performs two additional passes to 
copy the input twice. Overall, the sweeping transducer has size $\cO(n^2)$. 

As already mentioned, every one-way transducer that implements $f_n$ needs 
to remember input words $u$ of exponential length in order to output $u\,u$, which
roughly requires doubly exponentially many states. 
A more formal argument providing a lower bound for the size of a one-way
transducer implementing $f_n$ goes as follows. 

First of all, one observes that given a one-way transducer $\cT$, 
the language of its outputs, 
i.e., $L^{\text{out}}_\cT = \{ w \::\: (u,w)\in\sL(\cT) \text{ for
  some } u\}$ 
is regular. More precisely, if $\cT$ has size $N$, then the language 
$L^{\text{out}}_\cT$ is recognized by an automaton of size linear in $N$.
Indeed, while parsing $w$, the automaton can guess an input word $u$
and a run on $u$, 
together with a factorization of $w$ in which the $i$-th 
factor corresponds to the output of the transition 
on the $i$-th letter of $u$. Basically, this requires 
storing as control states the transition rules of $\cT$ and the 
suffixes of outputs.

Now, suppose that the function $f_n$ is implemented by a one-way 
transducer $\cT$ of size $N$. The language
$L^{\text{out}}_\cT = \{ u\,u \::\: u\in\dom(f_n) \}$ is then
recognized by an automaton of size $\cO(N)$.
Finally, we recall a result from \cite{GlaisterShallit96}, which shows that,
given a sequence of pairs of words $(u_i,v_i)$, for $i=1,\dots,M$,
every non-deterministic automaton that separates the language 
$\{u_i\,v_i \::\: 1\le i\le M\}$ from the language 
$\{u_i\,u_j \::\: 1\le i\neq j\le M\}$ must have at least $M$ states.
By applying this result to our language $L^{\text{out}}_\cT$, where 
$u_i=v_i$ for all $i=1,\dots,M=2^{2^n}$, we get that $N$ must be at
least linear in $M$, and hence $N \in \Omega(2^{2^n})$.
\end{proof}

\bigskip
\gabriele{improved a bit, but we may do better}%
The proof of Theorem~\ref{thm:main} will be developed in the next sections. 
The main idea is to decompose a run of the two-way transducer $\cT$ 
into factors that can be easily simulated in a one-way manner. We defer the 
formal definition of such a decomposition to Section \ref{sec:characterization-sweeping}, 
while here we refer to it simply as a ``$\bound$-decomposition'', where 
$\bound$ is a suitable number computed from $\cT$.
The reader can refer to Figure~\ref{fig:decomposition-sweeping} on page
\pageref{fig:decomposition-sweeping}, which provides some intuitive
account of a $\bound$-decomposition for a sweeping run.
Roughly speaking, 
each factor 
of a $\bound$-decomposition either already looks like a run of a 
one-way transducer (e.g.~the factors $D_1$ and $D_2$ of Figure~\ref{fig:decomposition-sweeping}),
or it produces a periodic output, where the period is bounded by $\bound$
(e.g.~the factor between $\ell_1$ and $\ell_2$). 
Identifying factors that look like runs of one-way transducers is rather easy.
On the other hand, to identify factors with periodic outputs we rely on
a notion of ``inversion'' of a run. Again, we defer the formal definition
and the important combinatorial properties of inversions 
to Section \ref{sec:combinatorics-sweeping}. 
The reader can refer to Figure \ref{fig:inversion-sweeping} 
on page \pageref{fig:inversion-sweeping} 
for an example of an inversion of a run of a sweeping transducer.
Intuitively, this is a portion of run that is potentially difficult to simulate
in a one-way manner, due to existence of long factors of the output that are 
generated following the opposite order of the input.
Finally, the complexity of the decision procedure in Theorem~\ref{thm:main}
is analyzed in Section~\ref{sec:complexity}.

\reviewOne[inline]{As neither inversions nor blocks are defined by Page 8, the presence of Theorem 3.6 is not as enlightening as one might hope. If the notion of inversion and their relevance is made quite clear at a glance of Page 11, the definition of run decomposition is not as intuitive. Hence, Theorem 3.6 cannot be appreciated before Page 17. It might be a good idea to give an intuition about inversion and diagonal/blocks before 3.6 (this solution being preferable, but made challenging by how involved the notion of decomposition seems to be in a first read).
  \felix[inline]{added the references to the corresponding figures}%
  \olivier[inline]{and Gabriele added some intuitions on decompositions}%
}%

\subsection*{Roadmap.} 
In order to provide a roadmap of our proofs, we state below the equivalence 
between the key properties related to one-way definability, inversions of runs,
and existence of decompositions:

\begin{thm}\label{thm:main2}
Given a functional two-way transducer $\cT$, 
an integer $\bound$ can be computed such that the following are equivalent:
\begin{itemize}
  \item[\PR1)] $\cT$ is one-way definable,
  \item[\PR2)] for every successful run of $\cT$ and every inversion in it, 
               the output produced amid
                the inversion has period at most $\bound$,
  \item[\PR3)] every input has a successful run of $\cT$ that admits a $\bound$-decomposition. 
\end{itemize}
\end{thm}


As the notions of inversion and $\bound$-decomposition are simpler to formalize
for sweeping transducers, we will first prove the theorem assuming that $T$ is 
a sweeping transducer; we will focus later on unrestricted two-way transducers.
Specifically, in Section~\ref{sec:combinatorics-sweeping} we introduce the basic combinatorics
on words and the key notion of inversion for a run of a sweeping transducer, and we prove the 
implication \PR1 $\Rightarrow$ \PR2.
In Section~\ref{sec:characterization-sweeping} we define $\bound$-decompositions of runs of 
sweeping transducers, prove the implication \PR2 $\Rightarrow$ \PR3, and sketch a proof of 
\PR3 $\Rightarrow$ \PR1 (as a matter of fact, this latter implication can be proved in a way 
that is independent of whether $\cT$ is sweeping or not, which explains why we 
only sketch the proof in the sweeping case).
Section~\ref{sec:loops-twoway} lays down the appropriate definitions
concerning loops of  two-way transducers, and analyzes in detail the effect of
pumping special idempotent loops.
In Section~\ref{sec:combinatorics-twoway} we further develop the combinatorial arguments
that are used to prove the implication \PR1 $\Rightarrow$ \PR2 in the general case.
Finally, in Section~\ref{sec:characterization-twoway} we prove the implications 
\PR2 $\Rightarrow$ \PR3 $\Rightarrow$ \PR1 in the general setting, 
and show how to decide the condition $\dom(\cT')=\dom(\cT)$ of Theorem \ref{thm:main}.


\section{Basic combinatorics for sweeping transducers}\label{sec:combinatorics-sweeping}

We fix for the rest of the section a functional \emph{sweeping transducer} $\cT$,
an input word $u$, and a (normalized) successful run $\rho$ of $\cT$ on $u$.

\medskip
\subsection*{Pumping loops.}
Loops turn out to be a basic concept for characterizing one-way definability. 
Formally, a \emph{loop} of $\rho$ is an interval $L=[x_1,x_2]$ such that $\rho|x_1=\rho|x_2$,
namely, with the same crossing sequences at the extremities.
The run $\rho$ can be pumped at any loop $L=[x_1,x_2]$, and this gives rise
to new runs with iterated factors. Below we study precisely the shape of 
these pumped runs.

\begin{defi}[anchor point, trace]
  Given a loop $L$ and a location $\ell$ of $\rho$, we say that $\ell$ is an 
\emph{anchor point in $L$} if $\ell$ is the first location of some factor 
of $\rho$ that is intercepted by $L$; 
this factor is then denoted%
\footnote{This is a slight abuse of notation, since the factor $\tr{\ell}$ 
          is not determined by $\ell$ alone, but requires also the knowledge of the loop $L$,
          which is usually clear from the context.} 
as $\tr{\ell}$ and called the \emph{trace of $\ell$}.
\end{defi}

Observe that a loop can have at most $\hmax = 2|Q|-1$ anchor points, since
we consider only normalized runs.

Given a loop $L$ of $\rho$ and a number $n\in\bbN$, we can replicate $n$ times 
the factor $u[x_1,x_2]$ of the input, obtaining a new input of the form
\begin{equation}\label{eq:pumped-word}
  \pump_L^{n+1}(u) ~=~ u[1,x_1]\cdot \big(u[x_1+1,x_2]\big)^{n+1} \cdot u[x_2+1,|u|].
\end{equation}
\reviewTwo[inline]{The notation $u[x_1,x_2]$ is not defined.\olivier[inline]{fixed}}
Similarly, we can replicate $n$ times the intercepted factors $\tr{\ell}$ of $\rho$, 
for all anchor points $\ell$ of $L$. In this way we obtain a successful run on $\pump_L^{n+1}(u)$ 
that is of the form
\begin{equation}\label{eq:pumped-run}
  \pump_L^{n+1}(\rho) ~=~ \rho_0 ~ \tr{\ell_1}^n ~ \rho_1 ~ \dots ~ \rho_{k-1} ~ \tr{\ell_k}^n ~ \rho_k
\end{equation}
where $\ell_1\leqtime\dots\leqtime\ell_k$ are all the anchor points in $L$ 
(listed according to the run order $\leqtime$), $\rho_0$ is the prefix of $\rho$ 
ending at $\ell_1$,
$\rho_k$ is the suffix of $\rho$ 
starting at $\ell_k$, and for all $i=1,\dots,k-1$, 
$\rho_i$ is the factor of $\rho$ between 
$\ell_i$ and $\ell_{i+1}$.
Note that $\pump_L^1(\rho)$ coincides with the original run $\rho$. As a matter of fact,
one could define in a similar way the run $\pump_L^0(\rho)$ obtained from removing the loop
$L$ from $\rho$. However, we do not need this, and we will always parametrize the operation 
$\pump_L$ by a positive number $n+1$.

An example of a pumped run $\pump_{L_1}^3(\rho)$ is given in Figure \ref{fig:pumping-sweeping}, 
together with the indication of the anchor points $\ell_i$ and the intercepted factors $\tr{\ell_i}$.

\begin{figure}[!t]
\centering
\begin{tikzpicture}[baseline=0, inner sep=0, outer sep=0, minimum size=0pt, xscale=0.4, yscale=0.75]
  \tikzstyle{dot} = [draw, circle, fill=white, minimum size=4pt]
  \tikzstyle{fulldot} = [draw, circle, fill=black, minimum size=4pt]
  \tikzstyle{factor} = [->, shorten >=1pt, dotted, rounded corners=10]
  \tikzstyle{fullfactor} = [->, >=stealth, shorten >=1pt, very thick, rounded corners=10]

\begin{scope}
  \fill [pattern=north east lines, pattern color=gray!25] (3,-1) rectangle (9,3);
  \draw [dashed, thin, gray] (3,-1) -- (3,3);
  \draw [dashed, thin, gray] (9,-1) -- (9,3);
  \draw [gray] (3,-1.25) -- (3,-1.5) -- (9,-1.5) -- (9,-1.25);
  \draw [gray] (6,-1.75) node [below] {\footnotesize $L$};

  \draw (0,0) node (node0) {};
  \draw (3,0) node [dot, minimum size=13pt] (node1) {$\ell_1$};
  \draw (9,0) node (node2) {};
  \draw (11,0) node (node3) {};

  \draw (11,1) node (node4) {};
  \draw (9,1) node [dot, minimum size=13pt] (node5) {$\ell_2$};
  \draw (3,1) node (node6) {};
  \draw (1,1) node (node7) {};

  \draw (1,2) node (node8) {};
  \draw (3,2) node [dot, minimum size=13pt] (node9) {$\ell_3$};
  \draw (9,2) node (node10) {};
  \draw (12,2) node (node11) {};

  \draw [factor] (node0) -- (node1);
  \draw [fullfactor] (node1) edge 
        node [midway, minimum size=10pt, rectangle, draw, fill=white, rounded corners=0, inner sep=1pt] 
             {\small $~\tr{\ell_1}~$} (node2);
  \draw [factor] (node2) -- (node3.center) -- (node4.center) -- (node5); 
  \draw [fullfactor] (node5) edge 
        node [midway, minimum size=10pt, rectangle, draw, fill=white, rounded corners=0, inner sep=1pt] 
             {\small $~\tr{\ell_2}~$} (node6);
  \draw [factor] (node6) -- (node7.center) -- (node8.center) -- (node9); 
  \draw [fullfactor] (node9) edge 
        node [midway, minimum size=10pt, rectangle, draw, fill=white, rounded corners=0, inner sep=1pt] 
             {\small $~\tr{\ell_3}~$} (node10);
  \draw [factor] (node10) -- (node11); 
 
\end{scope}

\begin{scope}[xshift=14cm]
  \fill [pattern=north east lines, pattern color=gray!25] (3,-1) rectangle (21,3);
  \draw [dashed, thin, gray] (3,-1) -- (3,3);
  \draw [dashed, thin, gray] (9,-1) -- (9,3);
  \draw [dashed, thin, gray] (15,-1) -- (15,3);
  \draw [dashed, thin, gray] (21,-1) -- (21,3);
  \draw [gray] (3,-1.25) -- (3,-1.5) -- (9,-1.5) -- (9,-1.25);
  \draw [gray] (9,-1.25) -- (9,-1.5) -- (15,-1.5) -- (15,-1.25);
  \draw [gray] (15,-1.25) -- (15,-1.5) -- (21,-1.5) -- (21,-1.25);
  \draw [gray] (6,-1.75) node [below] {\footnotesize 1st copy of $L$};
  \draw [gray] (12,-1.75) node [below] {\footnotesize 2nd copy of $L$};
  \draw [gray] (18,-1.75) node [below] {\footnotesize 3rd copy of $L$};

  \draw (0,0) node (node0) {};
  \draw (3,0) node [dot, minimum size=13pt] (node1) {$\ell_1$};
  \draw (9,0) node (node2) {};
  \draw (15,0) node (node3) {};
  \draw (21,0) node (node4) {};
  \draw (23,0) node (node5) {};

  \draw (23,1) node (node6) {};
  \draw (21,1) node [dot, minimum size=13pt] (node7) {$\ell_2$};
  \draw (15,1) node (node8) {};
  \draw (9,1) node (node9) {};
  \draw (3,1) node (node10) {};
  \draw (1,1) node (node11) {};

  \draw (1,2) node (node12) {};
  \draw (3,2) node [dot, minimum size=13pt] (node13) {$\ell_3$};
  \draw (9,2) node (node14) {};
  \draw (15,2) node (node15) {};
  \draw (21,2) node (node16) {};
  \draw (24,2) node (node17) {};

  \draw [factor] (node0) -- (node1);
  \draw [fullfactor] (node1) edge 
        node [midway, minimum size=10pt, rectangle, draw, fill=white, rounded corners=0, inner sep=1pt] 
             {\small $~\tr{\ell_1}~$} (node2);
  \draw [fullfactor] (node2) edge 
        node [midway, minimum size=10pt, rectangle, draw, fill=white, rounded corners=0, inner sep=1pt] 
             {\small $~\tr{\ell_1}~$} (node3);
  \draw [fullfactor] (node3) edge 
        node [midway, minimum size=10pt, rectangle, draw, fill=white, rounded corners=0, inner sep=1pt] 
             {\small $~\tr{\ell_1}~$} (node4);
  \draw [factor] (node4) -- (node5.center) -- (node6.center) -- (node7); 
  \draw [fullfactor] (node7) edge 
        node [midway, minimum size=10pt, rectangle, draw, fill=white, rounded corners=0, inner sep=1pt] 
             {\small $~\tr{\ell_2}~$} (node8);
  \draw [fullfactor] (node8) edge 
        node [midway, minimum size=10pt, rectangle, draw, fill=white, rounded corners=0, inner sep=1pt] 
             {\small $~\tr{\ell_2}~$} (node9);
  \draw [fullfactor] (node9) edge 
        node [midway, minimum size=10pt, rectangle, draw, fill=white, rounded corners=0, inner sep=1pt] 
             {\small $~\tr{\ell_2}~$} (node10);
  \draw [factor] (node10) -- (node11.center) -- (node12.center) -- (node13); 
  \draw [fullfactor] (node13) edge 
        node [midway, minimum size=10pt, rectangle, draw, fill=white, rounded corners=0, inner sep=1pt] 
             {\small $~\tr{\ell_3}~$} (node14);
  \draw [fullfactor] (node14) edge 
        node [midway, minimum size=10pt, rectangle, draw, fill=white, rounded corners=0, inner sep=1pt] 
             {\small $~\tr{\ell_3}~$} (node15);
  \draw [fullfactor] (node15) edge 
        node [midway, minimum size=10pt, rectangle, draw, fill=white, rounded corners=0, inner sep=1pt] 
             {\small $~\tr{\ell_3}~$} (node16);
  \draw [factor] (node16) -- (node17); 
 
\end{scope}

\end{tikzpicture}
\caption{A loop $L$ with $3$ anchor points, and the result of pumping.}\label{fig:pumping-sweeping}
\end{figure}
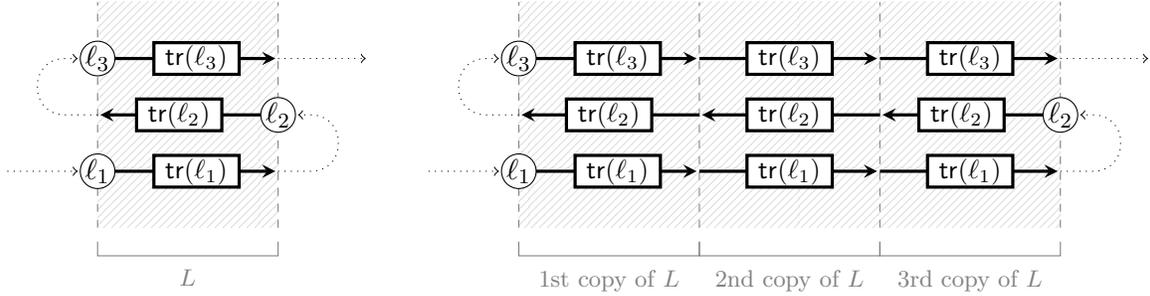

\medskip
\subsection*{Output minimality.}
We are interested into factors of the run $\rho$ that lie on a single level
and that contribute to the final output, but in a minimal way, in the sense 
that is formalized by the following definition:

\begin{defi}\label{def:output-minimal-sweeping}
Consider a factor $\a=\rho[\ell,\ell']$ of $\rho$.
We say that $\a$ is \emph{output-minimal} if 
 $\ell=(x,y)$ and $\ell'=(x',y)$, and all loops $L \subsetneq
 [x,x']$ 
produce empty output at level $y$.
\end{defi}

From now on, we set the constant $\bound = \cmax |Q|^\hmax +1$,
where $\cmax$ is the capacity of the transducer, that is, the 
maximal length of an output produced on a single transition (recall
that  $|Q|^\hmax$ is the maximal number of crossing sequences). 
As shown below, $\bound$ bounds the length of 
the output produced by an output-minimal factor:

\begin{lem}\label{lem:output-minimal-sweeping}
For all output-minimal factors $\alpha$, 
$|\out{\alpha}| \le \bound$.
\end{lem}

\begin{proof}
Suppose by contradiction that $|\out{\alpha}| > \bound$, with
$\a=\rho[\ell,\ell']$, $\ell=(x,y)$ and $\ell=(x',y)$.

Let $X$ be the set of all positions $x''$, with $\min(x,x') < x'' < \max(x,x')$,
that are sources of transitions of $\alpha$ that produce non-empty output. 
Clearly, the total number of letters produced by the transitions that depart
from locations in $X\times\{y\}$ is strictly larger than $\bound-1$.
Moreover, since each transition emits at most $\cmax$ symbols, we have $|X| > \frac{\bound-1}{\cmax} = |Q|^\hmax$.
Now, recall that crossing sequences are sequences of states of length at most $\hmax$.
Since $|X|$ is larger than the number of crossing sequences, $X$ contains two positions
$x_1<x_2$ such that $\rho|x_1=\rho|x_2$. In particular, $L=[x_1,x_2]$
is a loop strictly between  $x,x'$
with non-empty output on level $y$. 
This shows that $\rho[\ell,\ell']$ is not output-minimal.
\end{proof}

\medskip
\subsection*{Inversions and periodicity.}
Next, we define the crucial notion of inversion. Intuitively, an inversion
in a run identifies a part of the run that is potentially difficult to
simulate in a one-way manner because the order of generating the
output is reversed w.r.t.~the input. Inversions arise naturally in transducers 
that reverse arbitrarily long portions of the input, as well as in transducers
that produce copies of arbitrarily long portions of the input. 

\label{page-def-inversion}

\begin{defi}\label{def:inversion-sweeping}
An \emph{inversion} of the run $\rho$ is a tuple $(L_1,\ell_1,L_2,\ell_2)$ such that
\begin{enumerate}
\item $L_1,L_2$ are loops of $\r$,
  \item $\ell_1=(x_1,y_1)$ and $\ell_2=(x_2,y_2)$ 
        are anchor points of $L_1$ and $L_2$, respectively,
  \item $\ell_1 \lesstime \ell_2$ and $x_1 > x_2$ 
        \par\noindent
        (namely, $\ell_2$ follows $\ell_1$ in the run, 
         but the position of $\ell_2$ precedes the position of $\ell_1$),
  \item for both $i=1$ and $i=2$, $\out{\tr{\ell_i}}\neq\emptystr$ and $\tr{\ell_i}$ is output-minimal.
\end{enumerate}
\end{defi}

\noindent
The left hand-side of Figure~\ref{fig:inversion-sweeping} gives an example of an inversion, 
assuming that the outputs $v_1=\tr{\ell_1}$ and $v_2=\tr{\ell_2}$ are non-empty
and the intercepted factors are output-minimal.

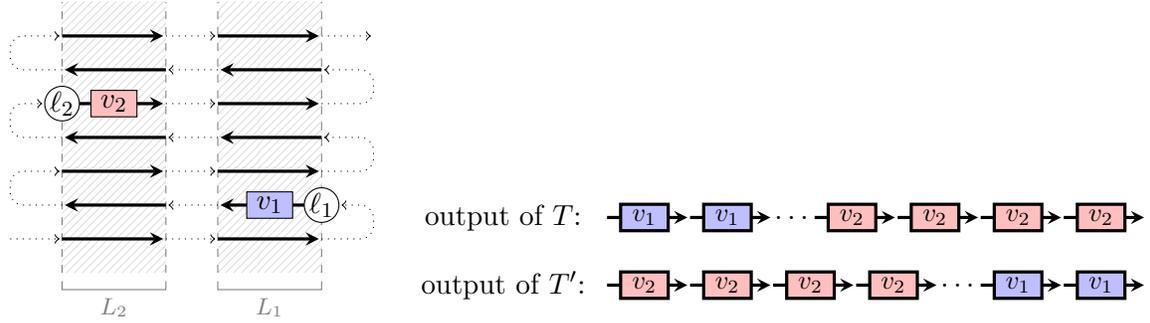
\begin{figure}[!t]
\centering
\begin{tikzpicture}[baseline=0, inner sep=0, outer sep=0, minimum size=0pt, xscale=0.345, yscale=0.45]
  \tikzstyle{dot} = [draw, circle, fill=white, minimum size=4pt]
  \tikzstyle{fulldot} = [draw, circle, fill=black, minimum size=4pt]
  \tikzstyle{factor} = [->, shorten >=1pt, dotted, rounded corners=5]
  \tikzstyle{fullfactor} = [->, >=stealth, shorten >=1pt, very thick, rounded corners=5]

\begin{scope}
  \fill [pattern=north east lines, pattern color=gray!25]
        (4,-1) rectangle (8,7);
  \fill [pattern=north east lines, pattern color=gray!25]
        (10,-1) rectangle (14,7);
  \draw [dashed, thin, gray] (4,-1) -- (4,7);
  \draw [dashed, thin, gray] (8,-1) -- (8,7);
  \draw [dashed, thin, gray] (10,-1) -- (10,7);
  \draw [dashed, thin, gray] (14,-1) -- (14,7);
  \draw [gray] (4,-1.25) -- (4,-1.5) -- (8,-1.5) -- (8,-1.25);
  \draw [gray] (10,-1.25) -- (10,-1.5) -- (14,-1.5) -- (14,-1.25);
  \draw [gray] (6,-1.75) node [below] {\footnotesize $L_2$};
  \draw [gray] (12,-1.75) node [below] {\footnotesize $L_1$};

  \draw (2,0) node (node0) {};
  \draw (4,0) node (node1) {};
  \draw (8,0) node (node2) {};
  \draw (10,0) node (node3) {};
  \draw (14,0) node (node4) {};
  \draw (16,0) node (node5) {};
  \draw (16,1) node (node6) {};
  \draw (14,1) node [dot, minimum size=13pt] (node7) {$\ell_1$};
  \draw (10,1) node (node8) {};
  \draw (8,1) node (node9) {};
  \draw (4,1) node (node10) {};
  \draw (2,1) node (node11) {};
  \draw (2,2) node (node12) {};
  \draw (4,2) node (node13) {};

  \draw (8,2) node (node14) {};
  \draw (10,2) node (node15) {};
  \draw (14,2) node (node16) {};
  \draw (16,2) node (node17) {};
  \draw (16,3) node (node18) {};
  \draw (14,3) node (node19) {};
  \draw (10,3) node (node20) {};
  \draw (8,3) node (node21) {};
  \draw (4,3) node (node22) {};
  \draw (2,3) node (node23) {};
  \draw (2,4) node (node24) {};
  \draw (4,4) node [dot, minimum size=13pt] (node25) {$\ell_2$};

  \draw (8,4) node  (node26) {};
  \draw (10,4) node (node27) {};
  \draw (14,4) node (node28) {};
  \draw (16,4) node (node29) {};
  \draw (16,5) node (node30) {};
  \draw (14,5) node (node31) {};
  \draw (10,5) node (node32) {};
  \draw (8,5) node (node33) {};
  \draw (4,5) node (node34) {};
  \draw (2,5) node (node35) {};
  \draw (2,6) node (node36) {};
  \draw (4,6) node (node37) {};
  
  \draw (8,6) node (node38) {};
  \draw (10,6) node (node39) {};
  \draw (14,6) node (node40) {};
  \draw (16,6) node (node41) {};

  \draw [factor] (node0) -- (node1);
  \draw [fullfactor] (node1) -- (node2);
  \draw [factor] (node2) -- (node3);
  \draw [fullfactor] (node3) -- (node4);
  \draw [factor] (node4) -- (node5.center) -- (node6.center) -- (node7); 
  \draw [fullfactor] (node7) -- (node8);
  \draw [factor] (node8) -- (node9);
  \draw [fullfactor] (node9) -- (node10);
  \draw [factor] (node10) -- (node11.center) -- (node12.center) -- (node13); 

  \draw [fullfactor] (node13) -- (node14);
  \draw [factor] (node14) -- (node15);
  \draw [fullfactor] (node15) -- (node16);
  \draw [factor] (node16) -- (node17.center) -- (node18.center) -- (node19); 
  \draw [fullfactor] (node19) -- (node20);
  \draw [factor] (node20) -- (node21);
  \draw [fullfactor] (node21) -- (node22);
  \draw [factor] (node22) -- (node23.center) -- (node24.center) -- (node25); 

  \draw [fullfactor] (node25) -- (node26);
  \draw [factor] (node26) -- (node27);
  \draw [fullfactor] (node27) -- (node28);
  \draw [factor] (node28) -- (node29.center) -- (node30.center) -- (node31); 
  \draw [fullfactor] (node31) -- (node32);
  \draw [factor] (node32) -- (node33);
  \draw [fullfactor] (node33) -- (node34);
  \draw [factor] (node34) -- (node35.center) -- (node36.center) -- (node37); 

  \draw [fullfactor] (node37) -- (node38);
  \draw [factor] (node38) -- (node39);
  \draw [fullfactor] (node39) -- (node40);
  \draw [factor] (node40) -- (node41);

   \draw (6,4) node [dot, minimum size=10pt, rectangle, fill=red!25] {$~v_2~$};
    \draw (12,1) node [dot, minimum size=10pt, rectangle, fill=blue!25] {$~v_1~$};
\end{scope}
\begin{scope}[xshift=25cm,yshift=18,scale=0.4]
  \draw (-10,0) node {output of $\cT$:};
  
  \draw (0,0) edge [fullfactor] 
        node [pos=0.45, minimum size=10pt, rectangle, draw, fill=white, rounded corners=0, inner sep=1pt, fill=blue!25] 
             {\small $~v_1~$} (8,0);
  \draw (8,0) edge [fullfactor] 
        node [pos=0.45, minimum size=10pt, rectangle, draw, fill=white, rounded corners=0, inner sep=1pt, fill=blue!25] 
             {\small $~v_1~$} (16,0);
  \draw (18,0) node {$\dots$};
  \draw (20,0) edge [fullfactor] 
        node [pos=0.45, minimum size=10pt, rectangle, draw, fill=white, rounded corners=0, inner sep=1pt, fill=red!25] 
             {\small $~v_2~$} (28,0);
  \draw (28,0) edge [fullfactor] 
        node [pos=0.45, minimum size=10pt, rectangle, draw, fill=white, rounded corners=0, inner sep=1pt, fill=red!25] 
             {\small $~v_2~$} (36,0);
  \draw (36,0) edge [fullfactor] 
        node [pos=0.45, minimum size=10pt, rectangle, draw, fill=white, rounded corners=0, inner sep=1pt, fill=red!25] 
             {\small $~v_2~$} (44,0);
  \draw (44,0) edge [fullfactor] 
        node [pos=0.45, minimum size=10pt, rectangle, draw, fill=white, rounded corners=0, inner sep=1pt, fill=red!25] 
             {\small $~v_2~$} (52,0);

  \draw (-10,-5) node {output of $\cT'$:};
  
  \draw (0,-5) edge [fullfactor] 
        node [pos=0.45, minimum size=10pt, rectangle, draw, fill=white, rounded corners=0, inner sep=1pt, fill=red!25] 
             {\small $~v_2~$} (8,-5);
  \draw (8,-5) edge [fullfactor] 
        node [pos=0.45, minimum size=10pt, rectangle, draw, fill=white, rounded corners=0, inner sep=1pt, fill=red!25] 
             {\small $~v_2~$} (16,-5);
  \draw (16,-5) edge [fullfactor] 
        node [pos=0.45, minimum size=10pt, rectangle, draw, fill=white, rounded corners=0, inner sep=1pt, fill=red!25] 
             {\small $~v_2~$} (24,-5);
  \draw (24,-5) edge [fullfactor] 
        node [pos=0.45, minimum size=10pt, rectangle, draw, fill=white, rounded corners=0, inner sep=1pt, fill=red!25] 
             {\small $~v_2~$} (32,-5);
  \draw (34,-5) node {$\dots$};
  \draw (36,-5) edge [fullfactor] 
        node [pos=0.45, minimum size=10pt, rectangle, draw, fill=white, rounded corners=0, inner sep=1pt, fill=blue!25] 
             {\small $~v_1~$} (44,-5);
  \draw (44,-5) edge [fullfactor] 
        node [pos=0.45, minimum size=10pt, rectangle, draw, fill=white, rounded corners=0, inner sep=1pt, fill=blue!25] 
             {\small $~v_1~$} (52,-5);
\end{scope}

\end{tikzpicture}
\caption{An inversion and the effect of pumping in an equivalent one-way transducer $\cT'$.}
\label{fig:inversion-sweeping}
\end{figure}

The rest of the section is devoted to prove the implication \PR1 $\Rightarrow$ \PR2
of Theorem \ref{thm:main2}.
We recall that a word $w=a_1 \cdots a_n$ has \emph{period} $p$ if for every $1\le i\le |w|-p$, 
we have $a_i = a_{i+p}$. For example, the word $abc \, abc \, ab$ has period $3$.

We remark that, thanks to Lemma \ref{lem:output-minimal-sweeping}, 
for every inversion $(L_1,\ell_1,L_2,\ell_2)$, the outputs 
$\out{\tr{\ell_1}}$ and $\out{\tr{\ell_2}}$ have length at most $\bound$.
By pairing this with the assumption that the transducer $\cT$ is one-way definable,
and by using some classical word combinatorics, we show that the output 
produced between the anchor points of every inversion has period that divides
the lengths of $\out{\tr{\ell_1}}$ and $\out{\tr{\ell_2}}$. In particular, this
period is at most $\bound$.
The proposition below shows a slightly stronger periodicity property, 
which refers to the output produced between the anchor points $\ell_1,\ell_2$
of an inversion, but extended on both sides with the words $\out{\tr{\ell_1}}$ and $\out{\tr{\ell_2}}$.
We will exploit this stronger periodicity property later, when dealing with 
overlapping portions of the run delimited by different inversions 
(cf.~Lemma \ref{lem:overlapping}).

\begin{prop}\label{prop:periodicity-sweeping}
If $\cT$ is one-way definable, then the following property \PR2 holds:
\begin{quote}
  For all inversions
$(L_1,\ell_1,L_2,\ell_2)$ of $\rho$, the period $p$ of the word 
\[
  \out{\tr{\ell_1}} ~ \out{\rho[\ell_1,\ell_2]} ~ \out{\tr{\ell_2}}
\]
divides both $|\out{\tr{\ell_1}}|$ and
$|\out{\tr{\ell_2}}|$. Moreover, $p \le \bound$.
\end{quote}
\end{prop}

The above proposition 
thus formalizes the implication \PR1 $\Rightarrow$ \PR2 of
Theorem \ref{thm:main2}.
Its proof relies on a few combinatorial results.
The first one is Fine and Wilf's theorem~\cite{Lothaire97}.
In short, this theorem says that, whenever two periodic
words $w_1,w_2$ share a sufficiently long factor, then they 
have the same period.
Here we use a slightly stronger variant of Fine and Wilf's theorem,
which additionally shows how to align a common factor of the two words $w_1,w_2$ 
so as to form a third word containing a prefix of $w_1$ and a suffix of $w_2$. 
This variant of Fine-Wilf's theorem will be particularly useful in the proof of Lemma~\ref{lem:overlapping}, while for all other applications the classical
statement suffices.
\olivier{This last sentence is outdated, right?}%
\gabriele{This refers to the variant of Fine-Wilf, which still used in the referenced lemma.
          But I rephrased so as to make it clear. Let me know if it makes more sense now.}%

\begin{thm}[Fine-Wilf's theorem]\label{thm:fine-wilf}
If $w_1 = w'_1\,w\:w''_1$ has period $p_1$, 
$w_2 = w'_2\,w\,w''_2$ has  period $p_2$, and 
the common factor $w$ has length at least $p_1+p_2-\gcd(p_1,p_2)$, 
then $w_1$, $w_2$, and $w_3 = w'_1\,w\,w''_2$ have period $\gcd(p_1,p_2)$.
\end{thm}

The second combinatorial result required in our proof concerns periods of words
with iterated factors, like those that arise from considering outputs of pumped runs,
and it is formalized precisely by the lemma below.
To improve readability, we often highlight the 
important iterations of factors inside a word.

%

\begin{lem}\label{lem:periods}
Assume that $v_0 \: \pmb{v_1^n} \: v_2 \: \cdots \: v_{k-1} \: \pmb{v_k^n} \: v_{k+1}$ 
has period $p$ for some $n >p$. 
Then  $v_0 \: \pmb{{v_1}^{n_1}} \: v_2 \: \cdots \: v_{k-1} \: \pmb{{v_k}^{n_k}} \: v_{k+1}$  
has period $p$ for all $n_1,\ldots,n_k \in \Nat$.
\end{lem}

\begin{proof}
Assume that 
$w=v_0 \: \pmb{v_1^n} \: v_2 \: \cdots \: v_{k-1} \: \pmb{v_k^n} \: v_{k+1}$ 
has period $p$, and that $n >p$. 
Consider an arbitrary factor $v_i^p$ of $w$. Since $v_i^p$ has periods $p$ and $|v_i|$, 
it has also period $r=\gcd(p,|v_i|)$.
By Fine-Wilf (Theorem \ref{thm:fine-wilf}), 
we know that $w$ has period $r$ as well.
Moreover, since the length of $v_i$ is multiple of $r$,
changing the number of repetitions of $v_i$ inside $w$ does
not affect the period $r$ of $w$. 
Since $v_i$ was chosen arbitrarily, this means that, for all $n_1,\dots,n_k\in\bbN$,
$v_0 \: \pmb{{v_1}^{n_1}} \: v_2 \: \cdots \: v_{k-1} \: \pmb{{v_k}^{n_k}} \: v_{k+1}$    
has period $r$, and hence period $p$ as well.
\end{proof}

Recall that our goal is to show that the output produced amid every 
inversion has period bounded by $\bound$.
The general idea is to pump the loops of the inversion and compare the outputs 
of the two-way transducer $\cT$ with those of an equivalent one-way 
transducer $\cT'$. 
The comparison leads to an equation between words with 
iterated factors, where the iterations are parametrized by two unknowns 
$n_1,n_2$ that occur in opposite order in the left, respectively right 
hand-side of the equation. 
Our third and last combinatorial result considers a word equation of 
this precise form, and derives from it a periodicity property.
\reviewOne[inline]{In Corollary 4.8, the definition of $v^{(n_1,n_2)}$ is hard to process, especially when introduced in the middle of a proof. The notion is used later in the part: it warrants a proper definition and/or clear example either in the preliminaries or at the beginning of Part 4.
  \felix[inline]{I don't understand this remark}
  \olivier[inline]{answered in review1-answers.txt: the best place is here.}
}
\gabriele[inline]{This reviewer commend has led to a number of corrections, and the removal of Saarela tool}
For the sake of brevity, we use the notation $v^{(n_1,n_2)}$
to represent words with factors iterated $n_1$ or $n_2$ times,
namely, words of the form 
$v_0 \: v_1^{n_{i_1}} \: v_2 \: \cdots \: v_{k-1} \: v_k^{n_{i_k}} \: v_{k+1}$,
where the $v_0,v_1,v_2,\dots,v_{k-1},v_k,v_{k+1}$ are fixed words (possibly empty)
and each index among $i_1,\dots,i_k$ is either $1$ or $2$.

\begin{lem}\label{lem:oneway-vs-twoway}
Consider a word equation of the form
\[
  v_0^{(n_1,n_2)} \: \pmb{v_1^{n_1}} \: v_2^{(n_1,n_2)} \: \pmb{v_3^{n_2}} \: v_4^{(n_1,n_2)}
  ~=~
  w_0 \: \pmb{w_1^{n_2}} \: w_2 \: \pmb{w_3^{n_1}} \: w_4
\]
where $n_1,n_2$ are the unknowns and $v_1,v_3$ are non-empty words.
If the above equation holds for all $n_1,n_2\in\bbN$, 
then 
\[
  \pmb{v_1} ~ \pmb{v_1^{n_1}} ~ v_2^{(n_1,n_2)} ~ \pmb{v_3^{n_2}} ~ \pmb{v_3}
\]
has period $\gcd(|v_1|,|v_3|)$ for all $n_1,n_2\in\bbN$.
\end{lem}

\begin{proof}
The idea of the proof is to let the parameters $n_1,n_2$ of the equation
grow independently, and apply Fine and Wilf's theorem (Theorem \ref{thm:fine-wilf}) 
a certain number of times to establish periodicities in overlapping factors of the 
considered words.

We begin by fixing $n_1$ large enough so that the factor
$\pmb{v_1^{n_1}}$ of the left hand-side of the equation becomes 
longer than $|w_0|+|w_1|$ (this is possible because $v_1$ is non-empty).
Now, if we let $n_2$ grow arbitrarily large, we see that the length of 
the periodic word $\pmb{w_1^{n_2}}$ is almost equal to the length of 
the left hand-side term
$v_0^{(n_1,n_2)} ~ \pmb{v_1^{n_1}} ~ v_2^{(n_1,n_2)} ~ \pmb{v_3^{n_2}} ~ v_4^{(n_1,n_2)}$:
indeed, the difference in length is given by the constant 
$|w_0| + |w_2| + n_1\cdot |w_3| + |w_4|$.
In particular, this implies that $\pmb{w_1^{n_2}}$ 
covers arbitrarily long prefixes of 
$\pmb{v_1} ~ v_2^{(n_1,n_2)} ~ \pmb{v_3^{n_2+1}}$,
which in its turn contains long repetitions of the word $v_3$.
Hence, by Theorem \ref{thm:fine-wilf},
the word 
$\pmb{v_1} ~ v_2^{(n_1,n_2)} ~ \pmb{v_3^{n_2+1}}$ 
has period $|v_3|$.

We remark that the periodicity shown so far holds for a large enough 
$n_1$ and for all but finitely many $n_2$, where the threshold for
$n_2$ depends on $n_1$: once $n_1$ is fixed, $n_2$ needs to be larger
than $f(n_1)$, for a suitable function $f$.  
In fact, by using 
\gabriele{Saarela has been replaced by the new lemma here}
Lemma \ref{lem:periods},
with $n_1$ fixed and $n=n_2$ large enough, 
we deduce that the periodicity holds for large enough $n_1\in\bbN$ and 
for all $n_2\in\bbN$.  

We could also apply a symmetric reasoning: we choose 
$n_2$ large enough and let $n_1$ grow arbitrarily large. Doing so, we 
prove that for a large enough $n_2$ and for all but finitely many $n_1$, 
the word $\pmb{v_1^{n_1+1}} ~ v_2^{(n_1,n_2)} ~ \pmb{v_3}$ is periodic
with period $|v_1|$. As before, with the help of 
\gabriele{same here for Saarela...}
Lemma \ref{lem:periods},
this can be strengthened to hold for large enough $n_2\in\bbN$ and for all $n_1\in\bbN$.

Putting together the results proven so far, we get that for all but finitely many $n_1,n_2$,
\[
  \rightward{ \underbracket[0.5pt]{ \phantom{ \pmb{v_1^{n_1}} \cdot \pmb{v_1} ~\cdot~
                                                   v_2^{(n_1,n_2)} ~\cdot~ \pmb{v_3} } }%
                                 _{\text{period } |v_1|} }
  \pmb{v_1^{n_1}} \cdot
  \overbracket[0.5pt]{ \pmb{v_1} ~\cdot~ v_2^{(n_1,n_2)} ~\cdot~ 
                       \pmb{v_3} \cdot \pmb{v_3^{n_2}} }%
                    ^{\text{period } |v_3|}
  .
\]
Finally, we observe that the prefix 
$\pmb{v_1^{n_1+1}}\cdot v_2^{(n_1,n_2)}\cdot \pmb{v_3}$ 
and the suffix
$\pmb{v_1}\cdot v_2^{(n_1,n_2)}\cdot\pmb{v_3^{n_2+1}}$ 
share a common factor of length at least $|v_1|+|v_3|$. 
By Theorem~\ref{thm:fine-wilf}, 
we derive that $\pmb{v_1^{n_1+1}}\cdot v_2^{(n_1,n_2)}\cdot
\pmb{v_3^{n_2+1}}$ has period $\gcd(|v_1|,|v_3|)$ for all but finitely
many $n_1,n_2$.
Finally, by using again 
\gabriele{...and here}
Lemma \ref{lem:periods},
we conclude that the periodicity holds for all $n_1,n_2\in\bbN$.
\end{proof}

\medskip
We are now ready to prove the implication \PR1 $\Rightarrow$ \PR2:

\begin{proof}[Proof of Proposition~\ref{prop:periodicity-sweeping}]
Let $\cT'$ be a one-way transducer equivalent to $\cT$, and consider 
an inversion $(L_1,\ell_1,L_2,\ell_2)$ of the successful run $\rho$ of $\cT$
on input $u$. 
The reader may refer to Figure \ref{fig:inversion-sweeping} 
to get basic intuition about the proof technique.
For simplicity, we assume that the loops $L_1$ and $L_2$ are disjoint,
as shown in the figure. If this were not the case, we would have 
at least $\max(L_1) > \min(L_2)$, since the anchor point $\ell_1$ 
is strictly to the right of the anchor point $\ell_2$.
We could then consider the pumped run $\pump_{L_1}^k(\rho)$ for a
large enough $k>1$ in such a way that the rightmost copy of $L_1$ 
turns out to be disjoint from and strictly to the right of $L_2$. 
We could thus reason as we do below, by replacing everywhere 
(except in the final part of the proof, cf.~{\em Transferring periodicity to the original run})
the run $\rho$ with the pumped run $\pump_{L_1}^k(\rho)$, and the formal 
parameter $m_1$ with $m_1+k$. 

\smallskip\noindent
{\em Inducing loops in $T'$.}~
We begin by pumping the run $\rho$ and the underlying input $u$, 
on the loops $L_1$ and $L_2$, in order to induce new loops $L'_1$ 
and $L'_2$ that are also loops in a successful run of $\cT'$.
Assuming that $L_1$ is strictly to the right of $L_2$, we define for all numbers $m_1,m_2\in\bbN$:
\[
\begin{array}{rcl}
  u^{(m_1,m_2)}    &=& \pump_{L_1}^{m_1+1}(\pump_{L_2}^{m_2+1}(u))  \\[1ex]
  \rho^{(m_1,m_2)} &=& \pump_{L_1}^{m_1+1}(\pump_{L_2}^{m_2+1}(\rho)).
\end{array}
\]
In the pumped run $\rho^{(m_1,m_2)}$, we identify the positions that mark the 
endpoints of the occurrences of $L_1,L_2$. More precisely, if $L_1=[x_1,x_2]$ 
and $L_2=[x_3,x_4]$, with $x_1>x_4$, then the sets of these positions are 
\[
\begin{array}{rcl}
   X_2^{(m_1,m_2)} &=& \big\{ x_3 + i(x_4-x_3) ~:~ 0\le i\le m_2+1 \big\} \\[1ex]
   X_1^{(m_1,m_2)} &=& \big\{ x_1 + j(x_2-x_1) + m_2(x_4-x_3) ~:~ 0\le j\le m_1+1 \big\}. 
\end{array}
\]

\smallskip\noindent
{\em Periodicity of outputs of pumped runs.}~
We use now the fact that $\cT'$ is a one-way transducer equivalent to $\cT$.
\reviewTwo[inline]{you fix $\lambda^{m_1,m_2}$ to be a particular run of T’ on the input $u^{m_1,m_2}$, and then, on line 10, you state that for some particular $m_1$ and $m_2$, this run can be obtained by pumping $\lambda^{k_0,k_0}$. This does not seem true in general: since $T$ is not deterministic, we might just have chosen a different run.}%
\gabriele[inline]{Corrected}%
We first recall (see, for instance, \cite{eilenberg1974automata,berstel2013transductions})
that every functional one-way transducer can be made unambiguous, namely,
can be transformed into an equivalent one-way transducer that admits at 
most one successful run on each input.
This means that, without loss of generality, we can assume that $\cT'$ too
is unambiguous, and hence it admits exactly one successful run, say 
$\lambda^{(m_1,m_2)}$, on each input $u^{(m_1,m_2)}$.

Since $\cT'$ has finitely many states, we can find, for a large enough number $k_0$ 
 two positions $x'_1<x'_2$, both in $X_1^{(k_0,k_0)}$, such that $L'_1=[x'_1,x'_2]$
is a loop of $\lambda^{(k_0,k_0)}$. Similarly, we can find two positions 
$x'_3<x'_4$, both in $X_2^{(k_0,k_0)}$, such that $L'_2=[x'_3,x'_4]$ is a 
loop of $\lambda^{(k_0,k_0)}$.
By construction $L'_1$ (resp.~$L'_2$) consists of $k_1\leq k_0$
(resp.~$k_2\leq k_0$) copies of $L_1$ (resp.~$L_2$), and 
hence $L'_1,L'_2$ are also loops of $\rho^{(k_0,k_0)}$.
In particular, this implies that for all $n_1,n_2\in\bbN$:
\[
\begin{array}{rcl}
  \pump_{L'_1}^{n_1+1}(\pump_{L'_2}^{n_2+1}(u^{(k_0,k_0)})) 
  &=& u^{(f(n_1),g(n_2))} \\[1ex]
  \pump_{L'_1}^{n_1+1}(\pump_{L'_2}^{n_2+1}(\rho^{(k_0,k_0)})) 
  &=& \rho^{(f(n_1),g(n_2))} \\[1ex]
  \pump_{L'_1}^{n_1+1}(\pump_{L'_2}^{n_2+1}(\lambda^{(k_0,k_0)})) 
  &=& \lambda^{(f(n_1),g(n_2))}.
\end{array}
\]
where $f(n_1)=k_1 n_1+k_0$ and $g(n_2)=k_2 n_2+k_0$. 

Now recall that $\rho^{(f(n_1),g(n_2))}$ and $\lambda^{(f(n_1),g(n_2))}$ are 
runs of $\cT$ and $\cT'$ on the same word $u^{(f(n_1),g(n_2))}$, and they 
produce the same output. Let us denote this output by $w^{(f(n_1),g(n_2))}$. 
Below, we show two possible factorizations of $w^{(f(n_1),g(n_2))}$
based on the shapes of the pumped runs $\lambda^{(f(n_1),g(n_2))}$ 
and $\rho^{(f(n_1),g(n_2))}$.
For the first factorization, we recall that $L'_2$ precedes $L'_1$, 
according to the ordering of positions, and that the run 
$\lambda^{(f(n_1),g(n_2))}$ is one-way (in particular loops have only one anchor point and one trace). We thus obtain:
\reviewTwo[inline]{eq. 4.3 The exponents should be $n_2 + 1$ and $n_1 + 1$. (same for eq. 4.4)}
\felix[inline]{Yes but not sure it makes sense to keep the +1 everywhere then}
\gabriele[inline]{We agreed to change the two itemized lists below and replace `right border'
          by `left border' everywhere. I have implemented this...}
\begin{equation}\label{eq:one-way}
  w^{(f(n_1),g(n_2))} ~=~ w_0 ~ \pmb{w_1^{n_2}} ~ w_2 ~ \pmb{w_3^{n_1}} ~ w_4
\end{equation}
where
\begin{itemize}
 \item $w_0$ is the output produced by the prefix of $\lambda^{(k_0,k_0)}$ 
        ending at the only anchor point of $L'_2$,
  \item $\pmb{w_1}$ is the trace of $L'_2$, 
   \item $w_2$ is the output produced by the factor of $\lambda^{(k_0,k_0)}$ 
        between the anchor points of $L'_2$ and $L'_1$, 
  \item $\pmb{w_3}$ is the trace of $L'_1$,
  \item $w_4$ is the output produced by the suffix of
        $\lambda^{(k_0,k_0)}$ starting at the anchor point of $L'_1$. Hence, $w_0 ~ w_2 ~ w_4$ is the output of $\lambda^{(k_0,k_0)}$.
\end{itemize}
\felix{added the last sentence to avoid the $+1$ or not confusion. In addition to the anchor point/trace usage that we talked about}
For the second factorization, we consider $L'_1$ and $L'_2$ as loops of $\rho^{(k_0,k_0)}$.
We recall that $\ell_1,\ell_2$ are anchor points of the loops $L_1,L_2$ of $\rho$, and that
there are corresponding copies of these anchor points in the pumped
run $\rho^{(f(n_1),g(n_2))}$. 
We define $\ell'_1$ (resp.~$\ell'_2$) to be the first (resp.~last)
location in $\rho^{(f(n_1),g(n_2))}$ 
that corresponds to $\ell_1$ (resp.~$\ell_2$) and that is an anchor point of a copy of $L'_1$ (resp.~$L'_2$).
For example, if $\ell_1=(x_1,y_1)$, with $y_1$ even, then $\ell'_1=\big(x_1+f(n_2)(x_4-x_3),y_1\big)$.
Thanks to Equation \ref{eq:pumped-run} we know that the output produced by 
$\rho^{(f(n_1),g(n_2))}$ is of the form 
\begin{equation}\label{eq:two-way}
  w^{(f(n_1),g(n_2))} ~=~ 
  v_0^{(n_1,n_2)} ~ \pmb{v_1^{n_1}} ~ v_2^{(n_1,n_2)} ~ \pmb{v_3^{n_2}} ~ v_4^{(n_1,n_2)}
\end{equation}
where 
\begin{itemize}
  \item $\pmb{v_1}=\out{\tr{\ell'_1}}$, where $\ell'_1$ is seen as an anchor point in a copy of $L'_1$,
  \item $\pmb{v_3}=\out{\tr{\ell'_2}}$, where $\ell'_2$ is seen as an anchor point in a copy of $L'_2$ 
        \par\noindent
        (note that the words $v_1,v_3$ depend on $k_0$, but not on $n_1,n_2$), 
  \item $v_0^{(n_1,n_2)}$ is the output produced by the prefix of $\rho^{(f(n_1),g(n_2))}$ 
        that ends at $\ell_1'$ 
        (this word may depend on the parameters $n_1,n_2$ since the loops
        $L'_1,L'_2$ may be traversed several times before reaching the first occurrence of $\tr{\ell'_1}$),
  \item $v_2^{(n_1,n_2)}$ is the output produced by the factor of $\rho^{(f(n_1),g(n_2))}$ 
        that starts at $\ell'_1$ 
        and ends at $\ell'_2$, 
  \item $v_4^{(n_1,n_2)}$ is the output produced by the suffix of $\rho^{(f(n_1),g(n_2))}$ 
        that starts at $\ell'_2$. 
\end{itemize}

Putting together Equations~(\ref{eq:one-way}) and (\ref{eq:two-way}), we get
\begin{equation}\label{eq:one-way-vs-two-way}
  v_0^{(n_1,n_2)} ~ \pmb{v_1^{n_1}} ~ v_2^{(n_1,n_2)} ~ \pmb{v_3^{n_2}} ~ v_4^{(n_1,n_2)}
  ~~=~~ 
  w_0 ~ \pmb{w_1^{n_2}} ~ w_2 ~ \pmb{w_3^{n_1}} ~ w_4 .
\end{equation}
Recall that the  definition of inversions (Definition
\ref{def:inversion-sweeping}) states 
that the words $v_1,v_3$ are non-empty.
This allows us to apply Lemma \ref{lem:oneway-vs-twoway}, which shows that the word 
$\pmb{v_1} ~ \pmb{v_1^{n_1}} ~ v_2^{(n_1,n_2)} ~ \pmb{v_3^{n_2}} ~ \pmb{v_3}$ 
has period $p=\gcd(|v_1|,|v_3|)$, for all $n_1,n_2\in\bbN$.

Note that the latter period $p$ still depends on $\cT'$, since 
the words $v_1,v_3$ were obtained from loops 
$L'_1,L'_2$ on the run $\lambda^{(k_0,k_0)}$ of $\cT'$. 
However, because each loop $L'_i$ consists of $k_i$ copies of the original loop $L_i$, 
we also know that $v_1=(\out{\tr{\ell_1}})^{k_1}$ and $v_3=(\out{\tr{\ell_2}})^{k_2}$.
By Theorem \ref{thm:fine-wilf}, this implies that for all $n_1,n_2\in\bbN$, the word
\[
  \big(\out{\tr{\ell_1}}\big) ~
  \big(\out{\tr{\ell_1}}\big)^{k_1 n_1} ~ 
  v_2^{(n_1,n_2)} ~ 
  \big(\out{\tr{\ell_2}}\big)^{k_2 n_2} ~
  \big(\out{\tr{\ell_2}}\big)
\]
has a period that divides $|\out{\tr{\ell_1}}|$ and $|\out{\tr{\ell_2}}|$.

\smallskip\noindent
{\em Transferring periodicity to the original run.}~
The last part of the proof amounts at showing a similar periodicity property 
for the output produced by the original run $\rho$.
By construction, the iterated factors inside $v_2^{(n_1,n_2)}$ in the previous
word are all of the form $v^{k_1 n_1 + k_0}$ or $v^{k_2 n_2 + k_0}$, 
for some words $v$. By taking out the constant factors $v^{k_0}$ from the 
latter repetitions, we can write $v_2^{(n_1,n_2)}$ as a word with iterated 
factors of the form $v^{k_1 n_1}$ or $v^{k_2 n_2}$, namely, as ${v'_2}^{(k_1 n_1, k_2 n_2)}$.
So the word
\[
  \big(\out{\tr{\ell_1}}\big) ~
  \big(\out{\tr{\ell_1}}\big)^{k'_1} ~ 
  {v'_2}^{(k'_1, k'_2)} ~ 
  \big(\out{\tr{\ell_2}}\big)^{k'_2} ~
  \big(\out{\tr{\ell_2}}\big)
\]
is periodic, with period that divides $|\out{\tr{\ell_1}}|$ and $|\out{\tr{\ell_2}}|$,
for all $k'_1 \in \{k_1 n \::\: n\in\bbN\}$ and all $k'_2 \in \{k_2 n \::\: n\in\bbN\}$.
We now apply 
\gabriele{Also here Saarela has been replaced by the new lemma}
Lemma \ref{lem:periods}, 
once with $n=k'_1$ and once with $n=k'_2$,
to conclude that the latter periodicity property holds also for $k'_1=1$ and $k'_2=1$.
This shows that the word
\[
  \out{\tr{\ell_1}} ~ \out{\rho[\ell_1,\ell_2]} ~ \out{\tr{\ell_2}}
\]
is periodic, with period that divides $|\out{\tr{\ell_1}}|$ and $|\out{\tr{\ell_2}}|$.
\end{proof}


\section{One-way definability in the sweeping case}%
\label{sec:characterization-sweeping}

\reviewOne[inline]{From Part 5 onwards, a major point is that an output you can bound is an output you can guess, and a periodic output is an output you don't have to guess, as you can produce it "out of order". On your Figures (and any you might care to add during potential revisions), having a specific visual representation (if possible coherent throughout the paper) to differentiate path sections with empty, bounded, or periodic output will make your point way more intelligible.
}%
\felix[inline]{so this is just about uniformizing the figure for sweeping and for two way ?}%

In the previous section we have shown the implication \PR1 $\Rightarrow$ \PR2 for a functional sweeping 
transducer $\cT$. Here we close the cycle by proving the implications \PR2 $\Rightarrow$ \PR3 and \PR3 $\Rightarrow$ \PR1.
In particular, we show how to derive the existence of successful runs admitting a $\bound$-decomposition
and construct a one-way transducer $\cT'$ that simulates $\cT$ on those runs.
This will basically prove Theorem~\ref{thm:main2} in the sweeping case.

\medskip
\subsection*{Run decomposition.}
We begin by giving the definition of $\bound$-decomposition for a run $\rho$ of $\cT$.
Intuitively,  a $\bound$-decomposition of $\rho$ identifies factors of $\rho$ that can be easily simulated 
in a one-way manner. The definition below describes precisely the shape of these factors.

First we need to  recall the notion of almost periodicity: 
a word $w$ is \emph{almost periodic with bound $p$} if $w=w_0~w_1~w_2$ 
for some words $w_0,w_2$ of length at most $p$ and some word $w_1$ 
of period at most $p$. 

We need to work with subsequences of the run $\rho$ that are induced by 
particular sets of locations, not necessarily consecutive. 
Recall that $\rho[\ell,\ell']$ denotes the factor of $\rho$ 
delimited by two locations $\ell\leqtime\ell'$. Similarly, given any 
set $Z$ of locations, we denote by $\rho|Z$ the subsequence of $\rho$ 
induced by $Z$. Note that, even though $\rho|Z$ might not be a valid
run 
\label{rhoZ}
of the transducer, we can still refer to the number of transitions in it 
and to the size of the produced output $\out{\rho|Z}$. 
Formally, a transition in $\rho|Z$ is a transition from some $\ell$ to $\ell'$, 
where both $\ell,\ell'$ belong to $Z$. The output $\out{\rho|Z}$ is the 
concatenation of the outputs of the transitions of $\rho| Z$ 
(according to the order given by $\rho$).

\reviewOne[inline]{Fig. 5 is moderately informative... Figure 17 and 18 are a more inspiring illustration of the intuition behind diagonal and blocks. Their "restricted" counterpart could be simpler to represent.}%
\gabriele[inline]{done}%

\begin{defi}\label{def:factors-sweeping}
Consider a factor $\rho[\ell,\ell ']$ of a run $\rho$ of $\cT$, 
where $\ell=(x,y)$, $\ell'=(x',y')$ are two locations with $x \le x'$.
We call $\rho[\ell,\ell']$
\begin{itemize}
  \medskip
  \item \parbox[t]{\dimexpr\textwidth-\leftmargin}{%
        \vspace{-2.75mm}
        \begin{wrapfigure}{r}{7.5cm}
        \end{wrapfigure} 
        a \emph{floor} if $y=y'$ is even, namely, if $\rho[\ell,\ell']$ 
        lies entirely on the same level and is rightward oriented; 
        }\noparbreak
  \bigskip
  \bigskip
  \item \parbox[t]{\dimexpr\textwidth-\leftmargin}{%
        \vspace{-2.75mm}
        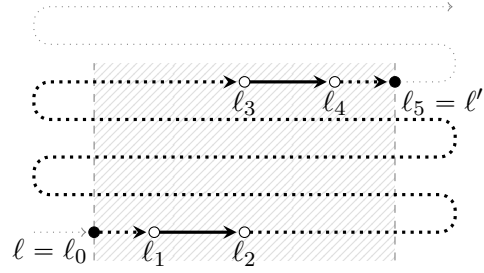
\begin{wrapfigure}{r}{7.5cm}
        \vspace{-22mm}
\centering
\begin{tikzpicture}[baseline=0, inner sep=0, outer sep=0, minimum size=0pt, xscale=0.4, yscale=0.5]
\begin{scope} 
  \tikzstyle{dot} = [draw, circle, fill=white, minimum size=4pt]
  \tikzstyle{fulldot} = [draw, circle, fill=black, minimum size=4pt]
  \tikzstyle{grayfactor} = [->, shorten >=1pt, rounded corners=6, gray, thin, dotted]
  \tikzstyle{factor} = [->, shorten >=1pt, rounded corners=6]
  \tikzstyle{dotfactor} = [->, shorten >=1pt, dotted, rounded corners=6]
  \tikzstyle{fullfactor} = [->, >=stealth, shorten >=1pt, very thick, rounded corners=6]
  \tikzstyle{dotfullfactor} = [->, >=stealth, shorten >=1pt, dotted, very thick, rounded corners=6]

  \fill [pattern=north east lines, pattern color=gray!25]
        (4,-0.75) rectangle (14,4.5);
  \draw [dashed, thin, gray] (4,-0.75) -- (4,4.5);
  \draw [dashed, thin, gray] (14,-0.75) -- (14,4.5);
 
  \draw (2,0) node (node0) {} ;
  \draw (4,0) node [fulldot] (node1) {};
  \draw (6,0) node [dot] (node2) {};
  \draw (9,0) node [dot] (node3) {};
  \draw (16,0) node (node4) {};
  \draw (16,1) node (node5) {};
  \draw (2,1) node (node6) {};
  \draw (2,2) node (node7) {};
  \draw (16,2) node (node8) {};
  \draw (16,3) node (node9) {};
  \draw (2,3) node (node10) {};
  \draw (2,4) node (node11) {};
  \draw (9,4) node [dot] (node12) {};
  \draw (12,4) node [dot] (node13) {};
  \draw (14,4) node [fulldot] (node14) {};
  \draw (16,4) node (node15) {};
  \draw (16,5) node (node16) {};
  \draw (2,5) node (node17) {};
  \draw (2,6) node (node18) {};
  \draw (16,6) node (node19) {};

  \draw [grayfactor] (node0) -- (node1) ;
  \draw [dotfullfactor] (node1) -- (node2);
  \draw [fullfactor] (node2) -- (node3);
  \draw [dotfullfactor] (node3) -- (node4.center) -- (node5.center) -- (node6.center) -- 
                        (node7.center) -- (node8.center) -- (node9.center) -- (node10.center) -- 
                        (node11.center) -- (node12);
  \draw [fullfactor] (node12) -- (node13);
  \draw [dotfullfactor] (node13) -- (node14);
  \draw [grayfactor] (node14) -- (node15.center) -- (node16.center) -- (node17.center) -- 
                     (node18.center) -- (node19);
  
  \draw (node1) node [below left = 1.2mm] {$\ell=\ell_0$};
  \draw (node2) node [below = 1.2mm] {$\ell_1$};
  \draw (node3) node [below = 1.2mm] {$\ell_2$};
  \draw (node12) node [below = 1.2mm] {$\ell_3$};
  \draw (node13) node [below = 1.2mm] {$\ell_4$};
  \draw (node14) node [below right = 1.2mm] {$\ell_5=\ell'$};
\end{scope}
\end{tikzpicture}
\vspace{-2mm}
\caption{A diagonal}\label{fig:diagonal-sweeping}

        \vspace{-4mm}
        \end{wrapfigure} 
        a \emph{$\bound$-diagonal} if 
        there is a sequence 
        $\ell = \ell_0 \leqtime \ell_1 \leqtime \dots \leqtime \ell_{2n+1} = \ell'$, 
        where each $\rho[\ell_{2i+1},\ell_{2i+2}]$ is a floor,
        each $\rho[\ell_{2i}, \ell_{2i+1}]$ produces an output 
        of length at most $2\hmax\bound$,
        and the position of each $\ell_{2i}$ 
        is to the left of the position of $\ell_{2i+1}$;
        }
  \bigskip
  \smallskip
  \item \parbox[t]{\dimexpr\textwidth-\leftmargin}{%
        \vspace{-2.75mm}
        \begin{wrapfigure}{r}{7.5cm}
        \vspace{-10mm}
\centering
\begin{tikzpicture}[baseline=0, inner sep=0, outer sep=0, minimum size=0pt, xscale=0.4, yscale=0.5]
\begin{scope} 
  \tikzstyle{dot} = [draw, circle, fill=white, minimum size=4pt]
  \tikzstyle{fulldot} = [draw, circle, fill=black, minimum size=4pt]
  \tikzstyle{grayfactor} = [->, shorten >=1pt, rounded corners=6, gray, thin, dotted]
  \tikzstyle{factor} = [->, shorten >=1pt, rounded corners=6]
  \tikzstyle{dotfactor} = [->, shorten >=1pt, dotted, rounded corners=6]
  \tikzstyle{fullfactor} = [->, >=stealth, shorten >=1pt, very thick, rounded corners=6]
  \tikzstyle{dotfullfactor} = [->, >=stealth, shorten >=1pt, dotted, very thick, rounded corners=6]

  \fill [pattern=north east lines, pattern color=gray!25]
        (6,-0.75) rectangle (12,2.5);
  \draw [dashed, thin, gray] (6,-0.75) -- (6,2.5);
  \draw [dashed, thin, gray] (12,-0.75) -- (12,2.5);
 
  \draw (2,0) node (node0) {} ;
  \draw (6,0) node [fulldot] (node1) {};
  \draw (12,0) node (node2) {};
  \draw (16,0) node (node3) {};
  \draw (16,1) node (node4) {};
  \draw (12,1) node (node5) {};
  \draw (6,1) node (node6) {};
  \draw (2,1) node (node7) {};
  \draw (2,2) node (node8) {};
  \draw (6,2) node (node9) {};
  \draw (12,2) node [fulldot] (node10) {};
  \draw (16,2) node (node11) {};
  \draw (16,3) node (node12) {};
  \draw (12,3) node (node13) {};
  \draw (6,3) node (node14) {};
  \draw (2,3) node (node15) {};
  \draw (2,4) node (node16) {};
  \draw (6,4) node (node17) {};
  \draw (12,4) node (node18) {};
  \draw (16,4) node (node19) {};

  \draw [grayfactor] (node0) -- (node1) ;
  \draw [fullfactor] (node1) -- (node2.center); 
  \draw [dotfullfactor] (node2.center) -- (node3.center) -- (node4.center) -- (node5.center); 
  \draw [fullfactor] (node5.center) -- (node6.center); 
  \draw [dotfullfactor] (node6.center) -- (node7.center) -- (node8.center) -- (node9.center); 
  \draw [fullfactor] (node9.center) -- (node10); 
  \draw [grayfactor] (node10) -- (node11.center) -- (node12.center) -- (node13.center) --
                     (node14.center) -- (node15.center) -- (node16.center) -- (node17.center) --
                     (node18.center) -- (node19.center);
  
  \draw (node1) node [below left = 1.2mm] {$\ell~$};
  \draw (node10) node [above right = 1.2mm] {$~\ell'$};
\end{scope}
\end{tikzpicture}
\vspace{-2mm}
\caption{A \rightward{block.}\phantom{diagonal.}}\label{fig:block-sweeping}

        \vspace{-5mm}
        \end{wrapfigure} 
        a \emph{$\bound$-block} if 
        the output produced by $\rho[\ell,\ell']$ is almost periodic with bound $2\bound$, 
        and the output produced by the subsequence $\rho|Z$, 
        where $Z=[\ell,\ell'] ~\setminus~ \big([x,x']\times[y,y']\big)$,
        has length at most $2\hmax\bound$.
        }
        \vspace{8mm}
\end{itemize}
\end{defi}

%
%
Before continuing we give some intuition on the above definitions. 
The simplest concept is that of floor, which is a rightward oriented factor of a run. 
Diagonals are sequences of consecutive floors interleaved by factors that 
produce small (bounded) outputs. We see an example of a diagonal 
in Figure \ref{fig:diagonal-sweeping}, where we marked the important
locations and highlighted with thick arrows the two floors that form 
the diagonal. The factors of the diagonal that are not floors are 
represented instead by dotted arrows. 
The third notion is that of a block.
An important constraint in the definition of a block 
is that the produce output must be almost periodic, with small enough bound.
In Figure \ref{fig:block-sweeping}, the periodic output is represented by the 
thick arrows, either solid or dotted, that go from location $\ell$ to
location $\ell'$.
In addition, the block must satisfy a constraint on the length of 
the output produced by the subsequence $\rho|Z$, where 
$Z=[\ell,\ell'] ~\setminus~ \big([x,x']\times[y,y']\big)$.
The latter set $Z$ consists of location that are either to the left
of the position of $\ell$ or to the right of the position of $\ell'$.
For example, in Figure \ref{fig:block-sweeping} the set $Z$ coincides
with the area outside the hatched rectangle. Accordingly, the portion 
of the subsequence $\rho|Z$ is represented by the dotted bold arrows.

Diagonals and blocks are used as key objects to derive a notion of
decomposition for a run of a sweeping transducer. We formalize this
notion below.

\begin{defi}\label{def:decomposition-sweeping}
A \emph{$\bound$-decomposition} of a run $\rho$ of $\cT$ is a factorization 
$\prod_i\,\rho[\ell_i,\ell_{i+1}]$ of $\rho$ into $\bound$-diagonals and $\bound$-blocks.
\end{defi}

\reviewOne[inline]{Fig.6 does not look like a fully factorized run: what happens for example, between the end of $D_1$ and $l_1$? Can a concrete example be devised? This works rather well in Part 8.}%
\gabriele{Improved the figure and the explanation}%
Figure~\ref{fig:decomposition-sweeping} gives an example of such a
decomposition. 
\anca{dans Fig 7, les aretes de $B_2$ vont dans la
  mauvaise direction}%
\gabriele{I have corrected this. Now the figure is a bit more regular (perhaps too much), 
          but at least conveys the correct information!}%
Each factor is either a diagonal $D_i$ or a block $B_i$.
Inside each factor we highlight by thick arrows the portions 
of the run that can be simulated by a one-way transducer, 
either because they are produced from left to right (thus forming diagonals) 
or because they are periodic (thus forming blocks).
We also recall from Definition \ref{def:factors-sweeping} 
that most of the output is produced inside the 
hatched rectangles, since the output produced by a diagonal 
or a block outside the corresponding blue or red hatched rectangle 
has length at most $2\hmaxsq\bound$.
Finally, we observe that the locations delimiting the factors of the 
decomposition are arranged following the natural order of positions and levels. 
All together, this means that the output produced by a run that enjoys
a decomposition can be simulated in a one-way manner. 

\begin{figure}[t]
\centering
\begin{tikzpicture}[baseline=0, inner sep=0, outer sep=0, minimum size=0pt, scale=0.5]
  \tikzstyle{dot} = [draw, circle, fill=white, minimum size=4pt]
  \tikzstyle{fulldot} = [draw, circle, fill=black, minimum size=4pt]
  \tikzstyle{grayfactor} = [->, shorten >=1pt, rounded corners=6, gray, thin, dotted]
  \tikzstyle{factor} = [->, shorten >=1pt, rounded corners=6]
  \tikzstyle{dotfactor} = [->, shorten >=1pt, dotted, rounded corners=6]
  \tikzstyle{fullfactor} = [->, >=stealth, shorten >=1pt, very thick, rounded corners=6]
  \tikzstyle{dotfullfactor} = [->, >=stealth, shorten >=1pt, dotted, very thick, rounded corners=6]

  \draw (0,0) node (node0) {};
  \draw (30,0) node (node1) {};
  \draw (30,1) node (node2) {};
  \draw (0,1) node (node3) {};
  \draw (0,2) node (node4) {};
  \draw (30,2) node (node5) {};
  \draw (30,3) node (node6) {};
  \draw (0,3) node (node7) {};
  \draw (0,4) node (node8) {};
  \draw (30,4) node (node9) {};
  \draw (30,5) node (node10) {};
  \draw (0,5) node (node11) {};
  \draw (0,6) node (node12) {};
  \draw (30,6) node (node13) {};
  \draw (30,7) node (node14) {};
  \draw (0,7) node (node15) {};
  \draw (0,8) node (node16) {};
  \draw (30,8) node (node17) {};

  \draw [grayfactor] (node0) -- (node1.center) -- (node2.center) -- 
                            (node3.center) -- (node4.center) -- 
                            (node5.center) -- (node6.center) -- 
                            (node7.center) -- (node8.center) -- 
                            (node9.center) -- (node10.center) -- 
                            (node11.center) -- (node12.center) -- 
                            (node13.center) -- (node14.center) -- 
                            (node15.center) -- (node16.center) -- (node17);

  \fill [pattern=north east lines, pattern color=blue!18] (0,-0.5) rectangle (8,2.5) ;
  \fill [pattern=north east lines, pattern color=red!20] (8,1.5) rectangle (14,4.5);
  \fill [pattern=north east lines, pattern color=blue!18] (14,3.5) rectangle (23,6.5);
  \fill [pattern=north east lines, pattern color=red!20] (23,5.5) rectangle (29,8.5);

  \draw (0,0) node [fulldot] (l0) {};
  \draw (l0) node [below=2mm] {$\ell_0$};
  \draw (3,0) node [dot] (m1) {};
  \draw (4,2) node [dot] (m2) {};
  \draw (7,2) node [dot] (m3) {};
  \draw (8,2) node [fulldot] (l1) {};
  \draw (l1) node [below=2mm] {$\ell_1$};
  \draw (14,4) node [fulldot] (l2) {};
  \draw (l2) node [above=2mm] {$\ell_2$};
  \draw (15,4) node [dot] (m4) {};
  \draw (18,4) node [dot] (m5) {};
  \draw (19,6) node [dot] (m6) {};
  \draw (22,6) node [dot] (m7) {};
  \draw (23,6) node [fulldot] (l3) {};
  \draw (l3) node [below=2mm] {$\ell_3$};
  \draw (29,8) node [fulldot] (l4) {};
  \draw (l4) node [above=2mm] {$\ell_4$};

  \draw [fullfactor] (l0) -- (m1);
  \draw [fullfactor] (m2) -- (m3);

  \draw [fullfactor] (l1) -- (14,2);
  \draw [fullfactor] (14,3) -- (8,3);
  \draw [fullfactor] (8,4) -- (l2);

  \draw [fullfactor] (m4) -- (m5);
  \draw [fullfactor] (m6) -- (m7);

  \draw [fullfactor] (l3) -- (29,6);
  \draw [fullfactor] (29,7) -- (23,7);
  \draw [fullfactor] (23,8) -- (l4);

  \draw (3.5,0.5) node [above=0.3mm, rectangle, fill=white, minimum size=14pt] {\small $D_1$};
  \draw (18.5,5.5) node [below=0.3mm, rectangle, fill=white, minimum size=14pt] {\small $D_2$};
  \draw (11,3) node [rectangle, fill=white, minimum size=14pt] {$B_1$};
  \draw (26,7) node [rectangle, fill=white, minimum size=14pt] {$B_2$};
\end{tikzpicture}
\caption{Decomposition of a run into diagonals and blocks.}\label{fig:decomposition-sweeping}
\end{figure}
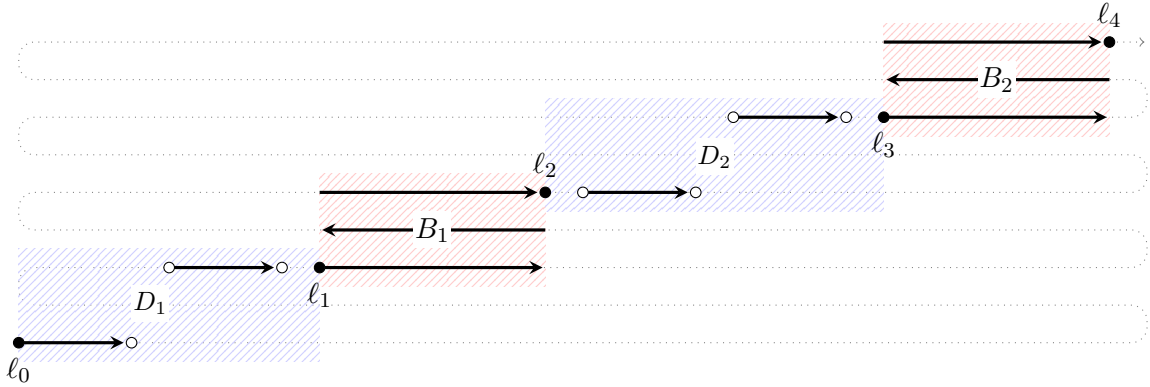


\medskip
\reviewOne[inline]{Part 5: In "From periodicity of inversions to existence of decompositions", the general proof structure in general and transition between lemmas in particular could be made more visually and conceptually explicit, as the global idea of the proof is difficult to grasp at first. A mini "road map" as made in Part 3 for the whole of Part 4 and 5 might be another way to help the reader through these rather intricate proofs.}%

\subsection*{From periodicity of inversions to existence of decompositions.}
Now that we set up the definition of $\bound$-decomposition, we turn 
towards proving the implication \PR2 $\Rightarrow$ \PR3 of Theorem \ref{thm:main2}.
In fact, we will prove a slightly stronger result than \PR2 $\Rightarrow$ \PR3,
which is stated further below.
Formally, when we say that a run \emph{$\rho$ satisfies \PR2} we mean 
that for every inversion $(L_1,\ell_1,L_2,\ell_2)$ of $\rho$, the word
$\out{\tr{\ell_1}} ~ \out{\rho[\ell_1,\ell_2]} ~ \out{\tr{\ell_2}}$
has period $\gcd(|\out{\tr{\ell_1}}|, |\out{\tr{\ell_2}}|) \le \bound$.
We aim at proving that every run that satisfies \PR2 enjoys a decomposition,
independently of whether other runs do or do not satisfy \PR2:

\begin{prop}\label{prop:decomposition-sweeping}
If $\rho$ is a run of $\cT$ that satisfies \PR2, then $\rho$ admits a $\bound$-decomposition.
\end{prop}

\felix{a few lines for the roadmap asked by the reviewer, but I don't think they're useful at all}%
\gabriele{I think it is good}%
Let us fix a run $\rho$ of $\cT$ and assume that it satisfies \PR2. 
To show that $\rho$ admits a $\bound$-decomposition, we will identify
the blocks of the decomposition as equivalence classes of a suitable 
relation based on inversions (cf.~Definition \ref{def:crossrel}). 
Then, we will use combinatorial arguments (notably, 
Lemmas \ref{lem:output-minimal-sweeping} and \ref{lem:overlapping})
to prove that the constructed blocks satisfy the desired properties. 
Finally, we will show how the resulting equivalence classes form all
the necessary blocks of the decomposition, in the sense that the factors 
in between those classes are diagonals.

We begin by introducing the equivalence relation by means of which we can
then identify the blocks of a decomposition of $\rho$.

\begin{defi}\label{def:crossrel}
Let $\crossrel$ be the relation that pairs every two locations $\ell,\ell'$ of $\rho$ 
whenever there is an inversion $(L_1,\ell_1,L_2,\ell_2)$ of $\rho$ such that 
$\ell_1 \leqtime \ell,\ell' \leqtime \ell_2$, namely, whenever $\ell$ and $\ell'$ 
occur within the same inversion. 
Let $\simeq$ be the reflexive and transitive closure of $\crossrel$.
\end{defi}

It is easy to see that every equivalence class of $\simeq$ is a convex 
subset with respect to the run order on locations of $\rho$. 
Moreover, every \emph{non-singleton} equivalence class of $\simeq$ is a 
union of a series of inversions that are two-by-two overlapping. 
One can refer to Figure~\ref{fig:overlapping}
for an intuitive account of what we mean by two-by-two overlapping: the thick 
arrows represent factors of the run that lie entire inside an $\simeq$-equivalence class,
each inversion is identified by a pair of consecutive anchor points with the same 
color. According to the run order, between every pair of anchor points with the 
same color, there is at least one anchor point of different color: this means that 
the inversions corresponding to the two colors are overlapping.

\reviewOne[inline]{Fig.7 might benefit from some $L_1$, $L_2$ appearing, to illustrate the point made in Claim 5.6
  \felix[inline]{I don't agree}%
  \olivier[inline]{I also disagree. Argued in review1-answer.txt.}%
}%
Formally, we say that an inversion $(L_1,\ell_1,L_2,\ell_2)$ \emph{covers} 
a location $\ell$ when $\ell_1 \leqtime \ell \leqtime \ell_2$. We say that 
two inversions $(L_1,\ell_1,L_2,\ell_2)$ and $(L_3,\ell_3,L_4,\ell_4)$
are \emph{overlapping} if $(L_1,\ell_1,L_2,\ell_2)$ covers $\ell_3$ 
and $(L_3,\ell_3,L_4,\ell_4)$ covers $\ell_2$ (or the other way around).

\begin{figure}[!t]
\centering
\begin{tikzpicture}[baseline=0, inner sep=0, outer sep=0, minimum size=0pt, scale=0.6, yscale=0.9]
  \tikzstyle{dot} = [draw, circle, fill=white, minimum size=4pt]
  \tikzstyle{fulldot} = [draw, circle, fill=black, minimum size=4pt]
  \tikzstyle{grayfactor} = [->, shorten >=1pt, rounded corners=6, gray, thin, dotted]
  \tikzstyle{factor} = [->, shorten >=1pt, rounded corners=6]
  \tikzstyle{dotfactor} = [->, shorten >=1pt, dotted, rounded corners=6]
  \tikzstyle{fullfactor} = [->, >=stealth, shorten >=1pt, very thick, rounded corners=6]
  \tikzstyle{dotfullfactor} = [->, >=stealth, shorten >=1pt, dotted, very thick, rounded corners=6]

  \draw [draw=red!25, line width=10pt, double distance = 0.5cm, line cap=round] (2,-1) -- (0,1);
  \draw [draw=blue!25, line width=10pt, double distance = 0.5cm, line cap=round] (6,0) -- (4,3);
  \draw [draw=green!50!black!25, line width=10pt, double distance = 0.5cm, line cap=round] (10,2) -- (8,5);
  \draw [draw=orange!25, line width=10pt, double distance = 0.5cm, line cap=round] (18,6) -- (16,8);


  \draw (-2,-2) node (left0) {};
  \draw (20,-2) node (right0) {};
  \draw (20,-1) node (right1) {};
  \draw (2,-1) node [dot, draw, minimum size=6pt, red] (l1) {};
  \draw (-2,-1) node (left1) {};
  \draw (-2,0) node (left2) {};
  \draw (6,0) node [dot, draw, minimum size=6pt, blue] (l3) {};
  \draw (20,0) node (right2) {};
  \draw (20,1) node (right3) {};
  \draw (0,1) node [dot, draw, minimum size=6pt, red] (l4) {};
  \draw (-2,1) node (left3) {};
  \draw (-2,2) node (left4) {};
  \draw (10,2) node [dot, draw, minimum size=6pt, green!50!black] (l5) {};
  \draw (20,2) node (right4) {};
  \draw (20,3) node (right5) {};
  \draw (4,3) node [dot, draw, minimum size=6pt, blue] (l6) {};
  \draw (-2,3) node (left5) {};
  \draw (-2,4) node (left6) {};
  \draw (20,4) node (right6) {};
  \draw (20,5) node (right7) {};
  \draw (8,5) node [dot, draw, minimum size=6pt, green!50!black] (l8) {};two
  \draw (-2,5) node (left7) {};
  \draw (-2,6) node (left8) {};
  \draw (18,6) node [dot, draw, minimum size=6pt, orange] (ln-3) {};
  \draw (20,6) node (right8) {};
  \draw (20,7) node (right9) {};
  \draw (-2,7) node (left9) {};
  \draw (-2,8) node (left10) {};
  \draw (16,8) node [dot, draw, minimum size=6pt, orange] (ln) {};
  \draw (20,8) node (lend) {};

  \draw (l1) node [below right=1.25mm] {$~\ell_0$};
  \draw (l3) node [below right=1.25mm] {$~\ell_2$};
  \draw (l4) node [below left=0.875mm] {$\ell_1~$};
  \draw (l5) node [below right=1.25mm] {$~\ell_4$};
  \draw (l6) node [below left=0.875mm] {$\ell_3~$};
  \draw (l8) node [below left=0.875mm] {$\ell_5~$};
  \draw (ln-3) node [below right=1.2mm] {$~\ell_{2k}$}; 
  \draw (ln) node [above=2mm] {$\ell_{2k+1}$}; 

  \draw [grayfactor] (left0) -- (right0.center) -- (right1.center) -- (2.5,-1);
  \draw [fullfactor] (l1) -- (left1.center) -- (left2.center) -- (l3);
  \draw [fullfactor] (l3) -- (right2.center) -- (right3.center) -- (l4);
  \draw [fullfactor] (l4) -- (left3.center) -- (left4.center) -- (l5);
  \draw [fullfactor] (l5) -- (right4.center) -- (right5.center) -- (l6);
  \draw [fullfactor] (l6) -- (left5.center) -- (left6.center) -- 
                     (right6.center) -- (right7.center) -- (l8);
  \draw [fullfactor,dashed] (l8) -- (left7.center) -- (left8.center) -- (ln-3);
  \draw [fullfactor] (ln-3) -- (right8.center) -- (right9.center) -- 
                     (left9.center) -- (left10.center) -- (ln);
  \draw [grayfactor] (ln) -- (lend);
\end{tikzpicture}
\caption{A non-singleton $\simeq$-equivalence class seen as a series of overlapping inversions.}\label{fig:overlapping}
\end{figure}
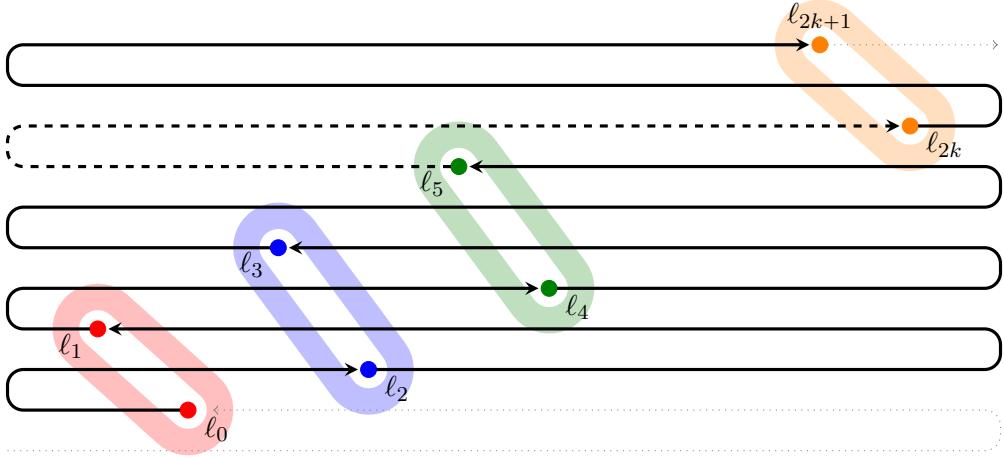


The next lemma uses the fact that $\rho$ satisfies \PR2 to deduce that 
the output produced inside every $\simeq$-equivalence class has
period at most $\bound$. Note that the proof below does not exploit
the fact that the transducer is sweeping. 

\begin{lem}\label{lem:overlapping}
If $\rho$ satisfies \PR2 and $\ell\leqtime\ell'$ are two locations of $\rho$ 
such that $\ell \simeq \ell'$, then the output $\out{\rho[\ell,\ell']}$ 
produced between these locations has period at most $\bound$.
\end{lem}

\begin{proof}
The claim for $\ell=\ell'$ holds trivially, so we assume that $\ell\lesstime\ell'$. 
Since the $\simeq$-equivalence class that contains $\ell,\ell'$ is non-singleton,
we know that there is a series of inversions
\[
  (L_0,\ell_0,L_1,\ell_1) \quad
  (L_2,\ell_2,L_3,\ell_3) \quad
  \dots\dots \quad
  (L_{2k},\ell_{2k},L_{2k+1},\ell_{2k+1})
\]
that are two-by-two overlapping and such that 
$\ell_0\leqtime\ell\lesstime\ell'\leqtime\ell_{2k+1}$. 
Without loss of generality, we can assume that every 
inversion $(L_{2i},\ell_{2i},L_{2i+1},\ell_{2i+1})$ is \emph{maximal} 
in the following sense: there is no other inversion 
$(\tilde L,\tilde\ell,\tilde L',\tilde\ell') \neq (L_{2i},\ell_{2i},L_{2i+1},\ell_{2i+1})$ 
such that $\tilde\ell \leqtime \ell_{2i} \leqtime \ell_{2i+1} \leqtime \tilde\ell'$.

For the sake of brevity, let $v_i = \out{\tr{\ell_i}}$ and $p_i = | v_i |$. 
Since $\rho$ satisfies \PR2 (recall Proposition~\ref{prop:periodicity-sweeping}), we know that, for all $i=0,\dots,n$, the word
\[
  v_{2i} ~ \out{\rho[\ell_{2i},\ell_{2i+1}]} ~ v_{2i+1}
\]
has period that divides both $p_{2i}$ and $p_{2i+1}$ and is at most $\bound$.
In order to show that the period of $\out{\rho[\ell,\ell']}$ is also bounded
by $\bound$, it suffices to prove the following claim by induction on $i$:

\reviewOne[inline]{Claim 5.6, end of page 18: making explicit what words you wish to use Fine-Wilf on might make for a more readable proof.
\olivier[inline]{answered in review1-answer.txt}%
}%

\begin{clm}
For all $i=0,\dots,k$, the word
$\outb{\rho[\ell_0,\ell_{2i+1}]} \: v_{2i+1}$ 
has period at most $\bound$ that divides $p_{2i+1}$.
\end{clm}

The base case $i=0$ follows immediately from our hypothesis,
since $(L_0,\ell_0,L_1,\ell_1)$ is an inversion.
For the inductive step, we assume that the claim holds for $i<k$, 
and we prove it for $i+1$. First of all, we factorize our word as follows:
\[
  \outb{\rho[\ell_0,\ell_{2i+3}]} ~ v_{2i+3} 
  ~=~
  \rightward{ \underbracket[0.5pt]{ \phantom{ 
                                              \outb{\rho[\ell_0,\ell_{2i+2}]} ~ 
                                              \outb{\rho[\ell_{2i+2},\ell_{2i+1}]} ~ } }%
                                 _{\text{period } p_{2i+1}} }
  \outb{\rho[\ell_0,\ell_{2i+2}]}
  \overbracket[0.5pt]{ ~ \outb{\rho[\ell_{2i+2},\ell_{2i+1}]} ~
                       \outb{\rho[\ell_{2i+1},\ell_{2i+3}]} ~
                       v_{2i+3} }%
                    ^{\text{periods } p_{2i+2} \text{ and } p_{2i+3}}
  .
\]
By the inductive hypothesis, the output produced between $\ell_0$ 
and $\ell_{2i+1}$, extended to the right with $v_{2i+1}$, 
has period that divides $p_{2i+1}$. 
Moreover, because $\rho$ satisfies \PR2 and 
$(L_{2i+2},\ell_{2i+2},L_{2i+3},\ell_{2i+3})$ is an inversion, 
the output produced between the locations 
$\ell_{2i+2}$ and $\ell_{2i+3}$, 
extended to the left with $v_{2i+2}$ and to the 
right with $v_{2i+3}$, has period that divides both $p_{2i+2}$ and $p_{2i+3}$.
Note that this is not yet sufficient for applying Fine-Wilf's theorem, 
since the common factor $\outb{\rho[\ell_{2i+2},\ell_{2i+1}]}$ 
might be too short (possibly just equal to $v_{2i+2}$).
The key argument here is to prove that the interval $[\ell_{2i+2},\ell_{2i+1}]$ 
is covered by an inversion
which is different from those that we considered above, namely,
i.e.~$(L_{2i+2},\ell_{2i+2},L_{2i+1},\ell_{2i+1})$. For example,
$[\ell_2,\ell_1]$ in
Figure~\ref{fig:overlapping}  is covered by the inversion
$(L_2,\ell_2,L_1,\ell_1)$.

For this, we have to prove that the anchor points $\ell_{2i+2}$ and $\ell_{2i+1}$ 
are correctly ordered w.r.t.~$\leqtime$ and the ordering of positions
(recall Definition~\ref{def:inversion-sweeping}).
First, we observe that $\ell_{2i+2} \leqtime \ell_{2i+1}$, since
$(L_{2i},\ell_{2i},L_{2i+1},\ell_{2i+1})$ and $(L_{2i+2},\ell_{2i+2},L_{2i+3},\ell_{2i+3})$ 
are overlapping inversions.
Next, we prove that the position of $\ell_{2i+1}$ is strictly to the left of 
the position of $\ell_{2i+2}$. By way of contradiction, suppose that this is 
not the case, namely, $\ell_{2i+1}=(x_{2i+1},y_{2i+1})$, $\ell_{2i+2}=(x_{2i+2},y_{2i+2})$, 
and $x_{2i+1} \ge x_{2i+2}$. Because $(L_{2i},\ell_{2i},L_{2i+1},\ell_{2i+1})$ and
$(L_{2i+2},\ell_{2i+2},L_{2i+3},\ell_{2i+3})$ are inversions, 
we know that $\ell_{2i+3}$ is strictly to the left of $\ell_{2i+2}$ 
and $\ell_{2i+1}$ is strictly to the left of $\ell_{2i}$. 
This implies that $\ell_{2i+3}$ is strictly to the left of $\ell_{2i}$,
and hence $(L_{2i},\ell_{2i},L_{2i+3},\ell_{2i+3})$ is also an inversion.
Moreover, recall that $\ell_{2i} \leqtime \ell_{2i+2} \leqtime \ell_{2i+1} \leqtime \ell_{2i+3}$.
This contradicts the maximality of $(L_{2i},\ell_{2i},L_{2i+1},\ell_{2i+1})$,
which we assumed at the beginning of the proof.
Therefore, we must conclude that $\ell_{2i+1}$ is strictly to the left of $\ell_{2i+2}$.

Now that we know that $\ell_{2i+2} \leqtime \ell_{2i+1}$ and that $\ell_{2i+1}$ is to the left of $\ell_{2i+2}$,
we derive the existence of the inversion $(L_{2i+2},\ell_{2i+2},L_{2i+1},\ell_{2i+1})$.
Again, because $\rho$ satisfies \PR2, we know that the word
$v_{2i+2} ~ \out{\rho[\ell_{2i+2},\ell_{2i+1}]} ~ v_{2i+1}$ has period at most $\bound$
that divides $p_{2i+2}$ and $p_{2i+1}$.
Summing up, we have:
\begin{enumerate}
  \item $w_1 ~=~ \outb{\rho[\ell_0,\ell_{2i+1}]} ~ v_{2i+1}$ has period $p_{2i+1}$,
  \item $w_2 ~=~ v_{2i+2} ~ \outb{\rho[\ell_{2i+2},\ell_{2i+1}]} ~ v_{2i+1}$ 
        has period $p = \gcd(p_{2i+2},p_{2i+1})$, 
  \item $w_3 ~=~ v_{2i+2} ~ \outb{\rho[\ell_{2i+2},\ell_{2i+3}]} ~ v_{2i+3}$
        has period $p' = \gcd(p_{2i+2},p_{2i+3})$.
\end{enumerate}
We are now ready to exploit our stronger variant of Fine-Wilf's theorem, 
that is, Theorem~\ref{thm:fine-wilf}. 

We begin with  (1) and (2) above.
Let $w = \outb{\rho[\ell_{2i+2},\ell_{2i+1}]}~ v_{2i+1}$ be the common suffix of
$w_1$ and $w_2$. First note that since $p$ divides $p_{2i+2}$, the
word $w$ is also a prefix of $w_2$, thus we can write
$w_2=w\,w'_2$. Second, note that the length of $w$ is at least
$|v_{2i+1}| = p_{2i+1} = p_{2i+1} + p - \gcd(p_{2i+1},p)$. We can
apply now Theorem~\ref{thm:fine-wilf} to $w_1=w'_1\,w$ and
$w_2=w\,w'_2$ and obtain:
\begin{enumerate}
  \setcounter{enumi}{3}
  \item $w_4 ~=~ w'_1 \: w \: w'_2 
             ~=~ \outb{\rho[\ell_0,\ell_{2i+2}]} ~ v_{2i+2} ~ \outb{\rho[\ell_{2i+2},\ell_{2i+1}]} ~ v_{2i+1}$ 
        has period $p$.
\end{enumerate}

We apply next Theorem~\ref{thm:fine-wilf} to (2) and (3), namely, to
the words $w_2$ and $w_3$ with $v_{2i+2}$ as common factor. It is not
difficult to check that $|v_{2i+2}|=p_{2i+2} \ge p+p'-p''$ with
$p''=\gcd(p,p')$, using the definitions of $p$ and $p'$: we can
write $p_{2i+2}=p''rq=p''r'q'$ with $p=p''r$ and $p'=p''r'$. It suffices  to check
that $p''rq+p''r'q' \ge 2(p''r + p''r'-p'')$, hence that $rq+r'q' \ge 2r+2r'-2$. This
is clear if $\min(q,q')>1$. Otherwise the inequality $p_{2i+2} \ge
p+p'-p''$ follows easily because $p=p''$ or $p'=p''$ holds. Hence we obtain
that $w_2$ and $w_3$ have both period $p''$.

Applying once more Theorem~\ref{thm:fine-wilf} to $w_3$ and $w_4$ with
$v_{2i+2}$ as common factor, yields period $p''$ for the word
\[
w_5 ~=~ \outb{\rho[\ell_0,\ell_{2i+2}]} ~ v_{2i+2} ~ 
                 \outb{\rho[\ell_{2i+2},\ell_{2i+3}]} ~ v_{2i+3}
\]

Finally, the periodicity is not affected when we remove
factors of length multiple than the period. In particular, 
by removing the factor $v_{2i+2}$ from $w_5$, 
we obtain the desired word
$\outb{\rho[\ell_0,\ell_{2i+3}]} ~ v_{2i+3}$, whose period
$p''$ divides $p_{2i+3}$. This proves the claim for the inductive step,
and completes the proof of the proposition.
\end{proof}

\medskip
The $\simeq$-classes considered so far cannot be directly used to
define the blocks in  the desired decomposition of $\rho$, since the $x$-coordinates of their 
endpoints might not be in the appropriate order. The next definition
takes care of this, by enlarging the $\simeq$-classes according to
$x$-coordinates of the anchor points in the equivalence class. 

\bigskip\noindent
\begin{minipage}[l]{\textwidth-5.8cm}
\begin{defi}\label{def:bounding-box-sweeping}
Consider a non-singleton $\simeq$-equivalence class $K=[\ell,\ell']$. 
Let $\an{K}$
be the restriction of $K$ to the anchor points occurring in some inversion, 
and $X_{\an{K}} = \{x \::\: \exists y\: (x,y)\in \an{K}\}$ 
be the projection of $\an{K}$ on positions.
We define $\block{K}=[\tilde\ell,\tilde\ell']$, where 
\begin{itemize}
  \item $\tilde\ell$ is the latest location $(\tilde x,\tilde y) \leqtime \ell$ 
        such that $\tilde x = \min(X_{\an{K}})$, 
  \item $\tilde\ell'$ is the earliest location $(\tilde x,\tilde y) \geqtime \ell'$ 
        such that $\tilde x = \max(X_{\an{K}})$
\end{itemize}  
(note that the locations $\tilde\ell,\tilde\ell'$ exist since $\ell,\ell'$ 
are anchor points in some inversion).
\end{defi}
\end{minipage}
\begin{minipage}[r]{5.7cm}
\centering
\scalebox{0.9}{
\begin{tikzpicture}[baseline=0, inner sep=0, outer sep=0, minimum size=0pt, scale=0.52, yscale=1.2]
  \tikzstyle{dot} = [draw, circle, fill=white, minimum size=4pt]
  \tikzstyle{fulldot} = [draw, circle, fill=black, minimum size=4pt]
  \tikzstyle{grayfactor} = [->, shorten >=1pt, rounded corners=8, gray, thin, dotted]
  \tikzstyle{factor} = [->, shorten >=1pt, rounded corners=8]
  \tikzstyle{dotfactor} = [->, shorten >=1pt, dotted, rounded corners=8]
  \tikzstyle{fullfactor} = [->, >=stealth, shorten >=1pt, very thick, rounded corners=8]
  \tikzstyle{dotfullfactor} = [->, >=stealth, shorten >=1pt, dotted, very thick, rounded corners=8]

  \fill [pattern=north east lines, pattern color=gray!25]
        (2,-0.325) rectangle (8,4.325);
  \draw [dashed, thin, gray] (2,-1.25) -- (2,5);
  \draw [dashed, thin, gray] (8,-1.25) -- (8,5);
  \draw (2,-1.5) node [below] {\footnotesize $\min(X_{\an{K}})$};
  \draw (8,-1.5) node [below] {\footnotesize $\max(X_{\an{K}})$};

  \draw (0,0) node (node0) {};
  \draw (2,0) node [dot] (node1) {};
  \draw (8,0) node (node2) {};
  \draw (10,0) node (node3) {};
  \draw (10,1) node (node4) {};
  \draw (8,1) node (node5) {};
  \draw (4,1) node [fulldot] (node6) {};
  \draw (2,1) node (node7) {};
  \draw (0,1) node (node8) {};
  \draw (0,2) node (node9) {};
  \draw (2,2) node (node10) {};
  \draw (8,2) node (node11) {};
  \draw (10,2) node (node12) {};
  \draw (10,3) node (node13) {};
  \draw (8,3) node (node14) {};
  \draw (6,3) node [fulldot] (node15) {};
  \draw (2,3) node (node16) {};
  \draw (0,3) node (node17) {};
  \draw (0,4) node (node18) {};
  \draw (2,4) node (node19) {};
  \draw (8,4) node [dot] (node20) {};
  \draw (10,4) node (node21) {};

  \draw [grayfactor] (node0) -- (node1); 
  \draw [dotfullfactor] (node1) -- (node2); 
  \draw [dotfullfactor] (node2) -- (node3.center) -- (node4.center) -- (node5); 
  \draw [dotfullfactor] (node5) -- (node6); 
  \draw [fullfactor] (node6) -- (node7.center) -- (node8.center) -- (node9.center) -- (node10.center)
                  -- (node11.center) -- (node12.center) -- (node13.center) -- (node14.center) -- (node15); 
  \draw [dotfullfactor] (node15) -- (node16); 
  \draw [dotfullfactor] (node16) -- (node17.center) -- (node18.center) -- (node19); 
  \draw [dotfullfactor] (node19) -- (node20); 
  \draw [grayfactor] (node20) -- (node21); 
  
  \draw (node1) node [below left = 2mm] {$\tilde\ell\ $};
  \draw (node6) node [below = 2mm] {$\ell$};
  \draw (node15) node [above = 2mm] {$\ell'$};
  \draw (node20) node [above right = 2mm] {$\ \tilde\ell'$};
\end{tikzpicture}
}

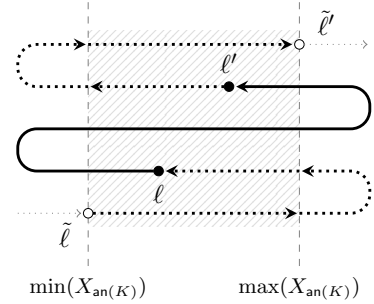
\captionof{figure}{Block construction.\label{fig:block-construction-sweeping}}
\end{minipage}

\smallskip

\begin{lem}\label{lem:bounding-box-sweeping}
If $K$ is a non-singleton $\simeq$-equivalence class, 
then $\rho|\block{K}$ is a $\bound$-block.
\end{lem}

\begin{proof}
Consider a non-singleton $\simeq$-class $K=[\ell,\ell']$ and let $\an{K}$, $X_{\an{K}}$, and 
$\block{K}=[\tilde\ell,\tilde\ell']$ 
be as in Definition \ref{def:bounding-box-sweeping}.
The reader can refer to Figure \ref{fig:block-construction-sweeping}
to quickly recall the notation.
We need to verify that $\rho[\tilde\ell,\tilde\ell']$ is a $\bound$-block 
(cf.~Definition \ref{def:factors-sweeping}), namely, that:
\begin{itemize}
  \item $\tilde\ell=(\tilde x,\tilde y)$, $\tilde\ell'=(\tilde x',\tilde y')$,
        with $\tilde x \le \tilde x'$,
  \item the output produced by $\rho[\tilde\ell,\tilde\ell']$ 
        is almost periodic with bound $2\bound$,
  \item the output produced by the subsequence $\rho|Z$, 
        where $Z=[\tilde\ell,\tilde\ell'] ~\setminus~ 
                 \big([\tilde x,\tilde x']\times[\tilde y,\tilde y']\big)$,
        has length at most $2\hmax\bound$.
\end{itemize}
The first condition $\tilde x \le \tilde x'$ follows immediately from 
the definition of $\tilde x$ and $\tilde x'$ as $\min(X_{\an{K}})$ 
and $\max(X_{\an{K}})$, respectively.

Next, we prove that the output produced by the factor 
$\rho[\tilde\ell,\tilde\ell']$ is almost periodic with bound $2\bound$.
By Definition \ref{def:bounding-box-sweeping}, we have 
$\tilde\ell \leqtime \ell \lesstime \ell' \leqtime \tilde\ell'$,
and by Lemma \ref{lem:overlapping} we know that $\out{\rho[\ell,\ell']}$ 
is periodic with period at most $\bound$ ($\le 2\bound$). So it suffices to show that
the lengths of the words $\out{\rho[\tilde\ell,\ell]}$ and 
$\out{\rho[\ell',\tilde\ell']}$ are at most $2\bound$. 
We shall focus on the former word, as the arguments for 
the latter are similar. 

First, we note that the factor $\rho[\tilde\ell,\ell]$ 
lies entirely to the right of position $\tilde x$, and 
in particular, it starts at an even level $\tilde y$. This follows
from the definition of $\tilde\ell$, and whether $\ell$ itself is at
an odd/even level. In particular, the location $\ell$ is either at the same level as $\tilde\ell$,
or just one level above. 

Now, suppose, by way of contradiction, that $|\out{\rho[\tilde\ell,\ell]}| > 2\bound$. 
We head towards a contradiction by finding a location $\ell'' \lesstime \ell$ 
that is $\simeq$-equivalent to the first location $\ell$ of the $\simeq$-equivalence 
class $K$. 
Since the location $\ell$ is either at the same level as $\tilde\ell$,
or just above it, 
the factor $\rho[\tilde\ell,\ell]$ is of the form $\alpha\,\beta$,
where $\alpha$ is rightward factor lying on the same level as $\tilde\ell$ 
and $\beta$ is either empty or a leftward factor on the next level.
Moreover, since $|\out{\rho[\tilde\ell,\ell]}| > 2\bound$, we know that
either $|\out{\alpha}| > \bound$ or $|\out{\beta}| > \bound$.
Thus,  Lemma \ref{lem:output-minimal-sweeping} 
says that one of the two factors $\alpha,\beta$ is not output-minimal.
In particular, there is a loop $L_1$, strictly to the right of $\tilde x$, 
that intercepts a subfactor $\gamma$ of $\rho[\tilde\ell,\ell]$,
with $\out{\gamma}$ non-empty and output-minimal.

Let $\ell''$ be the first location of the factor $\gamma$. 
Clearly, $\ell''$ is an anchor point of $L$ and $\out{\tr{\ell''}}\neq\emptystr$.
Further recall that $\tilde x=\min(X_{\an{K}})$ is the leftmost position of 
locations in the class $K=[\ell,\ell']$ that are also anchor points of inversions. 
In particular, there is a loop $L_2$ with some anchor point
$\ell''_2=(\tilde x,y''_2)\in \an{K}$, and such that $\tr{\ell''}$ is
non-empty and output-minimal. 
Since $\ell'' \lesstime \ell \leqtime \ell''_2$ 
and the position of $\ell''$ is to the right of the position of $\ell''_2$, 
we know that $(L_1,\ell'',L_2,\ell''_2)$ is also an inversion, 
and hence $\ell'' \simeq \ell''_2 \simeq \ell$.
But since $\ell'' \lesstime \ell$, we get a contradiction with the 
assumption that $\ell$ is the first location of a $\simeq$-class. 
In this way we have shown that $|\out{\rho[\ell_1,\ell]}| \le 2\bound$. 

It remains to show that the output produced 
by the subsequence $\rho|Z$, where
$Z=[\tilde\ell,\tilde\ell'] ~\setminus~ \big([\tilde x,\tilde x']\times[\tilde y,\tilde y']\big)$,
has length at most $2\hmax\bound$.
For this it suffices to prove that 
$|\out{\alpha}| \le \bound$ for every factor $\alpha$ of
$\rho[\tilde\ell, \tilde\ell']$ that lies 
at a single level and either to the left of $\tilde x$ or to the right of $\tilde x'$.
By symmetry, we consider only one of the two types of factors.
Suppose, by way of contradiction, that there is a factor $\alpha$
at level $y''$, to the left of $\tilde x$, 
and such that $|\out{\alpha}| > \bound$.
By Lemma \ref{lem:output-minimal-sweeping} we know that
$\alpha$ is not output-minimal, so there is some loop 
$L_2$ strictly to the left of $\tilde x$ that intercepts an
output-minimal subfactor 
$\beta$ of $\alpha$ with non-empty output.
Let $\ell''$ be the first location of $\beta$. We know that
$\tilde\ell \lesstime \ell'' \leqtime \tilde\ell'$. Since the level
$\tilde y$ is even, this means that the level of $\ell''$ is strictly
greater than $\tilde y$. Since we also know that $\ell$ 
is an anchor point of some inversion, we can take a suitable loop $L_1$ 
with anchor point $\ell$ and obtain that $(L_1,\ell,L_2,\ell'')$ is an 
inversion, so $\ell'' \simeq \ell$. But
this contradicts the fact that $\tilde x$ is the leftmost position of
$\an{K}$.
We thus conclude that $|\out{\alpha}| \le \bound$, and this 
completes the proof that $\rho|\block{K}$ is a $\bound$-block.
\end{proof}

The next lemma shows that blocks do not overlap along the input axis: 

\begin{lem}\label{lem:consecutive-blocks-sweeping}
Suppose that $K_1$ and $K_2$ are two different non-singleton $\simeq$-classes
such that $\ell \lesstime \ell'$ for all $\ell \in K_1$ and $\ell' \in K_2$.
Let $\block{K_1}=[\ell_1,\ell_2]$ and $\block{K_2}=[\ell_3,\ell_4]$,
with $\ell_2=(x_2,y_2)$ and $\ell_3=(x_3,y_3)$. 
Then $x_2 < x_3$.
\end{lem}

\begin{proof}
Suppose by contradiction that $K_1$ and $K_2$ are as in the
statement, but $x_2 \ge x_3$. 
By Definition \ref{def:bounding-box-sweeping}, 
$x_2=\max(X_{\an{K_1}})$ and $x_3=\min(X_{\an{K_2}})$.
This implies the existence of some inversions
$(L,\ell,L',\ell')$ and $(L'',\ell'',L''',\ell''')$ 
such that $\ell=(x_2,y)$ and $\ell'''=(x_3,y''')$.
Moreover, since $\ell \leqtime \ell'''$ and $x_2 \ge x_3$,
we know that $(L,\ell,L''',\ell''')$ is also an inversion,
thus implying that $K_1=K_2$.
\end{proof}

For the sake of brevity, we call \emph{$\simeq$-block} any 
factor of the form $\rho|\block{K}$ that is obtained by applying 
Definition~\ref{def:bounding-box-sweeping} to a non-singleton $\simeq$-class $K$.
The results obtained so far imply that every location covered by an 
inversion is also covered by an $\simeq$-block (Lemma \ref{lem:bounding-box-sweeping}), 
and that the order of occurrence of $\simeq$-blocks is the same as the order of positions 
(Lemma \ref{lem:consecutive-blocks-sweeping}).
So the $\simeq$-blocks can be used as factors for the $\bound$-decomposition
of $\rho$ we are looking for. Below, we show that the remaining 
factors of $\rho$, which do not overlap the $\simeq$-blocks, are $\bound$-diagonals. 
This will complete the construction of a $\bound$-decomposition of $\rho$.

Formally, we say that a factor $\rho[\ell,\ell']$ 
\emph{overlaps} another factor $\rho[\ell'',\ell''']$ if 
$[\ell,\ell'] \:\cap\: [\ell'',\ell'''] \neq \emptyset$, 
$\ell' \neq \ell''$, and $\ell \neq \ell'''$.

\begin{lem}\label{lem:diagonal-sweeping}
Let $\rho[\ell,\ell']$ be a 
factor of $\rho$,
with $\ell=(x,y)$, $\ell'=(x',y')$, and $x\le x'$,
that does not overlap any $\simeq$-block.
Then $\rho[\ell,\ell']$ is a $\bound$-diagonal.
\end{lem}

\begin{proof}
Consider a factor $\rho[\ell,\ell']$, with $\ell=(x,y)$, $\ell'=(x',y')$, 
and $x\le x'$, that does not overlap any $\simeq$-block.
We will focus on locations $\ell''$ with $\ell \leqtime \ell''
\leqtime \ell'$ 
that are anchor points of some loop with $\out{\tr{\ell''}}\neq\emptystr$. 
We denote by $A$ the set of all such locations.

First, we show that the locations in $A$ are monotonic
\reviewOne[inline]{Monotonic$\to$Monotonous?
  \olivier[inline]{we keep monotonic}%
}%
w.r.t.~the position order. Formally, 
we prove that for all $\ell_1,\ell_2\in A$, if $\ell_1=(x_1,y_1) \leqtime \ell_2=(x_2,y_2)$, then
$x_1 \le x_2$. Suppose that this were not the case, namely, that $A$ contained two anchor points
$\ell_1=(x_1,y_1)$ and $\ell_2=(x_2,y_2)$ with $\ell_1 \lesstime \ell_2$ and $x_1 > x_2$.
Let $L_1,L_2$ be the loops of $\ell_1,\ell_2$, respectively, and recall that 
$\out{\tr{\ell_1}},\out{\tr{\ell_2}}\neq\emptystr$. This means that $(L_1,\ell_1,L_2,\ell_2)$
is an inversion, and hence $\ell_1 \simeq \ell_2$. But this contradicts the hypothesis that 
$\rho[\ell,\ell']$ does not overlap any $\simeq$-block.

Next, we identify the floors of our diagonal. Let $y_0,y_1,\dots,y_{n-1}$ be all the {\sl even}
levels that have locations in $A$. For each $i=0,\dots,n-1$, let $\ell_{2i+1}$ (resp.~$\ell_{2i+2}$)
be the first (resp.~last) anchor point of $A$ at level $y_i$. Further let $\ell_0=\ell$ and $\ell_{2n+1}=\ell'$.
Clearly, each factor $\rho[\ell_{2i+1},\ell_{2i+2}]$ is a floor. Moreover, thanks to the previous 
arguments, each location $\ell_{2i}$ is to the left of the location $\ell_{2i+1}$.

It remains to prove that each factor $\rho[\ell_{2i},\ell_{2i+1}]$ produces an 
output of length at most $2\hmax\bound$. By construction, $A$ contains no anchor 
point at an even level and strictly between $\ell_{2i}$ and $\ell_{2i+1}$.
By Lemma \ref{lem:output-minimal-sweeping} this means that the outputs produced
by subfactors of $\rho[\ell_{2i},\ell_{2i+1}]$ that lie entirely at an {\sl even} 
level have length at most $\bound$.
Let us now consider the subfactors $\alpha$ of $\rho[\ell_{2i},\ell_{2i+1}]$ 
that lie entirely at an {\sl odd} level, and let us prove that they produce outputs
of length at most $2\bound$. Suppose that this is not the case, namely, that
$|\out{\alpha}| > 2\bound$. In this case we show that an inversion
would exist at this level. Formally, we can find two locations $\ell'' \lesstime \ell'''$ 
in $\alpha$ such that the prefix of $\alpha$ that ends at location $\ell''$ and the suffix
of $\alpha$ that starts at location $\ell'''$ produce outputs of
length greater than $\bound$. 
By Lemma \ref{lem:output-minimal-sweeping}, 
those two factors would not be output-minimal, 
and hence $\alpha$ would contain disjoint loops $L_1,L_2$ with anchor points $\ell''_1,\ell''_2$
forming an inversion $(L_1,\ell''_1,L_2,\ell''_2)$. But this would imply that $\ell''_1,\ell''_2$
belong to the same non-singleton $\simeq$-equivalence class, which contradicts the hypothesis
that $\rho[\ell,\ell']$ does not overlap any $\simeq$-block.
We must conclude that the subfactors of $\rho[\ell_{2i},\ell_{2i+1}]$
produce outputs of length at most $2\bound$.

Overall, this shows that the output produced by each factor $\rho[\ell_{2i},\ell_{2i+1}]$
has length at most $2\hmax\bound$.
\end{proof}

\smallskip
We have just shown how to construct a $\bound$-decomposition of the run $\rho$
that satisfies \PR2. This proves Proposition \ref{prop:decomposition-sweeping},
as well as the implication \PR2 $\Rightarrow $ \PR3 of Theorem \ref{thm:main2}.

\medskip
\subsection*{From existence of decompositions to an equivalent one-way transducer.}
We now focus on the last implication \PR3 $\Rightarrow $ \PR1 of Theorem \ref{thm:main2}.
More precisely, we show how to construct a one-way transducer $\cT'$ that simulates
the outputs produced by the successful runs of $\cT$ that admit $\bound$-decompositions.
In particular, $\cT'$ turns out to be equivalent to $\cT$ when $\cT$ is one-way definable.
Here we will only give a proof sketch of this construction (as there is no real difference between the sweeping and two-way cases) assuming that $\cT$ is a sweeping
transducer; a fully detailed construction of $\cT'$ from an arbitrary two-way transducer $\cT$ 
will be given in Section~\ref{sec:characterization-twoway} (Proposition \ref{prop:construction-twoway}), 
together with a procedure for deciding one-way definability of $\cT$ (Proposition \ref{prop:complexity}).

\reviewOne[inline]{The structure of Prop 8.6 is all but indistinguishable visually/ This is made all the more aggravating by the fact that there is no corresponding proof in the sweeping case (merely the intuition sketched on Page 23.
  \felix[inline]{added a sentence}
  \olivier[inline]{answered in review1-answer.txt}%
}%

\begin{prop}\label{prop:construction-sweeping}
Given a functional sweeping transducer $\cT$ 
a one-way transducer $\cT'$ can be constructed  in
$2\exptime$ such that the following hold:
\begin{enumerate}
  \item $\cT' \subseteq \cT$,
  \item $\dom(\cT')$ contains all words that induce successful runs of $\cT$ 
        admitting $\bound$-de\-com\-po\-si\-tions.
\end{enumerate}
In particular, $\cT'$ is equivalent to $\cT$ iff $\cT$ is one-way definable.
\end{prop}

\begin{proof}[Proof sketch.]
Given an input word $u$, the one-way transducer $\cT'$ needs to guess a successful run $\rho$ of $\cT$ on $u$
that admits a $\bound$-decomposition. This can be done by guessing the crossing sequences of $\rho$ at each
position, 
together with a 
sequence of locations $\ell_i$ that identify the factors of a $\bound$-decomposition of $\rho$.
To check the correctness of the decomposition, $\cT'$ also needs to guess a bounded amount of information 
(words of bounded length) to reconstruct the outputs produced by the $\bound$-diagonals 
and the $\bound$-blocks. For example, while scanning a factor of the input underlying a diagonal, $\cT'$ 
can easily reproduce the outputs of the floors and the guessed outputs of factors between them. 
In a similar way, while scanning a factor of the input underlying a block, $\cT'$ can simulate 
the almost periodic output by guessing its repeating pattern and the bounded prefix and suffix
of it, and by emitting the correct amount of letters, as it is done in Example \ref{ex:running}.
In particular, one can verify that the capacity of $\cT'$ is linear in $\hmax\bound$.
Moreover, because the guessed objects are of size linear in $\hmax\bound$ and $\hmax\bound$ is a 
simple exponential in the size of $\cT$, the size of the one-way transducer $\cT'$ has
doubly exponential size in that of $\cT$.
\end{proof}


\section{The structure of two-way loops}\label{sec:loops-twoway}

While loop pumping in a sweeping transducer is rather simple, 
we need a much better understanding when it comes to pump loops of
unrestricted two-way transducers.
This section is precisely devoted to untangling the structure of two-way loops.
We will focus on specific loops, called idempotent, 
that generate repetitions with a ``nice shape'', very similar 
to loops of sweeping transducers.

We fix throughout this section a functional two-way transducer $\cT$, 
an input word $u$, and a (normalized) successful run $\rho$ of $\cT$ on $u$. 
As usual, $\hmax=2|Q|- 1$ is the maximal length of a crossing sequence 
of $\rho$, and $\cmax$ is the maximal number of letters output by a single
transition. 

\reviewTwo[inline]{
I have some questions concerning the definition of the flows and effects.

Initially you fix a transducer T, a word u and a run $\rho$ of T on u, and you define the flow corresponding to an interval I of $\rho$. On line 10, you say that “we consider effects that are not necessarily associated with intervals of specific runs”. Then you define the set of all flows, and all effects.
If I understand correctly, a flow has to correspond to an interval I, whereas an effect is a triple containing a flow (corresponding to an interval I) and two crossing sequences (that might not correspond to I).

My first question is: wouldn’t it be more natural to have either
1. also flows that do not correspond to an interval,
2. or only effects that correspond to intervals?
I realise that in case 1, defining the flows in a way for Lemma 6.5 to still hold seems complicated, but I was wondering if you had considered it.

Second, I was wondering if it wouldn’t be easier to fix the set of nodes of the flows to $\{0,\dots , 2|Q| - 2 \}$.
This would allow a smoother definition of $G \cdot G'$ (l.25): as of now, the set of nodes of $G \cdot G'$ is not stated explicitly. If I am not mistaken, it is the set of nodes that are part of G and admit an outgoing edge, or are part of $G'$ and admit an incoming edge (which is not that complicated, but wouldn't be needed at all if we always have all the nodes).
\olivier[inline]{addressed in review2-answer.txt}%
}%

\medskip
\subsection*{Flows and effects.}
We start by analyzing the shape of factors of $\rho$ intercepted
by an interval $I=[x_1,x_2]$. 
We identify four types of factors $\alpha$ intercepted by $I$ 
depending on the first location $(x,y)$ and the last location
$(x',y')$:
\begin{itemize}
\item $\alpha$ is an $\LL$-factor if $x=x'=x_1$,
\item $\alpha$ is an $\RR$-factor if $x=x'=x_2$,
\item $\alpha$ is an $\LR$-factor if $x=x_1$ and $x'=x_2$, 
\item $\alpha$ is an $\RL$-factor if $x=x_2 $ and $x'=x_1$.
\end{itemize}
In Figure \ref{fig:intercepted-factors} we see that $\a$ is an
$\LL$-factor, $\b,\delta$ are $\LR$-factors, $\z$ is an $\RR$-factor,
and $\g$ is an $\RL$-factor.

\begin{defi}\label{def:flow} 
Let $I = [x_1,x_2]$ be an interval of $\rho$ and $h_i$ the length of
the crossing sequence $\rho|x_i$, for both $i=1$ and $i=2$.

The \emph{flow} $F_I$ of $I$ is the directed graph with set of nodes
$\set{0,\dots,\max(h_1,h_2)-1}$ and set of edges consisting of all
$(y,y')$ such that there is a factor of $\rho$ intercepted by $I$ 
that starts at location $(x_i,y)$ and ends at location $(x_j,y')$,
for $i,j\in\{1,2\}$.

The \emph{effect} $E_I$ of $I$ is the triple $(F_I,c_1,c_2)$, 
where $c_i=\r|x_i$ is the crossing sequence at $x_i$.
\end{defi}

\noindent
For example, the interval $I$ of Figure \ref{fig:intercepted-factors} 
has the flow graph $0\mapsto 1\mapsto 3\mapsto 4\mapsto 2\mapsto 0$.

It is easy to see that every node of a flow $F_I$ has at most one
incoming and at most one outgoing edge. More precisely, if $y<h_1$ is
even, then it has one outgoing edge (corresponding to an $\LR$- or $\LL$-factor
intercepted by $I$), and if it is odd it has one incoming edge (corresponding
to an $\RL$- or $\LL$-factor intercepted by $I$). Similarly, if $y<h_2$ is 
even, then it has one incoming edge (corresponding to an $\LR$- or $\RR$-factor), 
and if it is odd it has one outgoing edge (corresponding to an $\RL$- or
$\RR$-factor). 

In the following we consider  effects that are not necessarily
associated with intervals of specific runs. The definition of such effects
should be clear: they are triples consisting of a graph (called flow)
and two crossing sequences of lengths $h_1,h_2 \le \hmax$, with sets of 
nodes of the form $\{0,\ldots,\max(h_1,h_2)-1\}$, 
that satisfy the in/out-degree properties stated above. 
It is convenient to distinguish the edges in a flow based on the
parity of the source and target nodes. Formally, we partition any 
flow $F$ into the following subgraphs:
\begin{itemize}
  \item $F_\LR$ consists of all edges of $F$ between pairs of even nodes,
  \item $F_\RL$ consists of all edges of $F$ between pairs of odd nodes,
  \item $F_\LL$ consists of all edges of $F$ from an even node to an odd node,
  \item $F_\RR$ consists of all edges of $F$ from an odd node to an even node.
\end{itemize}

We denote by $\cF$ (resp.~$\cE$) the set of all flows (resp.~effects)
augmented with a dummy element $\bot$. We equip both sets $\cF$ and $\cE$ with 
a semigroup structure, where the corresponding products $\circ$ and $\odot$ are 
defined below (similar definitions appear in \cite{Birget1990}). 
Later we will use the semigroup structure to identify the \emph{idempotent loops},
that play a crucial role in our characterization of one-way definability. 

\begin{defi}\label{def:product} 
For two graphs $G,G'$, we denote by $G\cdot G'$ the graph with edges of 
the form $(y,y'')$ such that $(y,y')$ is an edge of $G$ and $(y',y'')$ is an
edge of $G'$, for some node $y'$ that belongs
to both $G$ and $G'$. 
Similarly, we denote by $G^*$ the graph with edges $(y,y')$ 
such that there exists a (possibly empty) path in $G$ from $y$ to $y'$.

The product of two flows $F,F'$ is the unique flow $F\circ F'$ (if it exists) such that:
\begin{itemize}
\item $(F\circ F')_\LR = F_\LR \cdot (F'_\LL \cdot F_\RR)^* \cdot F'_\LR$,
\item $(F\circ F')_\RL = F'_\RL \cdot (F_\RR \cdot F'_\LL)^* \cdot F_\RL$,
\item $(F\circ F')_\LL = F_\LL ~\cup~  F_\LR \cdot (F'_\LL \cdot F_\RR)^* \cdot F'_\LL \cdot F_\RL$,
\item $(F\circ F')_\RR = F'_\RR ~\cup~  F'_\RL \cdot (F_\RR \cdot F'_\LL)^* \cdot F_\RR \cdot F'_\LR$.
\end{itemize}
If no flow $F\circ F'$ exists with the above properties, 
then we let $F\circ F'=\bot$.

The product of two effects $E=(F,c_1,c_2)$ and $E'=(F',c'_1,c'_2)$ is either the effect
$E\odot E' = (F\circ F',c_1,c'_2)$ or the dummy element $\bot$, depending on whether 
$F\circ F'\neq \bot$ and $c_2=c'_1$.
\end{defi}

\noindent
For example, let $F$ be the flow of interval $I$ 
in Figure \ref{fig:intercepted-factors}. Then 
$(F \circ F)_\LL=\set{0\mapsto 1, 2\mapsto 3}$, 
$(F \circ F)_\RR=\set{1\mapsto 2, 3\mapsto 4}$, and
$(F \circ F)_\LR=\set{4\mapsto 0}$ 
--- one can quickly verify this with the 
help of Figure \ref{fig:pumping-twoway}.

It is also easy to see that $(\cF,\circ)$ and $(\cE,\odot)$ are finite semigroups, and that
for every run $\r$ and every pair of consecutive intervals $I=[x_1,x_2]$ and $J=[x_2,x_3]$ of $\r$,
$F_{I\cup J} = F_I \circ F_J$ and $E_{I\cup J} = E_I \odot E_J$.
In particular, the function $E$ that associates each interval $I$ of $\rho$ with the 
corresponding effect $E_I$ can be seen as a semigroup homomorphism. 


\medskip
\subsection*{Loops and components.}
Recall that a loop is an interval $L=[x_1,x_2]$ 
with the same crossing sequences at $x_1$ and $x_2$. 
We will follow techniques similar to those presented in Section \ref{sec:combinatorics-sweeping}
to show that the outputs generated in non left-to-right manner are essentially periodic.
However, differently from the sweeping case, we will consider only special types 
of loops:

\begin{defi}\label{def:idempotent} 
A loop $L$ is \emph{idempotent} if $E_L = E_L \odot E_L$ and $E_L\neq\bot$.
\end{defi}

\noindent
For example, the interval $I$ of Figure \ref{fig:intercepted-factors} is a loop,
if one assumes that the crossing sequences at the borders of $I$ are the same. 
By comparing with Figure \ref{fig:pumping-twoway}, it is easy to see that $I$ 
is not idempotent. On the other hand, the loop consisting of 2 copies of $I$ 
is idempotent. 

\input{pumping-twoway}

As usual, given a loop $L=[x_1,x_2]$ and a number $n\in\bbN$, we can 
introduce $n$ new copies of $L$ and connect the intercepted factors 
in the obvious way.
This results in a new run $\pump_L^{n+1}(\rho)$ on the word $\pump_L^{n+1}(u)$.
Figure \ref{fig:pumping-twoway} shows how to do this for $n=1$ and $n=2$. 
Below, we analyze in detail the shape of the pumped run $\pump_L^{n+1}(\rho)$ 
(and the produced output as well) when $L$ is an {\sl idempotent} loop. 
We will focus on idempotent loops because pumping non-idempotent loops 
may induce permutations of factors that are difficult to handle.  
For example, if we consider again the non-idempotent loop $I$ to the 
left of Figure \ref{fig:pumping-twoway}, the factor of the run between 
$\beta$ and $\gamma$ (to the right of $I$, highlighted in red) precedes 
the factor between $\gamma$ and $\delta$ (to the left of $I$, again in red), 
but this ordering is reversed when a new copy of $I$ is added.

When pumping a loop $L$, subsets of factors intercepted by $L$ are glued
together to form factors intercepted by the replication of $L$. 
The notion of component introduced below identifies 
groups of factors that are glued together.

\begin{defi}\label{def:component} 
A \emph{component} of a loop $L$ is any strongly 
connected component of its flow $F_L$ 
(note that this is also a cycle, since 
every node in it has in/out-degree $1$). 

Given a component $C$, we denote by 
$\min(C)$ (resp.~$\max(C)$) the minimum (resp.~maximum) 
node in $C$.
We say that $C$ is \emph{left-to-right} (resp.~\emph{right-to-left}) 
if $\min(C)$ is even (resp., odd).

An \emph{$(L,C)$-factor} is a factor of the run that is 
intercepted by $L$ and that corresponds to an edge of $C$.
\end{defi}

\noindent
We will usually list the $(L,C)$-factors based on their order of occurrence in the run.
For example, the loop $I$ of Figure \ref{fig:pumping-twoway} contains
a single component $C=0\mapsto 1\mapsto 3\mapsto 4\mapsto 2\mapsto 0$ 
which is left-to-right. 
Another example is given in Figure \ref{fig:many-components}, where the 
loop $L$ has three components $C_1,C_2,C_3$ (colored in blue, red, 
and green, respectively):
$\a_1,\a_2,\a_3$ are the $(L,C_1)$-factors, $\b_1,\b_2,\b_3$ are
the $(L,C_2)$-factors, and $\g_1$ is the unique $(L,C_3)$-factor.

\input{many-components}

\medskip
Below, we show that the levels of each component of a loop (not necessarily idempotent)
form an interval.

\begin{lem}\label{lem:component}
Let $C$ be a component of a loop $L=[x_1,x_2]$. 
The nodes of $C$ are precisely the levels in the interval $[\min(C),\max(C)]$. 
Moreover, if $C$ is left-to-right (resp.~right-to-left), then $\max(C)$
is the smallest level $\ge \min(C)$ 
such that between $(x_1,\min(C))$ and $(x_2,\max(C))$ (resp.~$(x_2,\min(C))$ and $(x_1,\max(C))$)
there are equally many $\LL$-factors and $\RR$-factors intercepted by $L$.
%
\end{lem}

\begin{proof}[Proof idea]
The proof of this lemma is rather technical and deferred to
Appendix~\ref{app:proof-component}, since the lemma is not at the core of the proof of the main
result. Let us first note that with the definition of $\max(C)$ stated
in the lemma it is rather easy to see that the interval
$[\min(C),\max(C)]$ is a union of cycles (i.e., components). This
can be shown by arguing that every node in $[\min(C),\max(C)]$ has
in-degree and out-degree one. What is much less obvious is that
$[\min(C),\max(C)]$ is connected, thus consists of a single cycle.

\begin{figure}[!t]
\centering
\begin{tikzpicture}[baseline=0, inner sep=0, outer sep=0, minimum size=0pt, scale=0.32]
  \tikzstyle{dot} = [draw, circle, fill=white, minimum size=4pt]
  \tikzstyle{fulldot} = [draw, circle, fill=black, minimum size=4pt]
  \tikzstyle{grayfactor} = [->, shorten >=1pt, rounded corners=6, gray, thin, dotted]
  \tikzstyle{factor} = [->, shorten >=1pt, rounded corners=6]
  \tikzstyle{dotfactor} = [->, shorten >=1pt, dotted, rounded corners=6]
  \tikzstyle{fullfactor} = [->, >=stealth, shorten >=1pt, very thick, rounded corners=6]
  \tikzstyle{dotfullfactor} = [->, >=stealth, shorten >=1pt, dotted, very thick, rounded corners=6]

\begin{scope}[yscale=1.3]
  \fill [pattern=north east lines, pattern color=gray!25]
        (0,-2) rectangle (6,11);
  \draw [dashed, thin, gray] (0,-2) -- (0,11);
  \draw [dashed, thin, gray] (6,-2) -- (6,11);
  \draw [gray] (0,-2.25) -- (0,-2.5) -- (6,-2.5) -- (6,-2.25);
  \draw [gray] (3,-2.75) node [below] {\footnotesize $L$};
 
  \draw (3,-0.8) node () {$\vdots$};
  \draw (3,10.5) node () {$\vdots$};

  \draw (0,0) node [dot] (node1) {};
  \draw (2,0) node (node2) {};
  \draw (2,1) node (node3) {};
  \draw (0,1) node [dot] (node4) {};

  \draw (0,2) node [dot] (node5) {};
  \draw (2,2) node (node6) {};
  \draw (2,3) node (node7) {};
  \draw (0,3) node [dot] (node8) {};

  \draw (0,4) node [dot] (node9) {};
  \draw (2,4) node (node10) {};
  \draw (2,5) node (node11) {};
  \draw (0,5) node [dot] (node12) {};

  \draw (0,6) node [fulldot] (node13) {};
  \draw (6,6) node [fulldot] (node14) {};

  \draw (6,7) node [dot] (node15) {};
  \draw (4,7) node (node16) {};
  \draw (4,8) node (node17) {};
  \draw (6,8) node [dot] (node18) {};


  \draw (6,9) node [fulldot] (node23) {};
  \draw (0,9) node [fulldot] (node24) {};

  \draw (node13) node [left = 2mm, gray] {\footnotesize $\ort y_i$};
  \draw (node14) node [right = 2mm, gray] {\footnotesize $\olft y_{i-1}+1$};
  \draw (node23) node [right = 2mm, gray] {\footnotesize $\olft y_i$};
  \draw (node24) node [left = 2mm, gray] {\footnotesize $\ort y_i+1$};

  \draw [fullfactor] (node1) -- (node2.center) -- (node3.center) -- (node4); 
  \draw [fullfactor] (node5) -- (node6.center) -- (node7.center) -- (node8); 
  \draw [fullfactor, dashed, gray] (node9) -- (node10.center) -- (node11.center) -- (node12); 

  \draw [fullfactor, blue] (node13) -- (node14);

  \draw [fullfactor] (node15) -- (node16.center) -- (node17.center) -- (node18); 

  \draw [fullfactor, red] (node23) -- (node24);
\end{scope}

\begin{scope}[yscale=1.5, xshift=20cm, yshift=-0.3cm]
  \draw [dashed, thin, gray] (0,-0.5) -- (0,8.5);
  \draw [gray] (0,-2) node [below] {\footnotesize $F_L$};
    
  \draw (0,-1) node () {$\vdots$};
  \draw (0,9.45) node () {$\vdots$};

  \draw (0,0) node [dot] (node1) {};
  \draw (2,0) node (node2) {};
  \draw (2,1) node (node3) {};
  
  \draw (0,1) node [dot] (node4) {};
  \draw (-2,1) node (node5) {};
  \draw (-2,2) node (node6) {};
  
  \draw (0,2) node [dot] (node7) {};
  \draw (2,2) node (node8) {};
  \draw (2,3) node (node9) {};

  \draw (0,4.65) node [dot] (node10) {};
  \draw (2,4.65) node (node11) {};
  \draw (2,5.65) node (node12) {};

  \draw (0,5.65) node [dot] (node13) {};

  \draw (0,3) node [fulldot] (node16) {};
  \draw (-3.5,3) node (node17) {};
  \draw (-3.5,8) node (node18) {};
  \draw (-2,8) node (node19) {};
  \draw (0,8) node [fulldot] (node20) {};

  \draw (0,7) node [fulldot] (node21) {};
  \draw (4,7) node (node22) {};
  \draw (4,-1) node (node23) {};
  \draw (2,-1) node (node24) {};
  \draw (0,0) node [fulldot] (node25) {};

  \draw (node1) node [left = 1mm, gray] {\footnotesize $\olft y_{i-1}+1$};
  \draw (node16) node [above right = 1.5mm, gray] {\footnotesize $\olft y_i$};
  \draw (node21) node [left = 1mm, gray] {\footnotesize $\ort y_i$};
  \draw (node20) node [right = 1mm, gray] {\footnotesize $\ort y_i+1$};

  \draw [fullfactor] (node1) -- (node2.center) -- (node3.center) -- (node4); 
  \draw [fullfactor] (node4) -- (node5.center) -- (node6.center) -- (node7); 
  \draw [fullfactor] (node7) -- (node8.center) -- (node9.center) -- (node16); 
  \draw [fullfactor, dashed, gray] (node10) -- (node11.center) -- (node12.center) -- (node13); 
  \draw [fullfactor, red] (node16) -- (node17.center) -- 
                          (node18.center) -- (node19.center) -- (node20);
  \draw [fullfactor, blue] (node21) -- (node22.center) -- 
                           (node23.center) -- (node24.center) -- (node25);  
\end{scope}
\end{tikzpicture}
\caption{Some factors intercepted by $L$ and the corresponding edges in the flow.}%
\label{fig:edges} 
\end{figure}
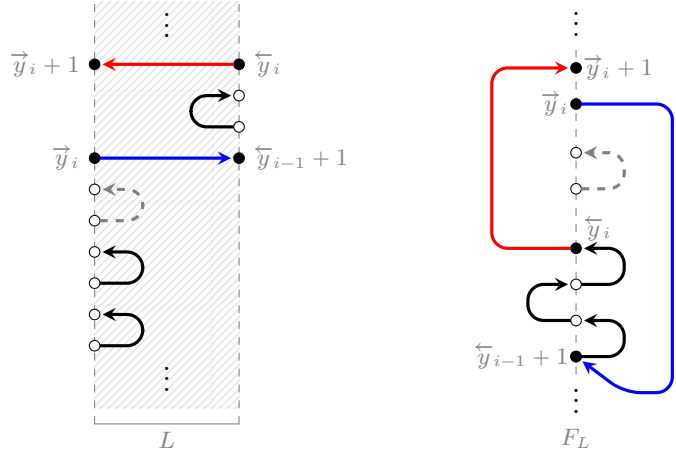



The crux is thus to show that the nodes visited by every cycle of the flow
(or, equally,  
every component) form an interval. For this, we use an induction based 
on portions of the cycle, namely, on {\sl paths} of the flow. 
The difficulty underlying the formalization of the inductive invariant comes from the
fact that, differently from cycles, paths of a flow may visit sets of levels that do 
not form intervals.
An example is given in Figure \ref{fig:edges}, which represents some edges of a flow
forming a path from $\ort y_i$ to $\ort y_i +1$ and covering a non-convex set of nodes: 
note that there could be a large gap between the nodes $\ort y_i$ and $\olft y_i$ due
to the unbalanced numbers of $\LL$-factors and $\RR$-factors below $\olft y_i$.

Essentially, the first part of the proof of the lemma amounts at identifying the sources 
$\ort y_i$ (resp.~$\olft y_i$) of the $\LR$-factors (resp.~$\RL$-factors), and at showing 
that the latter factors are precisely of the form $\ort y_i \rightarrow \olft y_{i-1}+1$
(resp.~$\olft y_i \rightarrow \ort y_i +1$). 
Once these nodes are identified, we show by induction on $i$ that every two 
consecutive nodes $\ort y_i$ and $\ort y_i +1$ must be connected by a path whose 
intermediate nodes form the interval $[\olft y_{i-1}+1,\olft y_i]$. 
Finally, we argue that every cycle (or component) $C$ visits all and only 
the nodes in the interval $[\min(C),\max(C)]$.
\end{proof}

The next lemma describes the precise shape and order of the intercepted factors 
when the loop $L$ is idempotent. 

\begin{lem}\label{lem:component2}
If $C$ is a left-to-right (resp.~right-to-left) component 
of an {\sl idempotent} loop $L$, then the $(L,C)$-factors are in the following order: 
$k$ $\LL$-factors (resp.~$\RR$-factors), followed by one $\LR$-factor (resp.~$\RL$-factor), 
followed by $k$ $\RR$-factors (resp.~$\LL$-factors), for some $k \ge 0$. 
\end{lem}

\begin{proof}
Suppose that $C$ is a left-to-right component of $L$.
We show by way of contradiction that $C$ has only one $\LR$-factor
and no $\RL$-factor. By Lemma~\ref{lem:component} this will yield
the claimed shape. Figure~\ref{fig:notidem} can be used as a reference
example for the arguments that follow.

We begin by listing the $(L,C)$-factors.
As usual, we order them based on their occurrences in the run $\rho$.
Let $\gamma$ be the first $(L,C)$-factor that is not an $\LL$-factor, 
and let $\beta_1,\dots,\beta_k$ be the $(L,C)$-factors that precede $\gamma$ 
(these are all $\LL$-factors).
Because $\gamma$ starts at an even level, it must be an $\LR$-factor.
Suppose that there is another $(L,C)$-factor, say $\zeta$, that comes 
after $\gamma$ and it is neither an $\RR$-factor nor an $\LL$-factor. 
Because $\zeta$ starts at an odd level, it must be an $\RL$-factor. 
Further let $\delta_1,\dots,\delta_{k'}$ be the intercepted $\RR$-factors 
that occur between $\gamma$ and $\zeta$.
We claim that $k'<k$, namely, that the number of $\RR$-factors between 
$\gamma$ and $\zeta$ is strictly less than the number of $\LL$-factors 
before $\gamma$. Indeed, if this were not the case, then, by Lemma 
\ref{lem:component}, the level where $\zeta$ starts would not belong to
the component $C$.

Now, consider the pumped run $\rho'=\pump_L^2(\rho)$, obtained by adding a new
copy of $L$. Let $L'$ be the loop of $\rho'$ obtained from the union
of $L$ and its copy. Since $L$ is idempotent, the components of $L$ are isomorphic
to the components of $L'$. In particular, we can denote by $C'$ the component
of $L'$ that is isomorphic to $C$. 
Let us consider the $(L',C')$-factors of $\rho'$. The first $k$ factors
are isomorphic to the $k$ $\LL$-factors $\beta_1,\ldots,\beta_k$ from $\rho$.
However, the $(k+1)$-th element has a different shape: it is isomorphic to
$\gamma~\beta_1~\delta_1~\beta_2~\cdots~\delta_{k'}~\beta_{k'+1}~\zeta$,
and in particular it is an $\LL$-factor.
This implies that the $(k+1)$-th edge of $C'$ is of the form $(y,y+1)$,
while the $(k+1)$-th edge of $C$ is of the form $(y,y-2k)$. 
This contradiction comes from having assumed the existence of 
the $\RL$-factor $\zeta$, and is illustrated in Figure~\ref{fig:notidem}.
\end{proof}

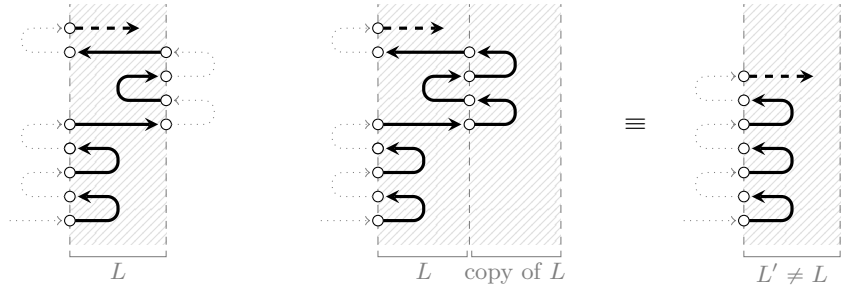
\begin{figure}[!t]
\centering
\begin{tikzpicture}[baseline=0, inner sep=0, outer sep=0, minimum size=0pt, scale=0.32]
  \tikzstyle{dot} = [draw, circle, fill=white, minimum size=4pt]
  \tikzstyle{fulldot} = [draw, circle, fill=black, minimum size=4pt]
  \tikzstyle{grayfactor} = [->, shorten >=1pt, rounded corners=4, gray, thin, dotted]
  \tikzstyle{factor} = [->, shorten >=1pt, rounded corners=4]
  \tikzstyle{dotfactor} = [->, shorten >=1pt, dotted, rounded corners=4]
  \tikzstyle{fullfactor} = [->, >=stealth, shorten >=1pt, very thick, rounded corners=4]
  \tikzstyle{dotfullfactor} = [->, >=stealth, shorten >=1pt, dotted, very thick, rounded corners=4]

\begin{scope}[xshift=-10cm]
  \fill [pattern=north east lines, pattern color=gray!25]
        (4,-1) rectangle (8,9);
  \draw [dashed, thin, gray] (4,-1) -- (4,9);
  \draw [dashed, thin, gray] (8,-1) -- (8,9);
  \draw [gray] (4,-1.25) -- (4,-1.5) -- (8,-1.5) -- (8,-1.25);
  \draw [gray] (6,-1.75) node [below] {\footnotesize $L$};

  \draw (1.5,0) node (node0) {};
  \draw (4,0) node [dot] (node1) {};
  \draw (6,0) node (node2) {};
  \draw (6,1) node (node3) {};
  \draw (4,1) node [dot] (node4) {};
  \draw (2,1) node (node5) {};
  \draw (2,2) node (node6) {};
  \draw (4,2) node [dot] (node7) {};
  \draw (4,3) node [dot] (node8) {};
  \draw (10,2) node (node9) {};
  \draw (10,3) node (node10) {};
  \draw (4,3) node [dot] (node12) {};
  \draw (2,3) node (node13) {};
  \draw (2,4) node (node14) {};
  \draw (4,4) node [dot] (node15) {};
  \draw (8,4) node [dot] (node16) {};
  \draw (10,4) node (node17) {};
  \draw (10,5) node (node18) {};
  \draw (8,5) node [dot] (node19) {};
  \draw (6,5) node (node20) {};
  \draw (6,6) node (node21) {};
  \draw (8,6) node [dot] (node22) {};
  \draw (8,7) node  [dot] (node23) {};
  \draw (10,6) node (node24) {};
  \draw (10,7) node (node25) {} ;
  \draw (4,7) node [dot] (node26) {};
  \draw (4,8) node [dot] (node27) {};
  \draw (2,7) node (node28) {};
  \draw (2,8) node (node29) {};
  \draw (7,8) node (node30) {};
  \draw (6,2) node (node77) {};
  \draw (6,3) node (node88) {};

  \draw [grayfactor] 
  (node0) -- (node1);
  \draw [fullfactor] (node1) -- (node2.center) -- 
                     (node3.center) -- (node4); 
  \draw [fullfactor] (node7) -- (node77.center) -- 
  (node88.center) -- (node8);
  
  \draw [grayfactor] (node4) -- (node5.center) -- (node6.center) -- (node7); 
  
  \draw [grayfactor] (node12) -- (node13.center) -- (node14.center) -- (node15);
  \draw [fullfactor] (node15) --  (node16);
  \draw [grayfactor] (node16) -- (node17.center) -- (node18.center) -- (node19);
  \draw [fullfactor] (node19) -- (node20.center) -- 
                     (node21.center) -- (node22);
  \draw [grayfactor] (node22) -- (node24.center) -- (node25.center) -- (node23);
  \draw [fullfactor] (node23) -- (node26) ;
  \draw [grayfactor] (node26) -- (node28.center) -- (node29.center) -- (node27);
  \draw [fullfactor, dashed] (node27) -- (node30) ;
\end{scope}

\begin{scope}[xshift=18cm]
  \fill [pattern=north east lines, pattern color=gray!25]
        (4,-1) rectangle (8,9);
  \draw [dashed, thin, gray] (4,-1) -- (4,9);
  \draw [dashed, thin, gray] (8,-1) -- (8,9);
  \draw [gray] (4,-1.25) -- (4,-1.5) -- (8,-1.5) -- (8,-1.25);
  \draw [gray] (6,-1.75) node [below] {\footnotesize $L' \neq L$};

  \draw (1.5,0) node (node0) {};
  \draw (4,0) node [dot] (node1) {};
  \draw (6,0) node (node2) {};
  \draw (6,1) node (node3) {};
  \draw (4,1) node [dot] (node4) {};
  \draw (2,1) node (node5) {};
  \draw (2,2) node (node6) {};
  \draw (4,2) node [dot] (node7) {};
  \draw (4,3) node [dot] (node8) {};
  \draw (4,4) node [dot] (node9) {};
  \draw (4,5) node [dot] (node10) {};
  \draw (6,4) node  (node11) {};
  \draw (6,5) node  (node12) {};
  \draw (2,3) node (node13) {};
  \draw (2,4) node (node14) {};
  \draw (4,6) node [dot] (node15) {};
  \draw (2,5) node (node17) {};
  \draw (2,6) node (node18) {};
  \draw (6,5) node (node20) {};
  \draw (6,6) node (node21) {};
  \draw (10.5,6) node (node23) {};
  \draw (6,2) node (node77) {};
  \draw (6,3) node (node88) {};

\draw (-0.5,4) node { $\equiv$}; 
  \draw [grayfactor] 
  (node0) -- (node1);
  \draw [fullfactor] (node1) -- (node2.center) -- 
                     (node3.center) -- (node4); 
  \draw [grayfactor] (node4) -- (node5.center) -- (node6.center) -- (node7); 
  \draw [fullfactor] (node7) -- (node77.center) --
  (node88.center) -- (node8);
  \draw [grayfactor] (node8) -- (node13.center) -- (node14.center) -- (node9);
  \draw [fullfactor] (node9) -- (node11.center) -- (node12.center) -- (node10);
  \draw [grayfactor] (node10) -- (node17.center) -- (node18.center) -- (node15);
  \draw [fullfactor, dashed] (node15) -- (7,6);
\end{scope}

\begin{scope}[xshift=3cm, xscale=0.95]
  \fill [pattern=north east lines, pattern color=gray!25]
        (4,-1) rectangle (8,9);
  \fill [pattern=north east lines, pattern color=gray!25]
        (8,-1) rectangle (12,9);
  \draw [dashed, thin, gray] (4,-1) -- (4,9);
  \draw [dashed, thin, gray] (8,-1) -- (8,9);
  \draw [dashed, thin, gray] (12,-1) -- (12,9);
  \draw [gray] (4,-1.25) -- (4,-1.5) -- (7.9,-1.5) -- (7.9,-1.25);
  \draw [gray] (8.1,-1.25) -- (8.1,-1.5) -- (12,-1.5) -- (12,-1.25);
  \draw [gray] (6,-1.75) node [below] {\footnotesize $L$};
  \draw [gray] (10,-1.75) node [below] {\footnotesize copy of $L$};

  \draw (1.5,0) node (node0) {};
  \draw (4,0) node [dot] (node1) {};
  \draw (6,0) node (node2) {};
  \draw (6,1) node (node3) {};
  \draw (4,1) node [dot] (node4) {};
  \draw (2,1) node (node5) {};
  \draw (2,2) node (node6) {};
  \draw (4,2) node [dot] (node7) {};
  \draw (4,3) node [dot] (node8) {};
  \draw (10,2) node (node9) {};
  \draw (10,3) node (node10) {};
  \draw (4,3) node [dot] (node12) {};
  \draw (2,3) node (node13) {};
  \draw (2,4) node (node14) {};
  \draw (4,4) node [dot] (node15) {};
  \draw (8,4) node [dot] (node16) {};
  \draw (10,4) node (node17) {};
  \draw (10,5) node (node18) {};
  \draw (8,5) node [dot] (node19) {};
  \draw (6,5) node (node20) {};
  \draw (6,6) node (node21) {};
  \draw (8,6) node [dot] (node22) {};
  \draw (8,7) node  [dot] (node23) {};
  \draw (10,6) node (node24) {};
  \draw (10,7) node (node25) {} ;
  \draw (4,7) node [dot] (node26) {};
  \draw (4,8) node [dot] (node27) {};
  \draw (2,7) node (node28) {};
  \draw (2,8) node (node29) {};
  \draw (7,8) node (node30) {};
  \draw (6,2) node (node77) {};
  \draw (6,3) node (node88) {};

  \draw [grayfactor] 
  (node0) -- (node1);
  \draw [fullfactor] (node1) -- (node2.center) -- 
                     (node3.center) -- (node4); 
  \draw [fullfactor] (node7) -- (node77.center) -- 
  (node88.center) -- (node8);
  
  \draw [grayfactor] (node4) -- (node5.center) -- (node6.center) -- (node7); 
  
  \draw [grayfactor] (node12) -- (node13.center) -- (node14.center) -- (node15);
  \draw [fullfactor] (node15) --  (node16);
  \draw [fullfactor] (node16) -- (node17.center) -- (node18.center) -- (node19);
  \draw [fullfactor] (node19) -- (node20.center) -- 
                     (node21.center) -- (node22);
  \draw [fullfactor] (node22) -- (node24.center) -- (node25.center) -- (node23);
  \draw [fullfactor] (node23) -- (node26) ;
  \draw [grayfactor] (node26) -- (node28.center) -- (node29.center) -- (node27);
  \draw [fullfactor, dashed] (node27) -- (node30);
\end{scope}

\end{tikzpicture}
\caption{Pumping a loop $L$ with a wrong shape and showing it is not
  idempotent.}\label{fig:notidem} 
\end{figure}


\begin{rem}
Note that every loop in the sweeping case  is
idempotent. Moreover, the $(L,C)$-factors are precisely the factors
intercepted by the loop $L$.
\end{rem}

\medskip
\subsection*{Pumping idempotent loops.}
To describe in a formal way the run obtained by pumping an idempotent loop,
we need to generalize the notion of anchor point in the two-way case
(the reader may compare this with the analogous definitions in 
Section~\ref{sec:combinatorics-sweeping} for the sweeping case). 
Intuitively, the anchor point of a component $C$ of an idempotent loop $L$ 
is the source location of the unique $\LR$- or $\RL$-factor intercepted by 
$L$ that corresponds to an edge of $C$ (recall Lemma~\ref{lem:component2}):

\begin{defi}\label{def:anchor} 
Let $C$ be a component of an idempotent loop $L = [x_1,x_2]$.
The \emph{anchor point} of $C$ inside $L$, denoted%
\footnote{In denoting the anchor point --- and similarly the trace --- of a component $C$ 
          inside a loop $L$, we omit the annotation specifying $L$, since this is 
          often understood from the context.}
$\an{C}$, is either the location $(x_1,\max(C))$ or the location 
$(x_2,\max(C))$, depending on whether $C$ is left-to-right or right-to-left.
\end{defi}

\noindent
We will usually depict anchor points by black circles (like, for instance, in Figure \ref{fig:many-components}).

It is also convenient to redefine the notation $\tr{\ell}$ for representing 
an appropriate sequence of transitions associated with each anchor point 
$\ell$ of an idempotent loop:

\begin{defi}\label{def:trace} 
Let $C$ be a component of some idempotent loop $L$, let $\ell=\an{C}$
be the anchor point of $C$ inside $L$, and let
$i_0 \mapsto i_1 \mapsto i_2 \mapsto \dots \mapsto i_k \mapsto i_{k+1}$ 
be a cycle of $C$, where $i_0=i_{k+1}=\max(C)$. 
For every $j=0,\dots,k$, further let $\beta_j$ be the factor intercepted 
by $L$ that corresponds to the edge $i_j \mapsto i_{j+1}$ of $C$.

The \emph{trace} of $\ell$ inside $L$ is the run $\tr{\ell} = \beta_0 ~ \beta_1 ~ \cdots ~ \beta_k$.
\end{defi}

\noindent
Note that $\tr{\ell}$ is not necessarily a factor of the original run
$\rho$. However, $\tr{\ell}$ is indeed a run, since $L$ is a loop and
the factors $\b_i$ are concatenated according to the flow. As we will see below, 
$\tr{\ell}$ will appear as (iterated) factor of the pumped version of $\rho$, 
where the loop $L$ is iterated.

As an example, by referring again to the components $C_1,C_2,C_3$ of 
Figure~\ref{fig:many-components}, we have the following traces:
$\tr{\an{C_1}}=\alpha_2\:\alpha_1\:\alpha_3$,
$\tr{\an{C_2}}=\beta_2\:\beta_1\:\beta_3$, and 
$\tr{\an{C_3}}=\gamma_1$. 

The next proposition shows the effect of pumping idempotent loops. The
reader can note the similarity with the sweeping case.

\begin{prop}\label{prop:pumping-twoway}
Let $L$ be an idempotent loop of $\rho$ with components $C_1,\dots,C_k$, 
listed according to the order of their anchor points:
$\ell_1=\an{C_1} \lesstime \cdots \lesstime \ell_k=\an{C_k}$. 
For all $n\in\bbN$, we have 
\[
  \pump_L^{n+1}(\rho) ~=~ 
  \rho_0 ~ \tr{\ell_1}^n ~ \rho_1 ~ \cdots ~ \rho_{k-1} ~ \tr{\ell_k}^n ~ \rho_k
\]
where 
\begin{itemize}
  \item $\rho_0$ is the prefix of $\rho$ that ends at the first anchor point $\ell_1$, 
  \item $\rho_k$ is the suffix of $\rho$ that starts at the last anchor point $\ell_k$,
  \item $\rho_i$ is the factor $\rho[\ell_i,\ell_{i+1}]$, for all $1\le i<k$.
\end{itemize}
\end{prop}

\begin{proof}
Along the proof we sometimes refer to Figure \ref{fig:many-components} to 
ease the intuition of some definitions and arguments.
For example, in the left hand-side of Figure \ref{fig:many-components}, 
the run $\rho_0$ goes until the first location marked by a black circle;
the runs $\rho_1$ and $\rho_2$, resp., are between the first and the 
second black dot, and the second and third black dot; finally, $\rho_3$ 
is the suffix starting at the last black dot. The pumped run 
$\pump_L^{n+1}(\rho)$ for $n=2$ is depicted to the right of the figure.

Let $L=[x_1,x_2]$ be an idempotent loop and, for all $i=0,\dots,n$, let 
$L'_i=[x'_i,x'_{i+1}]$ be the $i$-th copy of the loop $L$ in the pumped run
$\rho'=\pump_L^{n+1}(\rho)$, where $x'_i = x_1 + i\cdot (x_2-x_1)$
(the ``$0$-th copy of $L$'' is the loop $L$ itself). 
Further let $L'=L'_0\cup\dots\cup L'_n = [x'_0,x'_{n+1}]$, that is, 
$L'$ is the loop of $\rho'$ that spans across the $n+1$ occurrences of $L$.
As $L$ is idempotent, the loops $L'_0,\dots,L'_n$ and $L'$ have all the 
same effect as $L$.
In particular, the components of $L'_0,\dots,L'_n$, and $L'$ are isomorphic 
to and in same order as those of $L$.
We denote these components by $C_1,\dots,C_k$.

We let $\ell_j=\an{C_j}$ be the anchor point of each component $C_j$ inside 
the loop $L$ of $\rho$
(these locations are marked by black dots in the left hand-side 
of Figure \ref{fig:many-components}).
Similarly, we let $\ell'_{i,j}$ (resp.~$\ell'_j$) 
be the anchor point of $C_j$ inside the loop $L'_i$ (resp.~$L'$).
From Definition \ref{def:anchor}, we have that either $\ell'_j=\ell'_{1,j}$ or $\ell'_j=\ell'_{n,j}$, 
depending on whether $C_j$ is left-to-right or right-to-left (or, equally, on whether $j$ is odd or even).

Now, let us consider the factorization of the pumped run $\rho'$ 
induced by the locations $\ell'_{i,j}$, for all $i=0,\dots,n$ and for $j=1,\dots,k$ 
(these locations are marked by black dots in the right hand-side of the figure).
By construction, the prefix of $\rho'$ that ends at location $\ell'_{0,1}$ 
coincides with the prefix of $\rho$ that ends at $\ell_1$, 
i.e.~$\rho_0$ in the statement of the proposition.
Similarly, the suffix of $\rho'$ that starts at location $\ell'_{n,k}$ is isomorphic 
to the suffix of $\rho$ that starts at $\ell_k$, i.e. $\rho_k$ in the statement.
Moreover, for all odd (resp.~even) indices $j$, the factor
$\rho'[\ell'_{n,j},\ell'_{n,j+1}]$ (resp.~$\rho'[\ell_{0,j},\ell_{0,j+1}]$) is isomorphic 
to $\rho[\ell_j,\ell_{j+1}]$, i.e.~the $\rho_j$ of the statement. 

The remaining factors of $\rho'$ are those delimited by the pairs of locations 
$\ell'_{i,j}$ and $\ell'_{i+1,j}$, for all $i=0,\dots,n-1$ and all $j=1,\dots,k$. 
Consider one such factor $\rho'[\ell'_{i,j},\ell'_{i+1,j}]$, 
and assume that the index $j$ is odd (the case of an even $j$ is similar). 
This factor can be seen as a concatenation of factors intercepted by $L$ 
that correspond to edges of $C_j$ inside $L'_i$. 
More precisely, $\rho'[\ell'_{i,j},\ell'_{i+1,j}]$ is obtained by concatenating 
the unique $\LR$-factor of $C_j$ --- recall that by Lemma \ref{lem:component2} 
there is exactly one such factor --- with an interleaving of the $\LL$-factors 
and the $\RR$-factors of $C_j$. 
As the components are the same for all $L'_i$'s and for $L$, this corresponds 
precisely to the trace $\tr{\ell_j}$ (cf.~Definition \ref{def:trace}).
Now that we know that $\rho'[\ell'_{i,j},\ell'_{i+1,j}]$ is isomorphic to $\tr{\ell_j}$, 
we can conclude that
$\rho'[\ell'_{0,j},\ell'_{n,j}] \:=\: 
 \rho'[\ell'_{0,j},\ell'_{1,j}] ~ \dots ~ \rho'[\ell'_{n-1,j},\ell'_{n,j}]$
is isomorphic to $\tr{\ell_j}^n$. 
\end{proof}


\section{Combinatorics in the two-way case}\label{sec:combinatorics-twoway}

In this section we develop the main combinatorial techniques required
in the general case.
In particular, we will show how to derive the existence of idempotent 
loops with bounded outputs using Ramsey-based arguments, and we will
use this to derive periodicity properties for the outputs produced
between inversions.

As usual, $\rho$ is a fixed successful run of $\cT$ on some input word $u$.

\medskip
\subsection*{Ramsey-type arguments.}
We start with a technique used for bounding the lengths of the outputs 
of certain factors, or subsequences of a two-way run. This technique
is a Ramsey-type argument, more precisely it relies on
Simon's ``factorization forest'' theorem~\cite{factorization_forests,factorization_forests_for_words_paper},
which is recalled below. The classical version of Ramsey theorem would
yield a similar result, but without the tight bounds that we get here.

Let $X$ be a set of positions of $\rho$.
A \emph{factorization forest} for $X$ is an unranked tree, where the nodes are 
intervals $I$ with endpoints in $X$ and labeled with the corresponding effect $E_I$, 
the ancestor relation is given by the containment order on intervals, the leaves are the 
minimal intervals $[x_1,x_2]$, with $x_2$ successor of $x_1$ in $X$, and for every 
internal node $I$ with children $J_1,\dots,J_k$, we have:
\begin{itemize}
  \item $I=J_1\cup\dots\cup J_k$, 
  \item $E_I = E_{J_1}\odot\dots\odot E_{J_k}$, 
  \item if $k>2$, then $E_I = E_{J_1} = \dots = E_{J_k}$ 
        is an idempotent of the semigroup $(\cE,\odot)$.
\end{itemize}

Recall that in a normalized run there are at most $|Q|^{\hmax}$ distinct 
crossing sequences. Moreover, a flow contains at most $\hmax$ edges, 
and each edge has one of the 4 possible types $\LL,\LR,\RL,\RR$, so they are at most $4^{\hmax}$ different flows.
Hence, the effect semigroup $(\cE,\odot)$ has size at most 
$\emax=4^{\hmax}\cdot|Q|^{2\hmax}\cdot|Q|^{2\hmax} =(2|Q|)^{2\hmax}$.
\reviewTwo[inline]{the way you bound the number of flows could be explained more precisely.
Here you just state that there is at most H edges, and that each edge has one of the possible 4 types.
Then, you bound the number of flows by $4^{H}$, which seems to correspond to any choice of types for the H edges.
However, as you state, there is at most H edges, but there could be less.}%
\felix{done}%
Further recall that $\cmax$ is the maximum number of letters output by a 
single transition of $\cT$.
Like we did in the sweeping case, we define the constant
${\boldsymbol{B = \cmax \cdot \hmax \cdot (2^{3\emax}+4) + 4\cmax}}$
that will be used to bound the lengths of some outputs of $\cT$. 
Note that now $\boldsymbol{B}$ is doubly exponential with respect 
to the size of $\cT$, due to the size of the effect semigroup.
\felix{added this explanation, it could be misleading otherwise}%

\begin{thm}[Factorization forest theorem \cite{factorization_forests_for_words_paper,factorization_forests}]%
\label{th:simon}
For every set $X$ of positions of $\rho$, there is a factorization forest for $X$
of height at most $3\emax$.
\end{thm}

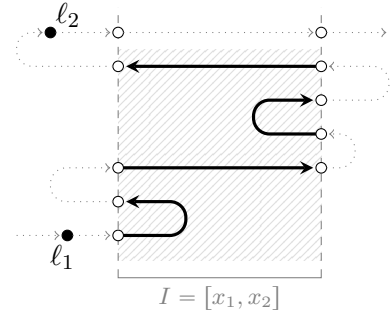
\begin{wrapfigure}{r}{5cm}
\vspace{-4mm}
\begin{tikzpicture}[baseline=0, inner sep=0, outer sep=0, minimum size=0pt, scale=0.45]
  \tikzstyle{dot} = [draw, circle, fill=white, minimum size=4pt]
  \tikzstyle{fulldot} = [draw, circle, fill=black, minimum size=4pt]
  \tikzstyle{grayfactor} = [->, shorten >=1pt, rounded corners=6, gray, thin, dotted]
  \tikzstyle{factor} = [->, shorten >=1pt, rounded corners=6]
  \tikzstyle{dotfactor} = [->, shorten >=1pt, dotted, rounded corners=6]
  \tikzstyle{fullfactor} = [->, >=stealth, shorten >=1pt, very thick, rounded corners=6]
  \tikzstyle{dotfullfactor} = [->, >=stealth, shorten >=1pt, dotted, very thick, rounded corners=6]

  \fill [pattern=north east lines, pattern color=gray!25]
        (6,-0.75) rectangle (12,5.5);
  \draw [dashed, thin, gray] (6,-0.75) -- (6,6.75);
  \draw [dashed, thin, gray] (12,-0.75) -- (12,6.75);
  \draw [gray] (6,-1) -- (6,-1.25) -- (12,-1.25) -- (12,-1);
  \draw [gray] (9,-1.5) node [below] {\footnotesize $I=[x_1,x_2]$};

  \draw (3,0) node (node0) {};
  \draw (4.5,0) node [fulldot] (node0') {};
  \draw (6,0) node [dot] (node1) {};
  \draw (8,0) node (node2) {};
  \draw (8,1) node (node3) {};
  \draw (6,1) node [dot] (node4) {};
  \draw (4,1) node (node5) {};
  \draw (4,2) node (node6) {};
  \draw (6,2) node [dot] (node7) {};
  \draw (12,2) node [dot] (node8) {};
  \draw (13,2) node (node9) {};
  \draw (13,3) node (node10) {};
  \draw (12,3) node [dot] (node11) {};
  \draw (10,3) node (node12) {};
  \draw (10,4) node (node13) {};
  \draw (12,4) node [dot](node14) {};
  \draw (14,4) node (node15) {};
  \draw (14,5) node (node16) {};
  \draw (12,5) node [dot] (node17) {};
  \draw (6,5) node [dot] (node18) {};
  \draw (3,5) node (node19) {};
  \draw (3,6) node (node20) {};
  \draw (4,6) node [fulldot] (node20') {};
  \draw (6,6) node [dot] (node21) {};
  \draw (12,6) node [dot] (node22) {};
  \draw (14,6) node (node23) {};

  \draw [grayfactor] (node0) -- (node0');
  \draw [grayfactor] (node0') -- (node1);
  \draw [fullfactor] (node1) -- (node2.center) -- (node3.center) -- (node4); 
  \draw [grayfactor] (node4) -- (node5.center) -- (node6.center) -- (node7); 
  \draw [fullfactor] (node7) -- (node8);
  \draw [grayfactor] (node8) -- (node9.center) -- (node10.center) -- (node11); 
  \draw [fullfactor] (node11) -- (node12.center) -- (node13.center) -- (node14); 
  \draw [grayfactor] (node14) -- (node15.center) -- (node16.center) -- (node17); 
  \draw [fullfactor] (node17) -- (node18);
  \draw [grayfactor] (node18) -- (node19.center) -- (node20.center) -- (node20'); 
  \draw [grayfactor] (node20') -- (node21); 
  \draw [grayfactor] (node21) -- (node22);
  \draw [grayfactor] (node22) -- (node23);
  
  \draw (node0') node [below = 1.2mm] {$\ell_1~~$};
  \draw (node20') node [above right = 1mm] {$\ell_2$};
\end{tikzpicture}
\caption{A subrun $\rho|Z$.} 
\label{fig:ramsey}

\end{wrapfigure} 
The above theorem can be used to show that if $\rho$
produces an output longer than $\bound$, then it contains an idempotent 
loop and a trace with non-empty output.  Below, we present a result 
in the same spirit, but refined in a way that it can be used to 
find anchor points inside specific intervals.
To formally state the result, we consider subsequences 
of $\r$ induced by sets of locations that are not necessarily
contiguous.
Recall the notation $\rho|Z$ introduced on page~\pageref{rhoZ}: 
$\rho|Z$ is the subsequence of $\rho$ induced by the location set $Z$. 
For example, Figure \ref{fig:ramsey} depicts a set $Z=[\ell_1,\ell_2]\cap (I\times\bbN)$
by a hatched area, together with the induced subrun $\rho|Z$, represented by
thick arrows.

\smallskip

\begin{thm}\label{thm:simon2}
Let $I=[x_1,x_2]$ be an interval of positions, $K=[\ell_1,\ell_2]$ 
an interval of locations, and $Z = K \:\cap\: (I\times\bbN)$.
If $|\out{\rho|Z}| > \bound$, 
then there exists an idempotent loop $L$ and an anchor point $\ell$ of $L$ such that
\begin{enumerate}
  \item $x_1 < \min(L) < \max(L) < x_2$ (in particular, $L\subsetneq I$),
  \item $\ell_1 \lesstime \ell \lesstime \ell_2$ (in particular, $\ell \in K$),
  \item $\out{\tr{\ell}} \neq \emptystr$.
\end{enumerate}
\end{thm}

\begin{proof}
Let $I$, $K$, $Z$ be as in the statement, 
and suppose that $\big|\out{\rho| Z}\big| > \bound$.
We define 
$Z' = Z ~\setminus~ (\{\ell_1,\ell_2\} \cup \{x_1,x_2\}\times\bbN)$
and we observe that there are at most $2\hmax +2$ locations
that are missing from $Z'$. This means that $\rho| Z'$ contains 
all but $4\hmax+4$ 
transitions of $\rho|Z$, 
and because each transition outputs at most $\cmax$ letters, we have
$|\out{\rho| Z'}| > \bound - 4\cmax\cdot\hmax -4\cmax = \cmax\cdot\hmax\cdot 2^{3\emax}$.

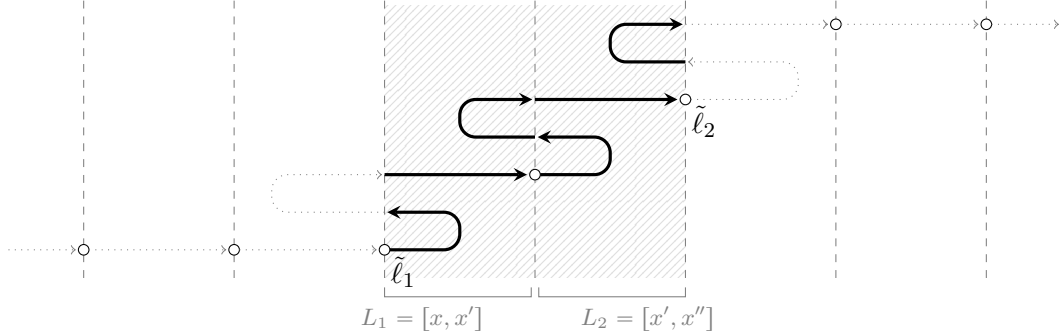
\begin{figure}[!t]
\centering
\begin{tikzpicture}[baseline=0, inner sep=0, outer sep=0, minimum size=0pt, scale=0.5]
  \tikzstyle{dot} = [draw, circle, fill=white, minimum size=4pt]
  \tikzstyle{fulldot} = [draw, circle, fill=black, minimum size=4pt]
  \tikzstyle{grayfactor} = [->, shorten >=1pt, rounded corners=6, gray, thin, dotted]
  \tikzstyle{factor} = [->, shorten >=1pt, rounded corners=6]
  \tikzstyle{dotfactor} = [->, shorten >=1pt, dotted, rounded corners=6]
  \tikzstyle{fullfactor} = [->, >=stealth, shorten >=1pt, very thick, rounded corners=6]
  \tikzstyle{dotfullfactor} = [->, >=stealth, shorten >=1pt, dotted, very thick, rounded corners=6]

  \fill [pattern=north east lines, pattern color=gray!25]
        (10,-0.75) rectangle (18,6.5);
  \draw [dashed, thin, gray] (2,-0.75) -- (2,6.75);
  \draw [dashed, thin, gray] (6,-0.75) -- (6,6.75);
  \draw [dashed, thin, gray] (10,-0.75) -- (10,6.75);
  \draw [dashed, thin, gray] (14,-0.75) -- (14,6.75);
  \draw [dashed, thin, gray] (18,-0.75) -- (18,6.75);
  \draw [dashed, thin, gray] (22,-0.75) -- (22,6.75);
  \draw [dashed, thin, gray] (26,-0.75) -- (26,6.75);
  \draw [gray] (10,-1) -- (10,-1.25) -- (13.9,-1.25) -- (13.9,-1);
  \draw [gray] (14.1,-1) -- (14.1,-1.25) -- (18,-1.25) -- (18,-1);
  \draw [gray] (11,-1.5) node [below] {\footnotesize $L_1=[x,x']$};
  \draw [gray] (17,-1.5) node [below] {\footnotesize $L_2=[x',x'']$};

  \draw (0,0) node (node0) {};
  \draw (2,0) node [dot] (node1) {};
  \draw (6,0) node [dot] (node2) {};
  \draw (10,0) node  [dot] (node3) {};
  \draw (12,0) node (node4) {};
  \draw (12,1) node (node5) {};
  \draw (10,1) node (node6) {};
  \draw (7,1) node (node7) {};
  \draw (7,2) node (node8) {};
  \draw (10,2) node (node9) {};
  \draw (14,2) node [dot] (node10) {};
  \draw (16,2) node (node11) {};
  \draw (16,3) node (node12) {};
  \draw (14,3) node (node13) {};
  \draw (12,3) node (node14) {};
  \draw (12,4) node (node15) {};
  \draw (14,4) node (node16) {};
  \draw (18,4) node  [dot] (node17) {};
  \draw (21,4) node (node18) {};
  \draw (21,5) node (node19) {};
  \draw (18,5) node (node20) {};
  \draw (16,5) node (node21) {};
  \draw (16,6) node (node22) {};
  \draw (18,6) node (node23) {};
  \draw (22,6) node [dot] (node24) {};
  \draw (26,6) node [dot] (node25) {};
  \draw (28,6) node (node26) {};

  \draw [grayfactor] (node0) -- (node1);
  \draw [grayfactor] (node1) -- (node2);
  \draw [grayfactor] (node2) -- (node3);
  \draw [fullfactor] (node3) -- (node4.center) -- (node5.center) -- (node6);
  \draw [grayfactor] (node6) -- (node7.center) -- (node8.center) -- (node9);
  \draw [fullfactor] (node9) -- (node10);
  \draw [fullfactor] (node10) -- (node11.center) -- (node12.center) -- (node13);
  \draw [fullfactor] (node13) -- (node14.center) -- (node15.center) -- (node16);
  \draw [fullfactor] (node16) -- (node17);
  \draw [grayfactor] (node17) -- (node18.center) -- (node19.center) -- (node20);
  \draw [fullfactor] (node20) -- (node21.center) -- (node22.center) -- (node23);
  \draw [grayfactor] (node23) -- (node24);
  \draw [grayfactor] (node24) -- (node25);
  \draw [grayfactor] (node25) -- (node26);
  
  \draw (node3) node [below right = 1.2mm] {$\tilde\ell_1$};
  \draw (node17) node [below right = 1.2mm] {$\tilde\ell_2$};
\end{tikzpicture}
\caption{Two consecutive idempotent loops with the same effect.}%
\label{fig:ramsey-proof}
\end{figure}


For every level $y$, let $X_y$ be the set of positions $x$ such that
$(x,y)$ is the source location of some transition of $\rho|Z'$ 
that produces non-empty output.
For example, if we refer to Figure~\ref{fig:ramsey-proof},
the vertical dashed lines represent the positions 
of $X_y$ for a particular level $y$; accordingly, the circles 
in the figure represent the locations of the form $(x,y)$, for 
all $x\in X_y$.
Since each transition outputs at most $\cmax$ letters, 
we have $\sum_y |X_y| > \hmax\cdot 2^{3\emax}$.
Moreover, since there are at most $\hmax$ levels, 
there is a level $y$ (which we fix hereafter) such that $|X_y| > 2^{3\emax}$.

We now prove the following:

\begin{clm}
There are two consecutive loops $L_1=[x,x']$ and $L_2=[x',x'']$ 
with endpoints $x,x',x''\in X_y$ and such that $E_{L_1}=E_{L_2}=E_{L_1\cup L_2}$.
\end{clm}

\begin{proof}
By Theorem \ref{th:simon}, 
there is a factorization forest for $X_y$ of height at most $3\emax$.
Since $\rho$ is a valid run, the dummy element $\bot$ of the effect 
semigroup does not appear in this factorization forest. 
Moreover, since $|X_y| > 2^{3\emax}$, we know that the factorization 
forest contains an internal node $L'=[x'_1,x'_{k+1}]$ with $k > 2$ children, 
say $L_1=[x'_1,x'_2], \dots, L_k=[x'_k,x'_{k+1}]$. 
By definition of factorization forest, the effects 
$E_{L'}$, $E_{L_1}$, \dots, $E_{L_k}$ are all equal and idempotent.
In particular, the effect $E_{L'}=E_{L_1}=\dots=E_{L_k}$ 
is a triple of the form $(F_{L'},c_1,c_2)$,
where $c_i=\rho|x_i$ is the crossing sequence at $x'_i$. 
Finally, since $E_{L'}$ is idempotent, we have that $c_1 = c_2$ and 
this is equal to the crossing sequences of $\rho$ at the positions 
$x'_1,\dots,x'_{k+1}$. This shows that $L_1,L_2$ are idempotent loops.
\end{proof}

Turning back to the proof of the theorem, we know from the above claim that 
there are two consecutive idempotent loops $L_1=[x,x']$ and $L_2=[x',x'']$ with the 
same effect and with endpoints $x,x',x''\in X_y \subseteq I \:\setminus\: \{x_1,x_2\}$
(see again Figure~\ref{fig:ramsey-proof}).

Let $\tilde\ell_1=(x,y)$ and $\tilde\ell_2=(x'',y)$, and observe that 
$\tilde\ell_1,\tilde\ell_2\in Z'$. In particular, $\tilde\ell_1$ and $\tilde\ell_2$ 
are strictly between $\ell_1$ and $\ell_2$. 
Suppose by symmetry that $\tilde\ell_1 \leqtime \tilde\ell_2$.
Further let $C$ be the component of $L_1\cup L_2$ (or, equally, of $L_1$ or $L_2$) 
that contains the node $y$.
Below, we focus on the factors of $\rho[\tilde\ell_1,\tilde\ell_2]$ 
that are intercepted by $L_1\cup L_2$: these are represented in 
Figure~\ref{fig:ramsey-proof} by the thick arrows. 
By Lemma~\ref{lem:component2} all these factors correspond to edges
of the same component $C$, namely, they are $(L_1 \cup L_2,C)$-factors.

Let us fix an arbitrary factor $\alpha$ of $\rho[\tilde\ell_1,\tilde\ell_2]$ 
that is intercepted by $L_1\cup L_2$, and assume that $\alpha=\beta_1 \cdots \beta_k$, 
where $\beta_1,\dots,\beta_k$ are the factors intercepted by either $L_1$ or $L_2$. 

\begin{clm}
If $\beta,\beta'$ are two factors intercepted by $L_1=[x,x']$ and $L_2=[x',x'']$,
with $E_{L_1}=E_{L_2}=E_{L_1\cup L_2}$, and $\beta,\beta'$ are adjacent in the run
$\rho$ (namely, they share an endpoint at position $x'$), then $\beta,\beta'$
correspond to edges in the same component of $L_1$ (or, equally, $L_2$).
\end{clm}

\begin{proof}
Let $C$ be the component of $L_1$ and $y_1 \mapsto y_2$ the edge of $C$ 
that corresponds to the factor $\beta$ intercepted by $L_1$. 
Similarly, let $C'$ be the component of $L_2$ and $y_3 \mapsto y_4$ the edge of $C'$
that corresponds to the factor $\beta'$ intercepted by $L_2$. 
Since $\beta$ and $\beta'$ share an endpoint at position $x'$, 
we know that $y_2=y_3$. This shows that $C \cap C' \neq \emptyset$, 
and hence $C=C'$.
\end{proof}

The above claim shows that any two adjacent factors $\beta_i,\beta_{i+1}$ 
correspond to edges in the same component of $L_1$ and $L_2$, respectively.
Thus, by transitivity, all factors $\beta_1,\dots,\beta_k$ correspond to 
edges in the same component, say $C'$. 
We claim that $C'=C$. Indeed, if $\beta_1$ is intercepted by $L_1$, 
then $C'=C$ because $\alpha$ and $\beta_1$ start from the same location
and hence they correspond to edges of the flow that depart from the 
same node. The other case is where $\beta_1$ is intercepted by $L_2$,
for which a symmetric argument can be applied.

So far we have shown that every factor of $\rho[\tilde\ell_1,\tilde\ell_2]$ intercepted 
by $L_1\cup L_2$ can be factorized into some $(L_1,C)$-factors and some $(L_2,C)$-factors.
We conclude the proof with the following observations:
\begin{itemize}
  \item By construction, both loops $L_1,L_2$ are contained in the interval of positions $I=[x_1,x_2]$,
        and have endpoints different from $x_1,x_2$.
  \item Both anchor points of $C$ inside $L_1,L_2$ belong to the interval of locations 
        $K\:\setminus\:\{\ell_1,\ell_2\}$. 
        This holds because
        $\rho[\tilde\ell_1,\tilde\ell_2]$ contains a factor $\alpha$ that is intercepted 
        by $L_1\cup L_2$ and spans across all the positions from $x$ to $x''$, namely,
        an $\LR$-factor. 
        This factor starts at the anchor point of $C$ inside $L_1$ 
        and visits the anchor point of $C$ inside $L_2$. 
        Moreover, by construction, $\alpha$ is also a factor of the subsequence $\rho|Z'$.
        This shows that the anchor points of $C$ inside $L_1$ and $L_2$ belong to $Z'$, 
        and in particular to $K\:\setminus\:\{\ell_1,\ell_2\}$.
  \item The first factor of $\rho[\tilde\ell_1,\tilde\ell_2]$ that is intercepted by 
        $L_1\cup L_2$ starts at $\tilde\ell_1=(x,y)$, which by construction is the 
        source location of some transition producing non-empty output.
        By the previous arguments, this factor is a concatenation of $(L_1,C)$-factors 
        and $(L_2,C)$-factors. This implies that the trace of the anchor point 
        of $C$ inside $L_1$, or the trace of $C$ inside $L_2$ produces non-empty output.
\qedhere
\end{itemize}        
\end{proof}

\medskip
\subsection*{Inversions and periodicity.}
The first important notion that is used to characterize one-way definability 
is that of inversion. It turns out that 
the
definition of inversion in the sweeping case (see
page~\pageref{page-def-inversion})  can be reused almost verbatim 
in the two-way setting. The only difference is that here we require the loops
to be idempotent and we do not enforce output-minimality (we will discuss this
latter choice further below, with a formal definition of output-minimality at hand).

\begin{defi}\label{def:inversion-twoway}
An \emph{inversion} of the run $\rho$ is a tuple $(L_1,\ell_1,L_2,\ell_2)$ such that
\begin{enumerate}
  \item $L_1,L_2$ are idempotent loops,
  \item $\ell_1=(x_1,y_1)$ and $\ell_2=(x_2,y_2)$ 
        are anchor points inside $L_1$ and $L_2$, respectively,
  \item $\ell_1 \lesstime \ell_2$ and $x_1 > x_2$, 
  \item for both $i=1$ and $i=2$, $\out{\tr{\ell_i}}\neq\emptystr$. 
\end{enumerate}
\end{defi}

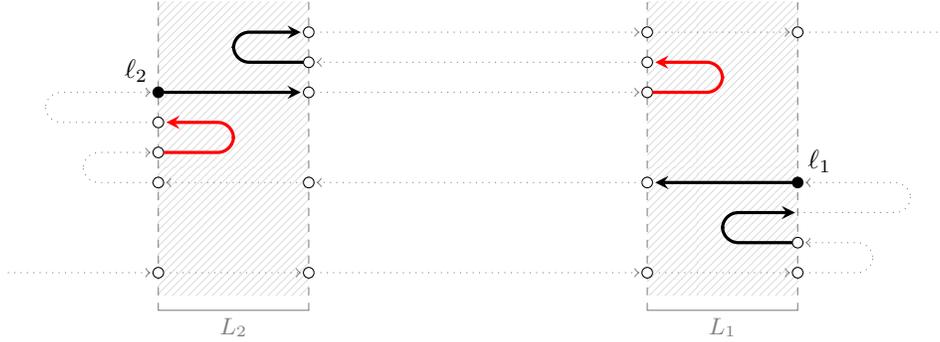
\begin{figure}[!t]
\centering
\begin{tikzpicture}[baseline=0, inner sep=0, outer sep=0, minimum size=0pt, scale=0.5, yscale=0.8]
  \tikzstyle{dot} = [draw, circle, fill=white, minimum size=4pt]
  \tikzstyle{fulldot} = [draw, circle, fill=black, minimum size=4pt]
  \tikzstyle{grayfactor} = [->, shorten >=1pt, rounded corners=6, gray, thin, dotted]
  \tikzstyle{factor} = [->, shorten >=1pt, rounded corners=6]
  \tikzstyle{dotfactor} = [->, shorten >=1pt, dotted, rounded corners=6]
  \tikzstyle{fullfactor} = [->, >=stealth, shorten >=1pt, very thick, rounded corners=6]
  \tikzstyle{dotfullfactor} = [->, >=stealth, shorten >=1pt, dotted, very thick, rounded corners=6]

  \fill [pattern=north east lines, pattern color=gray!25]
        (4,-0.75) rectangle (8,9);
  \draw [dashed, thin, gray] (4,-0.75) -- (4,9);
  \draw [dashed, thin, gray] (8,-0.75) -- (8,9);
  \draw [gray] (4,-1) -- (4,-1.25) -- (8,-1.25) -- (8,-1);
  \draw [gray] (6,-1.5) node [below] {\footnotesize $L_2$};

  \fill [pattern=north east lines, pattern color=gray!25]
        (17,-0.75) rectangle (21,9);
  \draw [dashed, thin, gray] (17,-0.75) -- (17,9);
  \draw [dashed, thin, gray] (21,-0.75) -- (21,9);
  \draw [gray] (17,-1) -- (17,-1.25) -- (21,-1.25) -- (21,-1);
  \draw [gray] (19,-1.5) node [below] {\footnotesize $L_1$};

  \draw (0,0) node (node0) {};
  \draw (4,0) node [dot] (node1) {};
  \draw (8,0) node [dot] (node2) {};
  \draw (17,0) node [dot] (node3) {};
  \draw (21,0) node [dot] (node4) {};
  \draw (23,0) node (node5) {};
  \draw (23,1) node (node6) {};

  \draw (21,1) node [dot] (node7) {};
  \draw (19,1) node (node8) {};
  \draw (19,2) node (node9) {};
  \draw (21,2) node (node10) {};
  \draw (24,2) node (node11) {};
  \draw (24,3) node (node12) {};
  \draw (21,3) node [fulldot] (node13) {};
  \draw (17,3) node [dot] (node14) {};

  \draw (8,3) node [dot] (node15) {};
  \draw (4,3) node [dot] (node16) {};
  \draw (2,3) node (node17) {};
  \draw (2,4) node (node18) {};

  \draw (4,4) node [dot] (node19) {};
  \draw (6,4) node (node20) {};
  \draw (6,5) node (node21) {};
  \draw (4,5) node [dot] (node22) {};
  \draw (1,5) node (node23) {};
  \draw (1,6) node (node24) {};
  \draw (4,6) node [fulldot] (node25) {};
  \draw (8,6) node [dot] (node26) {};

  \draw (17,6) node [dot] (node27) {};
  \draw (19,6) node (node28) {};
  \draw (19,7) node (node29) {};
  \draw (17,7) node [dot] (node30) {};

  \draw (8,7) node [dot] (node31) {};
  \draw (6,7) node (node32) {};
  \draw (6,8) node (node33) {};
  \draw (8,8) node [dot] (node34) {};

  \draw (17,8) node [dot] (node35) {};
  \draw (21,8) node [dot] (node36) {};
  \draw (25,8) node (node37) {};

  \draw [grayfactor] (node0) -- (node1);
  \draw [grayfactor] (node1) -- (node2);
  \draw [grayfactor] (node2) -- (node3);
  \draw [grayfactor] (node3) -- (node4);
  \draw [grayfactor] (node4) -- (node5.center) -- (node6.center) -- (node7);
  
  \draw [fullfactor] (node7) -- (node8.center) -- (node9.center) -- (node10);
  \draw [grayfactor] (node10) -- (node11.center) -- (node12.center) -- (node13);
  \draw [fullfactor] (node13) -- (node14);

  \draw [grayfactor] (node14) -- (node15);
  \draw [grayfactor] (node15) -- (node16);
  \draw [grayfactor] (node16) -- (node17.center) -- (node18.center) -- (node19);

  \draw [fullfactor, red] (node19) -- (node20.center) -- (node21.center) -- (node22);
  \draw [grayfactor] (node22) -- (node23.center) -- (node24.center) -- (node25);
  \draw [fullfactor] (node25) -- (node26);
  
  \draw [grayfactor] (node26) -- (node27);
  \draw [fullfactor, red] (node27) -- (node28.center) -- (node29.center) -- (node30);
  \draw [grayfactor] (node30) -- (node31);
  \draw [fullfactor] (node31) -- (node32.center) -- (node33.center) -- (node34);

  \draw [grayfactor] (node34) -- (node35);
  \draw [grayfactor] (node35) -- (node36);
  \draw [grayfactor] (node36) -- (node37);
  
  \draw (node13) node [above right=2mm] {\small $\ell_1$};
  \draw (node25) node [above left=2mm] {\small $\ell_2$};
\end{tikzpicture}
\caption{An example of an inversion $(L_1,\ell_1,L_2,\ell_2)$ of a two-way run.}\label{fig:inversion-twoway}
\end{figure}

\noindent
Figure \ref{fig:inversion-twoway} gives an example of an inversion involving
the idempotent loop $L_1$ with anchor point $\ell_1$, and the idempotent 
loop $L_2$ with anchor point $\ell_2$. The intercepted factors that form 
the corresponding traces are represented by thick arrows; the one highlighted in red 
are those that produce non-empty output.

The implication \PR1 $\Rightarrow$ \PR2 of Theorem \ref{thm:main2} in the two-way
case is formalized below exactly as in Proposition~\ref{prop:periodicity-sweeping},
and the proof is very similar to the sweeping case. 
More precisely, it can be checked that the proof of the first claim in
Proposition~\ref{prop:periodicity-sweeping} was shown independently of the sweeping assumption 
--- one just needs to replace the use of Equation \ref{eq:pumped-run}
with Proposition \ref{prop:pumping-twoway}.
The sweeping assumption was used only for deriving the notion of \emph{output-minimal} 
factor, which was crucial to conclude that the period $p$ is bounded by the specific 
constant $\bound$.
In this respect, the proof of Proposition~\ref{prop:periodicity-twoway} 
requires a different argument for showing that $p \le\bound$: 

\begin{prop}\label{prop:periodicity-twoway}
If $\cT$ is one-way definable, then the following property \PR2 holds:
\begin{quote}
  For all inversions $(L_1,\ell_1,L_2,\ell_2)$ of $\rho$, 
  the period $p$ of the word
\[
  \out{\tr{\ell_1}} ~ \out{\rho[\ell_1,\ell_2]} ~ \out{\tr{\ell_2}}
\]
divides both $|\out{\tr{\ell_1}}|$ and $|\out{\tr{\ell_2}}|$.
Moreover, $p \le \bound$.
\end{quote}

\end{prop}


We only need to show here that $p \le \bound$.  Recall that in the
sweeping case we relied on the assumption that the factors
$\tr{\ell_1}$ and $\tr{\ell_2}$ of an inversion are output-minimal,
and on Lemma~\ref{lem:output-minimal-sweeping}.  In the general case
we need to replace output-minimality by the following notion:

\begin{defi}\label{def:output-minimal-twoway}
Consider pairs $(L,C)$ consisting of an idempotent loop $L$ 
and a component $C$ of $L$.
\begin{enumerate}
\item On such pairs, define the relation 
      $\lesspair$ by $(L',C') \lesspair (L,C)$ if
      $L'\subsetneq L$ and at least one $(L',C')$-factor 
      is contained in some $(L,C)$-factor. 
\item A pair $(L,C)$ is \emph{output-minimal} if 
      $(L',C') \lesspair (L,C)$ implies $\out{\tr{\an{C'}}}=\emptystr$.
\end{enumerate}
\end{defi}

\noindent
Note that the relation $\lesspair$ is not a partial order in general
(it is however antisymmetric).
Moreover, it is easy to see that the notion of output-minimal
pair $(L,C)$ generalizes that of output-minimal factor introduced
in the sweeping case: indeed, if $\ell$ is the anchor point of a 
loop $L$ of a sweeping transducer and $\tr{\ell}$ satisfies 
Definition \ref{def:output-minimal-sweeping}, then the pair
$(L,C)$ is output-minimal, where $C$ is the unique component
whose edge corresponds to $\tr{\ell}$.

The following lemma bounds the length of the output trace 
$\out{\tr{\an{C}}}$ for an output-minimal pair $(L,C)$:

\begin{lem}\label{lem:output-minimal-twoway}
For every output-minimal pair $(L,C)$, $|\out{\tr{\an{C}}}| \le \bound$.
\end{lem}

\begin{proof}
Consider a pair $(L,C)$ consisting of an idempotent loop $L=[x_1,x_2]$ and a component $C$ of $L$. 
Suppose by contradiction that $|\out{\tr{\an{C}}}|>\bound$. 
We will show that $(L,C)$ is not output-minimal.

Recall that $\tr{\an{C}}$ is a concatenation of $(L,C)$-factors, say,
$\tr{\an{C}}=\beta_1\cdots\beta_k$. Let $\ell_1$ (resp.~$\ell_2$) be the 
first (resp.~last) location that is visited by these factors. Further let
$K = [\ell_1,\ell_2]$ and $Z = K \:\cap\: (L\times\bbN)$. 
By construction, the subsequence $\rho|Z$ can be seen as a concatenation
of the factors $\beta_1,\dots,\beta_k$, possibly in a different order than
that of $\tr{\an{C}}$. This implies that $|\out{\rho|Z}| > \bound$.

By Theorem \ref{thm:simon2}, we know that there exist an idempotent 
loop $L'\subsetneq L$ and a component $C'$ of $L'$ such that 
$\an{C'} \in K$ and $\out{\tr{\an{C'}}}\neq\emptystr$.
Note that the $(L',C')$-factor that starts at the anchor point
$\an{C'}$ (an $\LR$- or $\RL$-factor) is entirely contained 
in some $(L,C)$-factor. 
This implies that $(L',C') \lesspair (L,C)$, and thus $(L,C)$ is not output-minimal.
\end{proof}

We remark that the above lemma cannot be used directly to bound the period 
of the output produced amid an inversion. The reason is that we cannot 
restrict ourselves to inversions $(L_1,\ell_1,L_2,\ell_2)$ that induce 
output-minimal pairs $(L_i,C_i)$ for $i=1,2$, where $C_i$ is the unique 
component of the anchor point $\ell_i$.
An example is given in Figure~\ref{fig:inversion-twoway},
assuming that the factors depicted in red are the only ones that 
produce non-empty output, and the lengths of these outputs exceed $\bound$. 
On the one hand $(L_1,\ell_1,L_2,\ell_2)$ is an inversion, but 
$(L_1,C_1)$ is not output-minimal. 
On the other hand, it is possible that $\rho$ contains no other 
inversion than $(L_1,\ell_1,L_2,\ell_2)$: any loop strictly contained
in the red factor in $L_1$ will have the anchor point \emph{after}
$\ell_2$. 

We are now ready to show the second claim of
Proposition~\ref{prop:periodicity-twoway}. 

\begin{figure}[!t]
\centering
\begin{tikzpicture}[baseline=0, inner sep=0, outer sep=0, minimum size=0pt, scale=0.39, xscale=0.6]
  \tikzstyle{dot} = [draw, circle, fill=white, minimum size=4pt]
  \tikzstyle{fulldot} = [draw, circle, fill=black, minimum size=4pt]
  \tikzstyle{grayfactor} = [->, shorten >=1pt, rounded corners=6, gray, thin, dotted]
  \tikzstyle{factor} = [->, shorten >=1pt, rounded corners=6]
  \tikzstyle{dotfactor} = [->, shorten >=1pt, dotted, rounded corners=6]
  \tikzstyle{fullfactor} = [->, >=stealth, shorten >=1pt, very thick, rounded corners=6]
  \tikzstyle{dotfullfactor} = [->, >=stealth, shorten >=1pt, dotted, very thick, rounded corners=6]

\begin{scope}
  \fill [pattern=north east lines, pattern color=gray!25]
        (4,-0.75) rectangle (10,9);
  \draw [dashed, thin, gray] (4,-0.75) -- (4,9);
  \draw [dashed, thin, gray] (10,-0.75) -- (10,9);
  \draw [gray] (4,-1) -- (4,-1.25) -- (10,-1.25) -- (10,-1);
  \draw [gray] (7,-1.5) node [below] {\footnotesize $L_2$};

  \fill [pattern=north east lines, pattern color=gray!25]
        (16,-0.75) rectangle (22,9);
  \draw [dashed, thin, gray] (16,-0.75) -- (16,9);
  \draw [dashed, thin, gray] (22,-0.75) -- (22,9);
  \draw [gray] (16,-1) -- (16,-1.25) -- (22,-1.25) -- (22,-1);
  \draw [gray] (19,-1.5) node [below] {\footnotesize $L_1$};

  \draw (0,0) node (node0) {};
  \draw (4,0) node [dot] (node1) {};
  \draw (10,0) node [dot] (node2) {};
  \draw (16,0) node [dot] (node3) {};
  \draw (22,0) node [dot] (node4) {};
  \draw (24,0) node (node5) {};
  \draw (24,1) node (node6) {};

  \draw (22,1) node [dot] (node7) {};
  \draw (19,1) node (node8) {};
  \draw (19,2) node (node9) {};
  \draw (22,2) node (node10) {};
  \draw (25,2) node (node11) {};
  \draw (25,3) node (node12) {};
  \draw (22,3) node [fulldot] (node13) {};
  \draw (16,3) node [dot] (node14) {};

  \draw (10,3) node [dot] (node15) {};
  \draw (4,3) node [dot] (node16) {};
  \draw (2,3) node (node17) {};
  \draw (2,4) node (node18) {};

  \draw (4,4) node [dot] (node19) {};
  \draw (7,4) node (node20) {};
  \draw (7,5) node (node21) {};
  \draw (4,5) node [dot] (node22) {};
  \draw (1,5) node (node23) {};
  \draw (1,6) node (node24) {};
  \draw (4,6) node [fulldot] (node25) {};
  \draw (10,6) node [dot] (node26) {};

  \draw (16,6) node [dot] (node27) {};
  \draw (20,6) node (node28) {};
  \draw (20,7) node (node29) {};
  \draw (16,7) node [dot] (node30) {};

  \draw (10,7) node [dot] (node31) {};
  \draw (6,7) node (node32) {};
  \draw (6,8) node (node33) {};
  \draw (10,8) node [dot] (node34) {};

  \draw (16,8) node [dot] (node35) {};
  \draw (22,8) node [dot] (node36) {};
  \draw (26,8) node (node37) {};

  \draw [grayfactor] (node0) -- (node1);
  \draw [grayfactor] (node1) -- (node2);
  \draw [grayfactor] (node2) -- (node3);
  \draw [grayfactor] (node3) -- (node4);
  \draw [grayfactor] (node4) -- (node5.center) -- (node6.center) -- (node7);
  
  \draw [fullfactor] (node7) -- (node8.center) -- (node9.center) -- (node10);
  \draw [grayfactor] (node10) -- (node11.center) -- (node12.center) -- (node13);
  \draw [fullfactor] (node13) -- (node14);

  \draw [grayfactor] (node14) -- (node15);
  \draw [grayfactor] (node15) -- (node16);
  \draw [grayfactor] (node16) -- (node17.center) -- (node18.center) -- (node19);

  \draw [fullfactor, red] (node19) -- (node20.center) -- (node21.center) -- (node22);
  \draw [grayfactor] (node22) -- (node23.center) -- (node24.center) -- (node25);
  \draw [fullfactor] (node25) -- (node26);
  
  \draw [grayfactor] (node26) -- (node27);
  \draw [fullfactor, red] (node27) -- (node28.center) -- (node29.center) -- (node30);
  \draw [grayfactor] (node30) -- (node31);
  \draw [fullfactor] (node31) -- (node32.center) -- (node33.center) -- (node34);

  \draw [grayfactor] (node34) -- (node35);
  \draw [grayfactor] (node35) -- (node36);
  \draw [grayfactor] (node36) -- (node37);

  \draw (node13) node [above right=1.5mm] {\footnotesize $\ell_1$};
  \draw (node25) node [above left=1.5mm] {\footnotesize $\ell_2$};
\end{scope}

\begin{scope}[xshift=32cm]
  \fill [pattern=north east lines, pattern color=gray!25]
        (4,-0.75) rectangle (10,11);
  \draw [dashed, thin, gray] (4,-0.75) -- (4,11);
  \draw [dashed, thin, gray] (10,-0.75) -- (10,11);
  \draw [gray] (4,-1) -- (4,-1.25) -- (10,-1.25) -- (10,-1);
  \draw [gray] (7,-1.5) node [below] {\footnotesize $L_2$};

  \fill [pattern=north east lines, pattern color=gray!25]
        (16,-0.75) rectangle (22,11);
  \fill [pattern=north east lines, pattern color=gray!25]
        (22,-0.75) rectangle (28,11);
  \draw [dashed, thin, gray] (16,-0.75) -- (16,11);
  \draw [dashed, thin, gray] (22,-0.75) -- (22,11);
  \draw [dashed, thin, gray] (28,-0.75) -- (28,11);
  \draw [gray] (16,-1) -- (16,-1.25) -- (22,-1.25) -- (22,-1);
  \draw [gray] (22,-1) -- (22,-1.25) -- (28,-1.25) -- (28,-1);
  \draw [gray] (22,-1.5) node [below] {\footnotesize two copies of $L_1$};

  \fill [pattern=north east lines, pattern color=red!25]
        (17.5,-0.75) rectangle (19.5,11);
  \draw [dashed, thin, red!25] (17.5,-0.75) -- (17.5,11);
  \draw [dashed, thin, red!25] (19.5,-0.75) -- (19.5,11);
  \draw [red!25] (17.5,11.25) -- (17.5,11.5) -- (19.5,11.5) -- (19.5,11.25);
  \draw [red!50] (18.5,11.75) node [above] {\footnotesize $L_0$};
 

  \draw (0,0) node (node0) {};
  \draw (4,0) node [dot] (node1) {};
  \draw (10,0) node [dot] (node2) {};
  \draw (16,0) node [dot] (node3) {};
  \draw (22,0) node [dot] (node4) {};
  \draw (28,0) node [dot] (node5) {};
  \draw (30,0) node (node6) {};
  \draw (30,1) node (node7) {};

  \draw (28,1) node [dot] (node8) {};
  \draw (25,1) node (node9) {};
  \draw (25,2) node (node10) {};
  \draw (28,2) node [dot] (node11) {};

  \draw (31,2) node (node12) {};
  \draw (31,3) node (node13) {};
  \draw (28,3) node [fulldot] (node14) {};
  \draw (22,3) node [dot] (node15) {};
  \draw (19,3) node (node16) {};
  \draw (19,4) node (node17) {};
  \draw (22,4) node [dot] (node18) {};
  \draw (26,4) node (node19) {};
  \draw (26,5) node (node20) {};
  \draw (22,5) node [fulldot] (node21) {};
  \draw (16,5) node [dot] (node22) {};

  \draw (10,5) node [dot] (node23) {};
  \draw (4,5) node [dot] (node24) {};
  \draw (2,5) node (node25) {};
  \draw (2,6) node (node26) {};
  \draw (4,6) node [dot] (node27) {};

  \draw (7,6) node (node28) {};
  \draw (7,7) node (node29) {};
  \draw (4,7) node [dot] (node30) {};
  \draw (1,7) node (node31) {};
  \draw (1,8) node (node32) {};
  \draw (4,8) node [fulldot] (node33) {};
  \draw (10,8) node [dot] (node34) {};

  \draw (16,8) node [dot] (node35) {};
  \draw (20,8) node (node36) {};
  \draw (20,9) node (node37) {};
  \draw (16,9) node [dot] (node38) {};

  \draw (10,9) node [dot] (node39) {};
  \draw (6,9) node (node40) {};
  \draw (6,10) node (node41) {};
  \draw (10,10) node [dot] (node42) {};

  \draw (16,10) node [dot] (node43) {};
  \draw (22,10) node [dot] (node44) {};
  \draw (28,10) node [dot] (node45) {};
  \draw (32,10) node (node46) {};

  \draw [grayfactor] (node0) -- (node1);
  \draw [grayfactor] (node1) -- (node2);
  \draw [grayfactor] (node2) -- (node3);
  \draw [grayfactor] (node3) -- (node4);
  \draw [grayfactor] (node4) -- (node5);
  \draw [grayfactor] (node5) -- (node6.center) -- (node7.center) -- (node8);
  
  \draw [fullfactor] (node8) -- (node9.center) -- (node10.center) -- (node11);
  \draw [grayfactor] (node11) -- (node12.center) -- (node13.center) -- (node14);
  \draw [fullfactor] (node14) -- (node15);
  \draw [fullfactor] (node15) -- (node16.center) -- (node17.center) -- (node18);
  \draw [fullfactor, red] (node18) -- (node19.center) -- (node20.center) -- (node21);
  \draw [fullfactor] (node21) -- (node22);

  \draw [grayfactor] (node22) -- (node23);
  \draw [grayfactor] (node23) -- (node24);
  \draw [grayfactor] (node24) -- (node25.center) -- (node26.center) -- (node27);

  \draw [fullfactor, red] (node27) -- (node28.center) -- (node29.center) -- (node30);
  \draw [grayfactor] (node30) -- (node31.center) -- (node32.center) -- (node33);
  \draw [fullfactor] (node33) -- (node34);

  \draw [grayfactor] (node34) -- (node35);
  \draw [fullfactor, red] (node35) -- (node36.center) -- (node37.center) -- (node38);

  \draw [grayfactor] (node38) -- (node39);
  \draw [fullfactor] (node39) -- (node40.center) -- (node41.center) -- (node42);

  \draw [grayfactor] (node42) -- (node43);
  \draw [grayfactor] (node43) -- (node44);
  \draw [grayfactor] (node44) -- (node45);
  \draw [grayfactor] (node45) -- (node46);
\end{scope}
\end{tikzpicture}
\caption{An inversion $(L_1,\ell_1,L_2,\ell_2)$ with non output-minimal pairs $(L_i,C_i)$.}
\label{fig:non-output-minimal-inversion-full}
\end{figure}
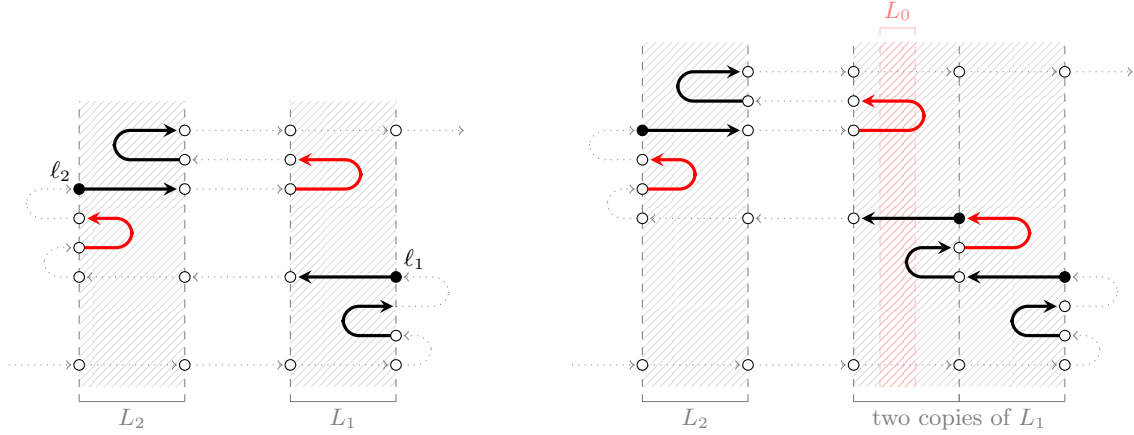


\begin{proof}[Proof of Proposition~\ref{prop:periodicity-twoway}]
The proof of the second claim requires a refinement of the arguments that involve 
pumping the run $\rho$ simultaneously on three different loops. As usual, we 
assume that the loops $L_1,L_2$ of the inversion are disjoint (otherwise,
we preliminarily pump one of the two loops a few times).

Recall that the word 
\[
  \out{\tr{\ell_1}} ~ \out{\rho[\ell_1,\ell_2]} ~ \out{\tr{\ell_2}}
\] 
has period $p = \gcd\big( |\out{\tr{\ell_1}}|, |\out{\tr{\ell_2}}|
\big)$, but that we cannot bound $p$ by assuming that $(L_1,\ell_1,L_2,\ell_2)$ 
is output-minimal.
However, in the pumped run $\rho^{(2,1)}$ we do find inversions with
output-minimal pairs.
For example, as depicted in the right part of Figure~\ref{fig:non-output-minimal-inversion-full},
we can consider the left and right copy of $L_1$ in $\rho^{(2,1)}$, 
denoted by $\olft L_1$ and $\ort L_1$, respectively.
Accordingly, we denote by $\olft\ell_1$ and $\ort\ell_1$ the left and
right copy of $\ell_1$ in $\rho^{(2,1)}$. 
Now, let $(L_0,C_0)$ be any {\sl output-minimal} pair such that $L_0$ 
is an idempotent loop, $\out{\tr{\an{C_0}}}\neq\emptystr$, and either 
$(L_0,C_0)=(\olft L_1,C_1)$ or $(L_0,C_0) \lesspair (\olft L_1,C_1)$.
Such a loop $L_0$ is  represented in  
Figure~\ref{fig:non-output-minimal-inversion-full} by 
the red vertical stripe. Further let $\ell_0=\an{C_0}$.

We claim that either $(L_0,\ell_0,L_2,\ell_2)$ or $(\ort L_1,\ort\ell_1,L_0,\ell_0)$ 
is an inversion of the run $\rho^{(2,1)}$, depending on whether $\ell_0$ occurs 
before or after $\ell_2$.
First, note that all the loops $L_0$, $L_2$, $\ort L_1$ are idempotent
and non-overlapping; 
more precisely, we have 
$\max(L_2) \le \min(L_0)$ and $\max(L_0) \le \min(\ort L_1)$. 
Moreover, the outputs of the traces $\tr{\ell_0}$, $\tr{\ort\ell_1}$, 
and $\tr{\ell_2}$ are all non-empty.
So it remains to distinguish two cases based on the ordering of the anchor
points $\ell_0$, $\ort\ell_1$, $\ell_2$.
If $\ell_0 \lesstime \ell_2$, then $(L_0,\ell_0,L_2,\ell_2)$ is an inversion.
Otherwise, because $(\ort L_1,\ort\ell_1,L_2,\ell_2)$ is an inversion, we know that 
$\ort\ell_1 \lesstime \ell_2 \leqtime \ell_0$, and hence $(\ort L_1,\ort\ell_1,L_0,C_0)$
is an inversion.

Now, we know that
$\rho^{(2,1)}$ contains the inversion $(\ort L_1,\ort \ell_1,L_2,\ell_2)$, but also 
an inversion involving the output-minimal pair $(L_0,C_0)$, with $L_0$ strictly 
between $\ort L_1$ and $L_2$.
For all $m_0,m_1,m_2$, we define $\rho^{(m_0,m_1,m_2)}$ as the run obtained from 
$\rho^{(2,1)}$ by pumping $m_0,m_1,m_2$ times the loops $L_0,\ort L_1,L_2$, respectively.
By reasoning as we did in the proof of Proposition~\ref{prop:periodicity-sweeping}
(cf.~{\em Periodicity of outputs of pumped runs}), one can show that there are 
arbitrarily large output factors of $\rho^{(m_0,m_1,m_2)}$ that are produced
within the inversion on $\ell_0$ (i.e.~either $(L_0,\ell_0,L_2,\ell_2)$ or $(\ort L_1,\ort\ell_1,L_0,\ell_0)$)
and that are periodic with period $p'$ that divides $|\out{\tr{\ell_0}}|$.
In particular, by Lemma~\ref{lem:output-minimal-twoway}, we know that 
$p' \le \bound$.
Moreover, large portions of these factors are also produced within the 
inversion $(\ort L_1,\ort\ell_1,L_2,\ell_2)$, and hence 
by Theorem~\ref{thm:fine-wilf} they have period $\gcd(p,p')$.

To conclude the proof we need to transfer the periodicity property
from the pumped runs $\rho^{(m_0,m_1,m_2)}$ to the original run $\rho$.
This is done exactly like in Proposition~\ref{prop:periodicity-sweeping}
by relying on 
\gabriele{Usual replacement of Saarela...} 
Lemma \ref{lem:periods}:
we observe that the periodicity
property holds for large enough parameters $m_0,m_1,m_2$, hence for
all values of the parameters, and in particular 
for $m_0 = m_1 = m_2 = 1$. This shows that the word 
\[ 
  \out{\tr{\ell_1}} ~ \out{\rho[\ell_1,\ell_2]} ~ \out{\tr{\ell_2}}
\]
has period $\gcd(p,p') \le \bound$.
\end{proof}

So far we have shown that the output produced amid every inversion of a
run of a one-way definable two-way transducer is periodic, with period
bounded by $\bound$ and dividing the lengths of the trace outputs of 
the inversion. This basically proves the implication \PR1 $\Rightarrow$ \PR2
of Theorem \ref{thm:main2}.
In the next section we will follow a line of arguments similar to that of
Section \ref{sec:characterization-sweeping} to prove the remaining
implications \PR2 $\Rightarrow$ \PR3 $\Rightarrow$ \PR1.


\section{The characterization in the two-way case}\label{sec:characterization-twoway}

In this section we generalize the characterization of one-way
definability of sweeping transducers to the general two-way case. As
usual, we fix through the rest of the section a successful run $\rho$
of $\cT$ on some input word $u$.

\medskip
\subsection*{From periodicity of inversions to existence of decompositions.}
We continue by proving the second implication \PR2 $\Rightarrow$ \PR3 of 
Theorem \ref{thm:main2} in the two-way case. This requires showing the
existence of a suitable decomposition of a run $\rho$ that 
\emph{satisfies} property \PR2. Recall that \PR2 says that for every inversion 
$(L_1,\ell_1,L_2,\ell_2)$, the period of the word 
$\out{\tr{\ell_1}} ~ \out{\rho[\ell_1,\ell_2]} ~ \out{\tr{\ell_2}}$
divides $\gcd(|\out{\tr{\ell_1}}|, |\out{\tr{\ell_2}}|) \le \bound$.
The definitions underlying the decomposition of $\rho$ are similar
to those given in the sweeping case:

\begin{defi}\label{def:factors-twoway} 
Let $\rho[\ell,\ell']$ be a factor of a run $\rho$ of $\cT$, 
where $\ell=(x,y)$, $\ell'=(x',y')$, and $x\le x'$.
We call $\rho[\ell,\ell']$
\begin{itemize}
  \medskip
  \item \parbox[t]{\dimexpr\textwidth-\leftmargin}{
        \vspace{-2.75mm}
        \begin{wrapfigure}{r}{8cm}
        \vspace{-6mm}
\centering
\begin{tikzpicture}[baseline=0, inner sep=0, outer sep=0, minimum size=0pt, xscale=0.5, yscale=0.36]
\begin{scope}
  \tikzstyle{dot} = [draw, circle, fill=white, minimum size=4pt]
  \tikzstyle{fulldot} = [draw, circle, fill=black, minimum size=4pt]
  \tikzstyle{grayfactor} = [->, shorten >=1pt, rounded corners=4, gray, thin, dotted]
  \tikzstyle{factor} = [->, shorten >=1pt, rounded corners=4]
  \tikzstyle{dotfactor} = [->, shorten >=1pt, dotted, rounded corners=4]
  \tikzstyle{fullfactor} = [->, >=stealth, shorten >=1pt, very thick, rounded corners=4]
  \tikzstyle{dotfullfactor} = [->, >=stealth, shorten >=1pt, dotted, very thick, rounded corners=4]

  \fill [pattern=north east lines, pattern color=gray!25]
        (1,-0.75) rectangle (8,2.5);
  \fill [pattern=north east lines, pattern color=gray!25]
        (8,1.5) rectangle (17,6.75);
  \draw [dashed, thin, gray] (8,-0.75) -- (8,6.75);
 
  \draw (2,0) node (node-1) {} ;
  \draw (6,0) node [fulldot] (node0) {};
  \draw (8,0) node (node1) {};
  \draw (14,0) node (node2) {};
  \draw (14,1) node (node3) {};
  \draw (8,1) node (node4) {};
  \draw (2,1) node (node5) {};
  \draw (2,2) node (node6) {};
  \draw (8,2) node [dot] (node7) {};
  \draw (12,2) node (node8) {};
  \draw (12,3) node (node9) {};
  \draw (8,3) node (node10) {};
  \draw (6,3) node (node11) {};
  \draw (6,4) node (node12) {};
  \draw (8,4) node (node13) {};
  \draw (14,4) node (node14) {};
  \draw (14,5) node (node15) {};
  \draw (8,5) node (node16) {};
  \draw (4,5) node (node17) {};
  \draw (4,6) node (node18) {};
  \draw (8,6) node (node19) {};
  \draw (12,6) node [fulldot] (node20) {};
  \draw (16,6) node (node21) {} ;

  \draw [grayfactor] (node-1) -- (node0);
  \draw [fullfactor] (node0) -- (node1); 
  \draw [dotfullfactor] (node1) -- (node2.center) -- (node3.center) -- (node4); 
  \draw [fullfactor] (node4) -- (node5.center) -- (node6.center) -- (node7); 
  \draw [fullfactor] (node7) -- (node8.center) -- (node9.center) -- (node10);
  \draw [dotfullfactor] (node10) -- (node11.center) -- (node12.center) -- (node13); 
  \draw [fullfactor] (node13) -- (node14.center) -- (node15.center) -- (node16); 
  \draw [dotfullfactor] (node16) -- (node17.center) -- (node18.center) -- (node19); 
  \draw [fullfactor] (node19) -- (node20);
  \draw [grayfactor] (node20) -- (node21);
  
  \draw (node0) node [below left = 1.2mm] {$\ell$};
  \draw (node7) node [rectangle, fill=white, above left = 1.8mm] 
        {$\phantom{i}\ell_z$};
  \draw (node20) node [above right = 1mm] {$\ell'$};

  \draw (15.75,0.4) node {$Z_{\ell_z}^\lowerright$};
  \draw (2.25,5) node {$Z_{\ell_z}^\upperleft$};
\end{scope}
\end{tikzpicture}
\vspace{-2mm}
\caption{A diagonal.}\label{fig:diagonal-twoway}

        \vspace{-5mm}
        \end{wrapfigure} 
        a \emph{$\bound$-diagonal} 
        if for all $z\in[x,x']$, there is a location $\ell_z$ at position $z$
        such that $\ell \leqtime \ell_z \leqtime \ell'$ and the words
        $\out{\rho|Z_{\ell_z}^\upperleft}$ and 
        $\out{\rho|Z_{\ell_z}^\lowerright}$ 
        have length at most $\bound$,
        where $Z_{\ell_z}^\upperleft = [\ell_z,\ell'] \:\cap\: \big([0,z]\times\bbN\big)$
        and $Z_{\ell_z}^\lowerright = [\ell,\ell_z] \:\cap\: \big([z,\omega]\times\bbN\big)$;
        }
  \bigskip
  \medskip
  \item \parbox[t]{\dimexpr\textwidth-\leftmargin}{%
        \vspace{-2.75mm}
        \begin{wrapfigure}{r}{8cm}
        \vspace{-6mm}
\centering
\begin{tikzpicture}[baseline=0, inner sep=0, outer sep=0, minimum size=0pt, xscale=0.5, yscale=0.36]
\begin{scope} 
  \tikzstyle{dot} = [draw, circle, fill=white, minimum size=4pt]
  \tikzstyle{fulldot} = [draw, circle, fill=black, minimum size=4pt]
  \tikzstyle{grayfactor} = [->, shorten >=1pt, rounded corners=4, gray, thin, dotted]
  \tikzstyle{factor} = [->, shorten >=1pt, rounded corners=4]
  \tikzstyle{dotfactor} = [->, shorten >=1pt, dotted, rounded corners=4]
  \tikzstyle{fullfactor} = [->, >=stealth, shorten >=1pt, very thick, rounded corners=4]
  \tikzstyle{dotfullfactor} = [->, >=stealth, shorten >=1pt, dotted, very thick, rounded corners=4]

  \fill [pattern=north east lines, pattern color=gray!25]
        (6,6.75) rectangle (12,-0.75);
  \draw [dashed, thin, gray] (6,-0.75) -- (6,6.75);
  \draw [dashed, thin, gray] (12,-0.75) -- (12,6.75);
 
  \draw (2,0) node (node0) {} ;
  \draw (6,0) node [fulldot] (node1) {};
  \draw (8,0) node (node2) {};
  \draw (8,1) node (node3) {};
  \draw (6,1) node (node4) {};
  \draw (2,1) node (node5) {};
  \draw (2,2) node (node6) {};
  \draw (6,2) node (node7) {};
  \draw (12,2) node (node8) {};
  \draw (16,2) node (node9) {};
  \draw (16,3) node (node10) {};
  \draw (12,3) node (node11) {};
  \draw (10,3) node (node12) {};
  \draw (10,4) node (node13) {};
  \draw (12,4) node (node14) {};
  \draw (14,4) node (node15) {};
  \draw (14,5) node (node16) {};
  \draw (12,5) node (node17) {};
  \draw (6,5) node (node18) {};
  \draw (4,5) node (node19) {};
  \draw (4,6) node (node20) {};
  \draw (6,6) node (node21) {};
  \draw (12,6) node [fulldot] (node22) {};
  \draw (16,6) node (node23) {} ;

  \draw [grayfactor] (node0) -- (node1);
  \draw [fullfactor] (node1) -- (node2.center) -- (node3.center) -- (node4); 
  \draw [dotfullfactor] (node4) -- (node5.center) -- (node6.center) -- (node7); 
  \draw [fullfactor] (node7) -- (node8);
  \draw [dotfullfactor] (node8) -- (node9.center) -- (node10.center) -- (node11); 
  \draw [fullfactor] (node11) -- (node12.center) -- (node13.center) -- (node14); 
  \draw [dotfullfactor] (node14) -- (node15.center) -- (node16.center) -- (node17); 
  \draw [fullfactor] (node17) -- (node18);
  \draw [dotfullfactor] (node18) -- (node19.center) -- (node20.center) -- (node21); 
  \draw [fullfactor] (node21) -- (node22);
  
  \draw (node1) node [below left = 1.2mm] {$\ell$};
  \draw (node22) node [above right = 1mm] {$\ell'$};

  \draw (2.5,5) node {$Z^\leftshort$};
  \draw (15.5,0.4) node {$Z^\rightshort$};
  \draw [grayfactor] (node22) -- (node23) ;
  
\end{scope}
\end{tikzpicture}
\vspace{-2mm}
\caption{A \rightward{block.}\phantom{diagonal.}}\label{fig:block-twoway}

        \vspace{-6mm}
        \end{wrapfigure} 
        a \emph{$\bound$-block} if the word
        $\out{\rho[\ell,\ell']}$ is almost periodic with bound $\bound$, 
        and 
        $\out{\rho|Z^\leftshort}$ and 
        $\out{\rho|Z^\rightshort}$ have length at most $\bound$,
        where $Z^\leftshort = [\ell,\ell'] \:\cap\: \big([0,x]\times \bbN\big)$
        and $Z^\rightshort = [\ell,\ell'] \:\cap\: \big([x',\omega]\times \bbN\big)$.
        }
        \vspace{8mm}
\end{itemize}
\end{defi}


The definition of $\bound$-decomposition is copied verbatim from the sweeping case,
but uses the new notions of $\bound$-diagonal and $\bound$-block:

\begin{defi}\label{def:decomposition-twoway}
A \emph{$\bound$-decomposition} of a run $\rho$ of $\cT$ is a factorization 
$\prod_i\,\rho[\ell_i,\ell_{i+1}]$ of $\rho$ into $\bound$-diagonals and $\bound$-blocks.
\end{defi}

\noindent
To provide further intuition,
we consider the transduction of Example~\ref{ex:running} 
and the two-way transducer $\cT$ that implements it in the most natural way.
Figure~\ref{fig:decomposition-twoway} shows an example of a run of $\cT$ on
an input of the form $u_1 \:\#\: u_2 \:\#\: u_3 \:\#\: u_4$, where
$u_2,\, u_4 \in (abc)^*$, $u_1,\,u_3\nin (abc)^*$, and $u_3$ has even length. 
The factors of the run that produce long outputs are highlighted 
by the bold arrows. The first and third factors of the decomposition, 
i.e.~$\rho[\ell_1,\ell_2]$ and $\rho[\ell_3,\ell_4]$, are diagonals
(represented by the blue hatched areas); the second and fourth factors 
$\rho[\ell_2,\ell_3]$ and $\rho[\ell_4,\ell_5]$ are blocks
(represented by the red hatched areas). 

To identify the blocks of a possible decomposition of $\rho$, 
we reuse the equivalence relation $\simeq$ introduced in Definition \ref{def:crossrel}.
Recall that this is the reflexive and transitive closure of the relation $\crossrel$ 
that groups any two locations $\ell,\ell'$ that occur between $\ell_1,\ell_2$, for some
inversion $(L_1,\ell_1,L_2,\ell_2)$.

The proof that the output produced inside each $\simeq$-equivalence class is periodic,
with period at most $\bound$ (Lemma \ref{lem:overlapping}) carries over in the 
two-way case without modifications.
Similarly, every $\simeq$-equivalence class can be extended to the left and to the
right by using Definition \ref{def:bounding-box-sweeping}, which 
we report here verbatim for the sake of readability, together with an exemplifying figure.

\begin{figure}[!t]
\centering
\begin{tikzpicture}[baseline=0, inner sep=0, outer sep=0, minimum size=0pt, xscale=0.5, yscale=0.4]
  \tikzstyle{dot} = [draw, circle, fill=white, minimum size=4pt]
  \tikzstyle{fulldot} = [draw, circle, fill=black, minimum size=4pt]
  \tikzstyle{grayfactor} = [->, shorten >=1pt, rounded corners=6, gray, thin, dotted]
  \tikzstyle{factor} = [->, shorten >=1pt, rounded corners=6]
  \tikzstyle{dotfactor} = [->, shorten >=1pt, dotted, rounded corners=6]
  \tikzstyle{fullfactor} = [->, >=stealth, shorten >=1pt, very thick, rounded corners=6]
  \tikzstyle{dotfullfactor} = [->, >=stealth, shorten >=1pt, dotted, very thick, rounded corners=6]

  \fill [pattern=north east lines, pattern color=blue!18]
        (0.5,-.5) rectangle (6.5,2.5);
  \fill [pattern=north east lines, pattern color=red!20]
        (6.5,1.5) rectangle (12.5,4.5);
  \fill [pattern=north east lines, pattern color=blue!18]
        (12.5,3.5) rectangle (18.5,6.5);
  \fill [pattern=north east lines, pattern color=red!20]
        (18.5,5.5) rectangle (24.5,8.5);
  \draw [dashed, thin, gray] (0.5,-1) -- (0.5,9);
  \draw [dashed, thin, gray] (6.5,-1) -- (6.5,9);
  \draw [dashed, thin, gray] (12.5,-1) -- (12.5,9);
  \draw [dashed, thin, gray] (18.5,-1) -- (18.5,9);
  \draw [dashed, thin, gray] (24.5,-1) -- (24.5,9);
  \draw [gray] (1,-1.25) -- (1,-1.5) -- (6,-1.5) -- (6,-1.25);
  \draw [gray] (3.5,-2) node [below] {\small $u_1$};
  \draw [gray] (6.5,-1.75) node [below] {\footnotesize $\#$};
  \draw [gray] (7,-1.25) -- (7,-1.5) -- (12,-1.5) -- (12,-1.25);
  \draw [gray] (9.5,-2) node [below] {\small $u_2$};
  \draw [gray] (12.5,-1.75) node [below] {\footnotesize $\#$};
  \draw [gray] (13,-1.25) -- (13,-1.5) -- (18,-1.5) -- (18,-1.25);
  \draw [gray] (15.5,-2) node [below] {\small $u_3$};
  \draw [gray] (18.5,-1.75) node [below] {\footnotesize $\#$};
  \draw [gray] (19,-1.25) -- (19,-1.5) -- (24,-1.5) -- (24,-1.25);
  \draw [gray] (21.5,-2) node [below] {\small $u_4$};

  \draw (0.5,0) node [dot] (node0) {};  

  \draw (2,0) node (node1) {};
  \draw (5.5,0) node (node2) {};
  \draw (12,0) node (node3) {};
  \draw (12,1) node (node4) {};
  \draw (5,1) node (node5) {};
  \draw (5,2) node (node6) {};

  \draw (6.5,2) node [dot] (node10') {};

  \draw (8,2) node (node11) {};
  \draw (11.5,2) node (node12) {};
  \draw (18,2) node (node13) {};
  \draw (18,3) node (node14) {};
  \draw (17,3) node (node15) {};
  \draw (8,3) node (node16) {};
  \draw (7,3) node (node17) {};
  \draw (7,4) node (node18) {};
  \draw (8,4) node (node19) {};
  \draw (11.5,4) node (node20) {};

  \draw (12.5,4) node [dot] (node20') {};

  \draw (14,4) node (node21) {};
  \draw (17.5,4) node (node22) {};
  \draw (24,4) node (node23) {};
  \draw (24,5) node (node24) {};
  \draw (17,5) node (node25) {};
  \draw (17,6) node (node26) {};

  \draw (18.5,6) node [dot] (node30') {};

  \draw (20,6) node (node31) {};
  \draw (23.5,6) node (node32) {};
  \draw (24,6) node (node33) {};
  \draw (24,7) node (node34) {};
  \draw (23,7) node (node35) {};
  \draw (20,7) node (node36) {};
  \draw (19,7) node (node37) {};
  \draw (19,8) node (node38) {};
  \draw (20,8) node (node39) {};
  \draw (23.5,8) node (node40) {};

  \draw (24.5,8) node [dot] (node40') {};

  \draw [grayfactor] (node0) -- (node1);
  \draw [fullfactor] (node1) -- (node2);
  \draw [grayfactor] (node2) -- (node3.center) -- (node4.center) 
                         -- (node5.center) -- (node6.center) -- (node10');

  \draw [grayfactor] (node10') -- (node11);
  \draw [fullfactor] (node11) -- (node12);
  \draw [grayfactor] (node12) -- (node13.center) -- (node14.center) -- (node15) -- 
                 (node15) -- (node16) --
                 (node16) -- (node17.center) -- (node18.center) -- (node19); 
  \draw [fullfactor] (node19) -- (node20); 
  \draw [grayfactor] (node20) -- (node20');

  \draw [grayfactor] (node20') -- (node21);
  \draw [fullfactor] (node21) -- (node22);
  \draw [grayfactor] (node22) -- (node23.center) -- (node24.center) 
                          -- (node25.center) -- (node26.center) -- (node30');

  \draw [grayfactor] (node30') -- (node31);
  \draw [fullfactor] (node31) -- (node32);
  \draw [grayfactor] (node32) -- (node33.center) -- (node34.center) -- (node35) --
                 (node35) -- (node36) -- 
                 (node36) -- (node37.center) -- (node38.center) -- (node39);
  \draw [fullfactor] (node39) -- (node40);
  \draw [grayfactor] (node40) -- (node40');

  \draw (node0) node [left = 1.5mm] {$\ell_1$};
  \draw (node10') node [above left = 2mm] {$\ell_2$};
  \draw (node20') node [above right = 2mm] {$\ell_3$};
  \draw (node30') node [above left = 2mm] {$\ell_4$};
  \draw (node40') node [right = 1.5mm] {$\ell_5$};
\end{tikzpicture}
\caption{A decomposition of a run of a two-way transducer.}\label{fig:decomposition-twoway}
\end{figure}
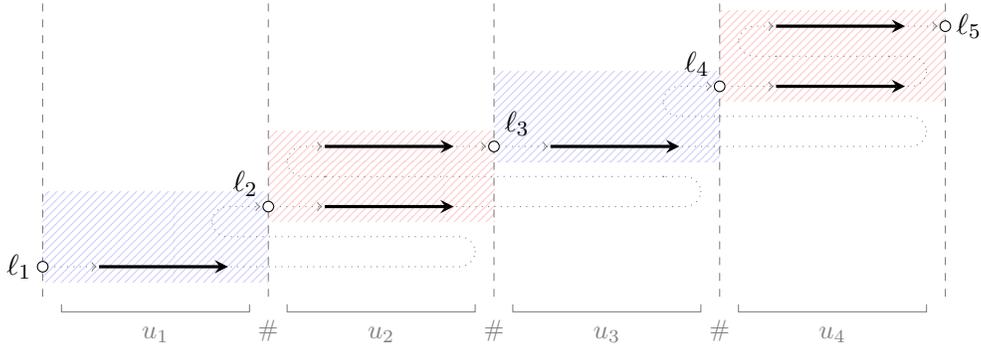


\bigskip\noindent
\begin{minipage}[l]{\textwidth-5.8cm}
\begin{defi}\label{def:bounding-box-twoway}
Consider a non-singleton $\simeq$-equivalence class $K=[\ell,\ell']$. 
Let $\an{K}$
be the restriction of $K$ to the anchor points occurring in some inversion, 
and $X_{\an{K}} = \{x \::\: \exists y\: (x,y)\in \an{K}\}$ 
be the projection of $\an{K}$ on positions.
We define $\block{K}=[\tilde\ell,\tilde\ell']$, where 
\begin{itemize}
  \item $\tilde\ell$ is the latest location $(\tilde x,\tilde y) \leqtime \ell$ 
        such that $\tilde x = \min(X_{\an{K}})$, 
  \item $\tilde\ell'$ is the earliest location $(\tilde x,\tilde y) \geqtime \ell'$ 
        such that $\tilde x = \max(X_{\an{K}})$
\end{itemize}  
(note that the locations $\tilde\ell,\tilde\ell'$ exist since $\ell,\ell'$ 
are anchor points in some inversion).
\end{defi}
\end{minipage}
\begin{minipage}[r]{5.7cm}
\vspace{-2mm}
\centering
\scalebox{0.9}{
\begin{tikzpicture}[baseline=0, inner sep=0, outer sep=0, minimum size=0pt, scale=0.52, yscale=1.2]
  \tikzstyle{dot} = [draw, circle, fill=white, minimum size=4pt]
  \tikzstyle{fulldot} = [draw, circle, fill=black, minimum size=4pt]
  \tikzstyle{grayfactor} = [->, shorten >=1pt, rounded corners=8, gray, thin, dotted]
  \tikzstyle{factor} = [->, shorten >=1pt, rounded corners=8]
  \tikzstyle{dotfactor} = [->, shorten >=1pt, dotted, rounded corners=8]
  \tikzstyle{fullfactor} = [->, >=stealth, shorten >=1pt, very thick, rounded corners=8]
  \tikzstyle{dotfullfactor} = [->, >=stealth, shorten >=1pt, dotted, very thick, rounded corners=8]

  \fill [pattern=north east lines, pattern color=gray!25]
        (2,-0.325) rectangle (8,4.325);
  \draw [dashed, thin, gray] (2,-1.25) -- (2,5);
  \draw [dashed, thin, gray] (8,-1.25) -- (8,5);
  \draw (2,-1.5) node [below] {\footnotesize $\min(X_{\an{K}})$};
  \draw (8,-1.5) node [below] {\footnotesize $\max(X_{\an{K}})$};

  \draw (0,0) node (node0) {};
  \draw (2,0) node [dot] (node1) {};
  \draw (6,0) node (node3) {};
  \draw (6,1) node (node4) {};
  \draw (4,1) node [fulldot] (node6) {};
  \draw (2,1) node (node7) {};
  \draw (0,1) node (node8) {};
  \draw (0,2) node (node9) {};
  \draw (2,2) node (node10) {};
  \draw (8,2) node (node11) {};
  \draw (10,2) node (node12) {};
  \draw (10,3) node (node13) {};
  \draw (8,3) node (node14) {};
  \draw (6,3) node [fulldot] (node15) {};
  \draw (4,3) node (node17) {};
  \draw (4,4) node (node18) {};
  \draw (8,4) node [dot] (node20) {};
  \draw (10,4) node (node21) {};

  \draw [grayfactor] (node0) -- (node1); 
  \draw [dotfullfactor] (node1) -- (node3.center) -- (node4.center) -- (node6); 
  \draw [fullfactor] (node6) -- (node7.center) -- (node8.center) -- (node9.center) -- (node10.center)
                  -- (node11.center) -- (node12.center) -- (node13.center) -- (node14.center) -- (node15); 
  \draw [dotfullfactor] (node15) -- (node17.center) -- (node18.center) -- (node20); 
  \draw [grayfactor] (node20) -- (node21); 
  
  \draw (node1) node [below left = 2mm] {$\tilde\ell\ $};
  \draw (node6) node [below = 2mm] {$\ell$};
  \draw (node15) node [above = 2mm] {$\ell'$};
  \draw (node20) node [above right = 2mm] {$\ \tilde\ell'$};
\end{tikzpicture}
}
\captionof{figure}{Block construction.\label{fig:block-construction-twoway}}
\end{minipage}


\medskip
As usual, we call \emph{$\simeq$-block} any factor of $\rho$ of the form $\rho|\block{K}$ 
that is obtained by applying the above definition to a non-singleton $\simeq$-class $K$.
Lemma \ref{lem:bounding-box-sweeping}, which shows that $\simeq$-blocks can indeed
be used as $\bound$-blocks in a decomposition of $\rho$, generalizes easily to the
two-way case:

\reviewOne[inline]{Most proofs are quite notation-heavy. A small figure to represent the positions and factors involved is a good thing to have. For example, Lemma 5.8 has Fig.5, but the corresponding (and slightly more involved) Lemme 8.4 hasn't. This seems like a worthwile investment for any long proof, or any time a huge number of notations are involved. 
}%

\gabriele[inline]{I think the reviewer is referring to the lack of a figure for the previous definition, 
exactly like for Definition 5.7. I put it.}%

\reviewTwo[inline]{
-p.37 l.21-22 the $\ell$ should be replaced with $\tilde{\ell}$
\olivier[inline]{I didn't find where}%
\felix[inline]{second and third line of the list of properties we will prove. I agree.}%
}%

\begin{lem}\label{lem:bounding-box-twoway}
If $K$ is a non-singleton $\simeq$-equivalence class, 
then $\rho|\block{K}$ is a $\bound$-block.
\end{lem}

\begin{proof}
The proof is similar to that of Lemma \ref{lem:bounding-box-sweeping}. The main
difference is that here we will bound the lengths of some outputs using 
a Ramsey-type argument (Theorem \ref{thm:simon2}), instead of output-minimality
of factors (Lemma \ref{lem:output-minimal-sweeping}). To follow the various constructions
and arguments the reader can refer to Figure \ref{fig:block-construction-twoway}.

Let $K=[\ell,\ell']$, $\an{K}$, $X_{\an{K}}$, and $\block{K}=[\tilde\ell,\tilde\ell']$
be as in Definition \ref{def:bounding-box-twoway}, where
$\tilde\ell=(\tilde x,\tilde y)$, $\tilde\ell'=(\tilde x',\tilde y')$,
$\tilde x=\min(X_{\an{K}})$, and $\tilde x'=\max(X_{\an{K}})$.
We need to verify that $\rho|\block{K}$ is a $\bound$-block, namely, that:
\begin{itemize}
  \item $\tilde x \le \tilde x'$,
  \item $\out{\rho[\tilde \ell,\tilde \ell']}$ is almost periodic with bound $\bound$, 
  \item $\out{\rho|Z^\leftshort}$ and $\out{\rho|Z^\rightshort}$ have length at most $\bound$,
        where $Z^\leftshort = [\tilde \ell,\tilde \ell'] \:\cap\: \big([0,x]\times \bbN\big)$
        and $Z^\rightshort = [\tilde \ell,\tilde \ell'] \:\cap\: \big([x',\omega]\times \bbN\big)$.
\end{itemize}
The first condition $\tilde x \le \tilde x'$ follows immediately from 
$\tilde x=\min(X_{\an{K}})$ and $\tilde x'=\max(X_{\an{K}})$.

Next, we prove that the output produced by the factor 
$\rho[\tilde\ell,\tilde\ell']$ is almost periodic with bound $\bound$. 
By Definition \ref{def:bounding-box-twoway}, we have 
$\tilde\ell \leqtime \ell \lesstime \ell' \leqtime \tilde\ell'$,
and by Lemma \ref{lem:overlapping} 
we know that $\out{\rho[\ell,\ell']}$ is periodic with 
period at most $\bound$. So it suffices to bound the length 
of the words $\out{\rho[\tilde\ell,\ell]}$ and $\out{\rho[\ell',\tilde\ell']}$. 
We shall focus on the former word, as the arguments for the latter 
are similar. 

First, we show that the factor $\rho[\tilde\ell,\ell]$ 
lies entirely to the right of position $\tilde x$ 
(in particular, it starts at an even level $\tilde y$).
Indeed, if this were not the case, there would exist another location 
$\ell''=(\tilde x,\tilde y + 1)$, on the same position as $\tilde\ell$, 
but at a higher level, such that $\tilde\ell \lesstime \ell'' \leqtime \ell$. 
But this would contradict Definition \ref{def:bounding-box-twoway}
($\tilde\ell$ is the \emph{latest} location $(x,y) \leqtime \ell$ 
such that $x = \min(X_{\an{K}})$).

Suppose now that the length of $|\out{\rho[\tilde\ell,\ell]}| >
\bound$.  We head towards a contradiction by finding a location
$\ell'' \lesstime \ell$ that is $\simeq$-equivalent to the first
location $\ell$ of the $\simeq$-equivalence class $K$.  Since the
factor $\rho[\tilde\ell,\ell]$ lies entirely to the right of position
$\tilde x$, it is intercepted by the interval $I=[\tilde x,\omega]$.
So $|\out{\rho[\tilde\ell,\ell]}| > \bound$ is equivalent to saying
$|\out{\rho|Z}| > \bound$, where $Z = [\tilde\ell,\ell] \:\cap\:
\big([\tilde x,\omega]\times\bbN\big)$.  Then, Theorem
\ref{thm:simon2} implies the existence of an idempotent loop $L$ and
an anchor point $\ell''$ of $L$ such that
\begin{itemize}
  \item $\min(L) > \tilde x$,
  \item $\tilde\ell \lesstime \ell'' \lesstime \ell$,
  \item $\out{\tr{\ell''}}\neq\emptystr$.
\end{itemize}
Further recall that $\tilde x=\min(X_{\an{K}})$ is the leftmost position of 
locations in the class $K=[\ell,\ell']$ that are also anchor points of inversions. 
In particular, there is an inversion $(L_1,\ell''_1,L_2,\ell''_2)$, with 
$\ell''_2=(\tilde x,y''_2) \in K$. 
Since $\ell'' \lesstime \ell \leqtime \ell''_2$ 
and the position of $\ell''$ is to the right of the position of $\ell''_2$, 
we know that $(L,\ell'',L_2,\ell''_2)$ is also an inversion, 
and hence $\ell'' \simeq \ell''_2 \simeq \ell$.
But since $\ell'' \neq \ell$, we get a contradiction with the 
assumption that $\ell$ is the first location of the $\simeq$-class $K$. 
In this way we have shown that $|\out{\rho[\tilde\ell,\ell]}| \le \bound$.

It remains to bound the lengths of the outputs produced 
by the subsequences $\rho|Z^\leftshort$ and $\rho|Z^\rightshort$, 
where $Z^\leftshort=[\tilde\ell,\tilde\ell'] \:\cap\: \big([0,\tilde x]\times\bbN\big)$ 
and $Z^\rightshort=[\tilde\ell,\tilde\ell'] \:\cap\: \big([\tilde x',\omega]\times\bbN\big)$.
As usual, we consider only one of the two symmetric cases.
Suppose, by way of contradiction, that $|\out{\rho|Z^\leftshort}| > \bound$.
By Theorem \ref{thm:simon2}, there exist an idempotent loop $L$ 
and an anchor point $\ell''$ of $L$ such that
\begin{itemize}
  \item $\max(L) < \tilde x$,
  \item $\tilde\ell \lesstime \ell'' \lesstime \tilde\ell'$,
  \item $\out{\tr{\ell''}}\neq\emptystr$.
\end{itemize}
By following the same line of reasoning as before, we recall that
$\ell$ is the first location of the non-singleton class $K$.
From this we derive the existence an inversion $(L_1,\ell''_1,L_2,\ell''_2)$ 
where $\ell''_1 = \ell$.
We claim that $\ell \leqtime \ell''$.
Indeed, if this were not the case, then, because $\ell''$ is strictly to the 
left of $\tilde x$ and $\ell$ is to the right of $\tilde x$, there would exist 
a location $\ell'''$ between $\ell''$ and $\ell$ that lies at position $\tilde x$. 
But $\tilde\ell \lesstime \ell'' \leqtime \ell''' \leqtime \ell$ would 
contradict the fact that $\tilde\ell$ is the {\sl latest} location before 
$\ell$ that lies at the position $\tilde x$.
Now that we know that $\ell \leqtime \ell''$ and that $\ell''$ is to the left of $\tilde x$, 
we observe that $(L_1,\ell''_1,L,\ell'')$ is also an inversion, and hence $\ell''\in \an{K}$. 
Since $\ell''$ is strictly to the left of $\tilde x$,
we get a contradiction with the definition of $\tilde x$ as leftmost 
position of the locations in $\an{K}$.
So we conclude that $|\out{\rho|Z^\leftshort}| \le \bound$.
\end{proof}

The proof of Lemma \ref{lem:consecutive-blocks-sweeping}, which 
shows that $\simeq$-blocks do not overlap along the input axis, 
carries over in the two-way case, again without modifications.
Finally, we generalize Lemma \ref{lem:diagonal-sweeping} to the new
definition of diagonal, which completes the construction of a 
$\bound$-decomposition for the run $\rho$:

\begin{lem}\label{lem:diagonal-twoway}
Let $\rho[\ell,\ell']$ be a 
factor of $\rho$,
with $\ell=(x,y)$, $\ell'=(x',y')$, and $x\le x'$,
that does not overlap any $\simeq$-block.
Then $\rho[\ell,\ell']$ is a $\bound$-diagonal.
\end{lem}

\begin{proof}
Suppose by way of contradiction that there is some $z \in [x,x']$
such that, for all locations $\ell''$ at position $z$ and between $\ell$ and $\ell'$, 
one of the two conditions holds:
\begin{enumerate}
  \item $|\out{\rho|Z_{\ell''}^\upperleft}| > \bound$, 
        where $Z_{\ell''}^\upperleft = [\ell'',\ell'] \:\cap\: \big([0,z]\times\bbN\big)$,
  \item $|\out{\rho|Z_{\ell''}^\lowerright}| > \bound$, 
        where $Z_{\ell''}^\lowerright = [\ell,\ell''] \:\cap\: \big([z,\omega]\times\bbN\big)$.
\end{enumerate}
First, we claim that \emph{each} of the two conditions above are satisfied at
some locations $\ell''\in [\ell,\ell']$ at position $z$. 
Consider the highest even level $y''$ 
such that $\ell''=(z,y'') \in[\ell,\ell']$
(use Figure \ref{fig:diagonal-twoway} as a reference).
Since $z\le x'$, the outgoing transition at $\ell''$ is rightward oriented, 
and the set $Z_{\ell''}^\upperleft$ is empty. This means that 
condition (1) is trivially violated at $\ell''$, 
and hence condition (2) holds at $\ell''$ by the initial assumption. 
Symmetrically, condition (1) holds at the location $\ell''=(z,y'')$,
where $y''$ is the lowest even level with $\ell'' \in[\ell,\ell']$.

Let us now compare the levels where the above conditions hold.
Clearly, the lower the level of location $\ell''$, 
the easier it is to satisfy condition (1), and symmetrically for condition (2).
So, let $\ell^+=(z,y^+)$ (resp.~$\ell^-=(z,y^-)$) be the highest (resp.~lowest) 
location in $[\ell,\ell']$ at position $z$ that satisfies
condition (1) (resp.~condition (2)). 

We claim that $y^+ \ge y^-$.
For this, we first observe that $y^+ \ge y^- - 1$, since otherwise there 
would exist a location $\ell''=(z,y'')$, with $y^+ < y'' < y^-$, that 
violates both conditions (1) and (2).
Moreover, $y^+$ must be odd, otherwise the transition departing from 
$\ell^+ = (z,y^+)$ would be rightward oriented and the location $\ell'' = (z,y^+ + 1)$ 
would still satisfy condition (1), contradicting the definition of highest location $\ell^+$. 
For similar reasons, $y^-$ must also be odd, otherwise there would be a location 
$\ell'' = (z,y^- - 1)$ below $\ell^-$ that satisfies condition (2).
But since $y^+ \ge y^- - 1$ and both $y^+$ and $y^-$ are odd, 
we need to have $y^+ \ge y^-$.

In fact, from the previous arguments we know that the location $\ell''=(z,y^+)$ 
(or equally the location $(x,y^-)$) 
satisfies {\sl both} conditions (1) and (2). We can thus apply Theorem \ref{thm:simon2} to the 
sets $Z_{\ell''}^\lowerright$ and $Z_{\ell''}^\upperleft$, deriving the existence of
two idempotent loops $L_1,L_2$ and two anchor points $\ell_1,\ell_2$ of $L_1,L_2$, 
respectively, such that
\begin{itemize}
  \item $\max(L_2) < z < \min(L_1)$,
  \item $\ell \lesstime \ell_1 \lesstime \ell'' \lesstime \ell_2 \lesstime \ell'$, 
  \item $\out{\tr{\ell_1}},\out{\tr{\ell_2}}\neq\emptystr$.
\end{itemize}
In particular, since $\ell_1$ is to the right of $\ell_2$ w.r.t.~the order
of positions, we know that $(L_1,\ell_1,L_2,\ell_2)$ is an inversion, and 
hence $\ell_1 \simeq \ell_2$. But this contradicts the assumption that 
$\rho[\ell,\ell']$ does not overlap with any $\simeq$-block.
\end{proof}

\medskip
\reviewOne[inline]{Part 5, "From existence of decompositions to an equivalent one-way transducer": so far the paper did a great job hinting at the proof that remains. Here an intuition is given for the construction later detailed in 8.6, but no mention of an intuition for Proposition 9.2: if it is at all possible to do so, this might help the paper's overall flow.
  \olivier[inline]{answered in review1-answers.txt}
}%
\subsection*{From existence of decompositions to an equivalent one-way transducer.}
It remains to prove the last implication \PR3 $\Rightarrow$ \PR1 of Theorem~\ref{thm:main2},
which amounts to construct a one-way transducer $\cT'$ equivalent to $\cT$.

Hereafter, we denote by $D$ the language of words $u\in\dom(\cT)$ such that 
{\sl all} successful runs of $\cT$ on $u$ admit a $\bound$-decomposition.
So far, we know that if $\cT$ is one-way definable (\PR1), 
then $D=\dom(\cT)$ (\PR3). 
As a matter of fact, this reduces the one-way definability problem 
for $\cT$ to the containment problem $\dom(\cT) \subseteq D$.
\label{testing-containment}
We will see later (in Section~\ref{sec:complexity})
how the latter problem can be decided 
in double exponential space 
by further reducing it to checking the emptiness of the 
intersection of the languages $\dom(\cT)$ and $D^\complement$, 
where $D^\complement$ is the complement of $D$.

Below, we show how to construct a one-way transducer 
$\cT'$ of triple exponential size such that 
$\cT' \subseteq \cT$ and
$\dom(\cT')$ is the set of all input words that have 
{\sl some} successful run admitting a $\bound$-decomposition
(hence $\dom(\cT')\supseteq D$).
In particular, we will have that
\[
  \cT|_D \:\subseteq\: \cT' \:\subseteq\: \cT.
\]
Note that this will prove \PR3 to \PR1, as well as the second 
item of Theorem~\ref{thm:main}, since $D=\dom(\cT)$ 
if and only if $\cT$ is one-way definable.
A sketch of the proof of this construction when
$\cT$ is a sweeping transducer was given at the
end of Section \ref{sec:characterization-sweeping}.

\begin{prop}\label{prop:construction-twoway}
Given a functional two-way transducer $\cT$, 
a one-way transducer $\cT'$ 
satisfying 
\[\cT' \subseteq \cT \quad \text{ and } \quad \dom(\cT') \supseteq D\] can be constructed in $3\exptime$.
Moreover, if $\cT$ is sweeping, then $\cT'$
can be constructed in $2\exptime$.
\end{prop}

\begin{proof}
Given an input word $u$, the transducer $\cT'$ will guess (and check) 
a successful run $\rho$ of $\cT$ on $u$, together with a $\bound$-decomposition 
$\prod_i \rho[\ell_i,\ell_{i+1}]$.
The latter decomposition will be used by $\cT'$ to simulate the output of 
$\rho$ in left-to-right manner, thus proving that $\cT' \subseteq \cT$.
Moreover, $u\in D$ implies the existence of a successful run that can be 
decomposed, thus proving that $\dom(\cT') \supseteq D$.
We now provide the details of the construction of $\cT'$.

Guessing the run $\rho$ is standard (see, for instance, \cite{she59,HU79}): 
it amounts to guess the crossing sequences $\rho|x$ for 
each position $x$ of the input. Recall that this is a bounded
amount of information for each position $x$, since the run is normalized.
As concerns the decomposition of $\rho$, it can be encoded
by the endpoints $\ell_i$ of its factors, that is, by annotating 
the position of each $\ell_i$ as the level of $\ell_i$.
In a similar way $\cT'$ guesses the information of whether
each factor $\rho[\ell_i,\ell_{i+1}]$ is a $\bound$-diagonal or a $\bound$-block.

Thanks to the definition of decomposition 
(see Definition~\ref{def:decomposition-twoway} and Figure \ref{fig:decomposition-twoway}), 
every two distinct factors span across non-overlapping intervals of positions. 
This means that each position $x$ is covered by exactly one factor of 
the decomposition. We call this factor the \emph{active factor at position $x$}.
The mode of computation of the transducer will depend on 
the type of active factor: if the active factor is a diagonal
(resp.~a block), then we say that $\cT'$ is in \emph{diagonal mode} 
(resp.~\emph{block mode}). 
Below we describe the behaviour for these two modes of computation.

\smallskip 
\par\noindent\emph{Diagonal mode.}~
We recall the key condition satisfied by the diagonal 
$\rho[\ell,\ell']$ that is active at position $x$ 
(cf.~Definition~\ref{def:factors-twoway} and Figure~\ref{fig:diagonal-twoway}):
there is a location $\ell_x=(x,y_x)$ between $\ell$ and $\ell'$ such that the words
$\out{\rho|Z_{\ell_x}^\upperleft}$ and $\out{\rho|Z_{\ell_x}^\lowerright}$ 
have length at most $\bound$, where 
$Z_{\ell_x}^\upperleft = [\ell_x,\ell'] \:\cap\: \big([0,x]\times\bbN\big)$
and $Z_{\ell_x}^\lowerright = [\ell,\ell_x] \:\cap\: \big([x,\omega]\times\bbN\big)$.

Besides the run $\rho$ and the decomposition, the transducer $\cT'$ will
also guess the locations $\ell_x=(x,y_x)$, that is, will annotate each $x$
with the corresponding $y_x$.
Without loss of generality, we can assume that the function that 
associates each position $x$ with the guessed location $\ell_x=(x,y_x)$ 
is monotone, namely, $x\le x'$ implies $\ell_x\leqtime\ell_{x'}$.
While the transducer $\cT'$ is in diagonal mode, the goal is to preserve 
the following invariant: 

\begin{quote}
\em
After reaching a position $x$ covered by the active diagonal, 
$\cT'$ must have produced the output of $\rho$ up to location $\ell_x$. 
\end{quote}

\noindent
To preserve the above invariant when moving from $x$ to the next 
position $x+1$, the transducer should output the word 
$\out{\rho[\ell_x,\ell_{x+1}]}$. This word consists of
the following parts:
\begin{enumerate}
  \item The words produced by the single transitions of $\rho[\ell_x,\ell_{x+1}]$
        with endpoints in $\{x,x+1\}\times\bbN$. 
        Note that there are at most $\hmax$ such words, 
        each of them has length at most $\cmax$, and they can all be determined 
        using the crossing sequences at $x$ and $x+1$ and the information
        about the levels of $\ell_x$ and $\ell_{x+1}$.
        We can thus assume that this information is readily available
        to the transducer.
  \item The words produced by the factors of $\rho[\ell_x,\ell_{x+1}]$ 
        that are intercepted by the interval $[0,x]$. 
        Thanks to the definition of diagonal, we know that
        the total length of these words is at most $\bound$.
        These words cannot be determined from the information
        on $\rho|x$, $\rho|x+1$, $\ell_x$, and $\ell_{x+1}$
        alone, so they need to be constructed while scanning the input.
        For this, some additional information needs to be stored. 
        
        More precisely, at each position $x$ of the input, 
        the transducer stores all the outputs produced by the factors of 
        $\rho$ that are intercepted by $[0,x]$ and that occur {\sl after} 
        a location of the form $\ell_{x'}$, for any $x'\ge x$ that is 
        covered by a diagonal.
        This clearly includes the previous words when $x'=x$, but also 
        other words that might be used later for processing other diagonals.
        Moreover, by exploiting the properties of diagonals,
        one can prove that those words have length at most $\bound$, 
        so they can be stored with triply exponentially many states.
        Using classical techniques, the stored information
        can be maintained while scanning the input $u$ using the
        guessed crossing sequences of $\rho$.
  \item The words produced by the factors of $\rho[\ell_x,\ell_{x+1}]$ 
        that are intercepted by the interval $[x+1,\omega]$. 
        These words must be guessed, since they depend on a portion
        of the input that has not been processed yet. 
        Accordingly, the guesses need to be stored into memory,
        in such a way that they can be checked later. For this, the transducer 
        stores, for each position $x$, the guessed words that correspond 
        to the outputs produced by the factors of $\rho$ intercepted by 
        $[x,\omega]$ and occurring {\sl before} a location of the form 
        $\ell_{x'}$, for any $x'\le x$ that is covered by a diagonal.
\end{enumerate}

\smallskip
\par\noindent\emph{Block mode.}~
Suppose that the active factor $\rho[\ell,\ell']$ is a $\bound$-block.
Let $I=[x,x']$ be the set of positions covered by this factor.
Moreover, for each position $z\in I$, let 
$Z^\leftshort_z = [\ell,\ell'] \:\cap\: \big([0,z]\times \bbN\big)$
and $Z^\rightshort_z = [\ell,\ell'] \:\cap\: \big([z,\omega]\times \bbN\big)$.
We recall the key property of a block
(cf.~Definition~\ref{def:factors-twoway} and Figure~\ref{fig:block-twoway}): 
the word $\out{\rho[\ell,\ell']}$ is almost periodic with bound $\bound$, 
and the words $\out{\rho|Z^\leftshort_x}$ and $\out{\rho|Z^\rightshort_{x'}}$ 
have length at most $\bound$.

For the sake of brevity, suppose that $\out{\rho[\ell,\ell']} = w_1\,w_2\,w_3$,
where $w_2$ is periodic with period $\bound$ and $w_1,w_3$
\reviewOne[inline]{$w_1,w_3$ I assume.
  \olivier[inline]{fixed (was $w_1,w_2$)}%
}%
have length at most $\bound$.
Similarly, let $w_0 = \out{\rho|Z^\leftshort_x}$ and $w_4 = \out{\rho|Z^\rightshort_{x'}}$.
The invariant preserved by $\cT'$ in block mode is the following: 

\begin{quote}
\em
After reaching a position $z$ covered by the active block $\rho[\ell,\ell']$, 
$\cT'$ must have produced the output of the prefix of $\rho$
up to location $\ell$, followed by a prefix of $\out{\rho[\ell,\ell']} = w_1\,w_2\,w_3$
of the same length as $\out{\rho|Z^\leftshort_z}$.
\end{quote}

\noindent
The initialization of the invariant is done when reaching the left 
endpoint $x$. At this moment, it suffices that $\cT'$ outputs 
a prefix of $w_1\,w_2\,w_3$ of the same length as 
$w_0 = \out{\rho|Z^\leftshort_x}$, thus bounded by $\bound$.
Symmetrically, when reaching the right endpoint $x'$,
$\cT'$ will have produced almost the entire word 
$\out{\rho[\ell,\ell']} \, w_1 \, w_2 \, w_3$,
but without the suffix $w_4 = \out{\rho|Z^\rightshort_{x'}}$ 
of length at most $\bound$. 
Thus, before moving to the next factor of the decomposition, the transducer will 
produce the remaining suffix, so as to complete the output 
of $\rho$ up to location $\ell_{i_x+1}$.

It remains to describe how the above invariant can be maintained
when moving from a position $z$ to the next position $z+1$ inside $I=[x,x']$.
For this, it is convenient to succinctly represent the word $w_2$ 
by its repeating pattern, say $v$, of length at most $\bound$. 
To determine the symbols that have to be output at each step,
the transducer will maintain a pointer on either $w_1\,v$ or $w_3$.
The pointer is increased in a deterministic way, and precisely
by the amount $|\out{\rho|Z^\leftshort_{z+1}}| - |\out{\rho|Z^\leftshort_z}|$.
The only exception is when the pointer lies in $w_1\,v$, but its 
increase would go over $w_1\,v$: in this case the transducer has 
the choice to either bring the pointer back to the beginning of $v$ 
(representing a periodic output inside $w_2$), or move it to $w_3$. 
Of course, this is a non-deterministic choice, but it can be 
validated when reaching the right endpoint of $I$.
Concerning the number of symbols that need to be emitted at each
step, this can be determined from the crossing sequences at
$z$ and $z+1$, and from the knowledge of the lowest and highest 
levels of locations that are at position $z$ and between 
$\ell$ and $\ell'$. We denote the latter levels by
$y^-_z$ and $y^+_z$, respectively.

Overall, this shows how to maintain the invariant of the block mode,
assuming that the levels $y^-_z,y^+_z$ are known, as well as
the words $w_0,w_1,v,w_3,w_4$ of bounded length.
Like the mapping $z \mapsto \ell_z=(z,y_z)$ used in diagonal mode, 
the mapping $z \mapsto (y^-_z,y^+_z)$ can be guessed and checked 
using the crossing sequences.
Similarly, the words $w_1,v,w_3$ can be guessed just before
entering the active block, and can be checked along the process.
As concerns the words $w_0,w_4$, these can be guessed and checked 
in a way similar to the words that we used in diagonal mode.
More precisely, for each position $z$ of the input, the
transducer stores the following additional information:
\begin{enumerate}
  \item the outputs produced by the factors of $\rho$ that are 
        intercepted by $[0,z]$ and that occur after the beginning
        $\ell''$ of some block, with $\ell''=(x'',y'')$ and $x''\ge z$;
  \item the outputs produced by the factors of $\rho$ that are 
        intercepted by $[z,\omega]$ and that occur before the ending
        $\ell'''$ of a block, where $\ell'''=(x''',y''')$ 
        and $x'''\le z$.
\end{enumerate}
By the definition of blocks, the above words have length 
at most $\bound$ and can be maintained while processing the input
and the crossing sequences.
Finally, we observe that the words, together with the information 
given by the lowest and highest levels $y^-_z,y^+_z$, for both $z=x$ and
$z=x'$, are sufficient for determining the content of $w_0$ and $w_4$.

\smallskip
We have just shown how to construct a one-way transducer $\cT' \subseteq \cT$ 
such that $\dom(\cT') \supseteq D$. 
From the above construction it is easy to see that the number of states
and transitions of $\cT'$, as well as the number of letters emitted by
each transition, are at most exponential in $\bound$. Since $\bound$ is 
doubly exponential in the size of $\cT$, this shows that $\cT'$ can
be constructed from $\cT$ in $3\exptime$.
Note that the triple exponential 
complexity  comes from the lengths of the words that need to be guessed 
and stored in the control states, and these lengths are bounded by $\bound$.
However, if $\cT$ is a sweeping transducer, then, according to the results 
proved in Section \ref{sec:characterization-sweeping}, the bound $\bound$ 
is simply exponential. In particular, in the sweeping case 
we can construct the one-way transducer $\cT'$ in $2\exptime$.
\end{proof}

\medskip
\subsection*{Generality of the construction.}
We conclude the section with a discussion on the properties of the one-way
transducer $\cT'$ constructed from $\cT$. Roughly speaking, we would like
to show that, even when $\cT$ is not one-way definable, $\cT'$ is somehow
the {\sl best one-way under-approximation of $\cT$}.
However, strictly speaking, the latter terminology is meaningless: 
if $\cT'$ is a one-way transducer strictly contained in $\cT$, then 
one can always find a better one-way transducer $\cT''$ that satisfies 
$\cT' \subsetneq \cT'' \subsetneq \cT$, for instance by extending $\cT'$ 
with a single input-output pair. Below, we formalize in an appropriate
way the notion of ``best one-way under-approximation''.

We are interested in comparing the domains of transducers, but only up to 
a certain amount. In particular, we are interested in languages that are 
preserved under pumping loops of runs of $\cT$. Formally, given a language 
$L$, we say that $L$ is \emph{$\cT$-pumpable} if $L \subseteq \dom(\cT)$ and 
for all words $u\in L$, all successful runs $\rho$ of $\cT$ on $u$, all 
loops $L$ of $\rho$, and all positive numbers $n$, the word $\pump_L^n(u)$ 
also belongs to $L$.
Clearly, the domain $\dom(\cT)$ of a transducer $\cT$ is a regular $\cT$-pumpable language. 

Another noticeable example of $\cT$-pumpable regular language is the domain 
of the one-way transducer $\cT'$, as defined in Proposition \ref{prop:construction-twoway}.
Indeed, $\dom(\cT')$ consists of words $u\in\dom(\cT)$ that induce
successful runs with $\bound$-decompositions, and the property of 
having a $\bound$-decomposition is preserved under pumping.

The following result shows that $\cT'$ is the best under-approximation 
of $\cT$ within the class of one-way transducers with $\cT$-pumpable domains:

\begin{cor}\label{cor:best-underapproximation}
Given a functional two-way transducer $\cT$, one can construct a one-way transducer $\cT'$ such that
\begin{itemize}
  \item $\cT' \subseteq \cT$ and $\dom(\cT')$ is $\cT$-pumpable,
  \item for all one-way transducers $\cT''$, if $\cT'' \subseteq \cT$ and $\dom(\cT'')$ is $\cT$-pumpable,
        then $\cT'' \subseteq \cT'$.
\end{itemize}
\end{cor}

\begin{proof}
The transducer $\cT'$ is precisely the one defined in Proposition \ref{prop:construction-twoway}.
As already explained, its domain $\dom(\cT')$ is a $\cT$-pumpable language. In particular, $\cT'$
satisfies the conditions in the first item.

For the conditions in the second item, consider a one-way transducer $\cT'' \subseteq \cT$ 
with a $\cT$-pumpable domain $L=\dom(\cT'')$. Let $\tilde\cT$ be the transducer obtained from
$\cT$ by restricting its domain to $L$. Clearly, $\tilde\cT$ is one-way definable, and one 
could apply Proposition \ref{prop:periodicity-twoway} to $\tilde\cT$, using $\cT''$ as a 
witness of one-way definability. In particular, when it comes to comparing the outputs of the
pumped runs of $\tilde\cT$ and $\cT''$, one could exploit the fact that the domain $L$ of $\cT''$,
and hence the domain of $\tilde\cT$ as well, is $\cT$-pumpable. This permits to derive
periodicities of inversions with the same bound $\bound$ as before, but only restricted 
to the successful runs of $\cT$ on the input words that belong to $L$. 
As a consequence, one can define $\bound$-decompositions of successful runs of $\cT$ 
on words in $L$, thus showing that $L \subseteq \dom(\cT')$. This proves that $\cT'' \subseteq \cT'$.
\end{proof}


\section{Complexity of the one-way definability problem}\label{sec:complexity}

In this section we analyze the complexity of the problem of deciding whether
a transducer $\cT$ is one-way definable. We begin with the case of a functional
two-way transducer. In this case, thanks to the results presented in 
Section~\ref{sec:characterization-twoway} page \pageref{testing-containment},
we know that $\cT$ is one-way 
definable if and only if $\dom(\cT) \subseteq D$, where $D$ is the language of words
$u\in\dom(\cT)$ such that all successful runs of $\cT$ on $u$ admit a $\bound$-decomposition.
In particular, the one-way definability problem reduces to an emptiness problem
for the intersection of two languages:
\[
  \cT \text{ one-way definable}
  \qquad\text{if and only if}\qquad
  \dom(\cT) \cap D^\complement = \emptyset.
\]
The following lemma exploits the characterization of Theorem~\ref{thm:main2} 
to show that the language $D^\complement$ can be recognized by a non-deterministic
finite automaton $\cA$ of triply exponential size w.r.t.~$\cT$. In fact,
this lemma shows that the automaton recognizing $D^\complement$ can be constructed
using doubly exponential {\sl workspace}. As before, we gain an exponent when 
restricting to sweeping transducers.

\reviewOne[inline]{Lemma 9.1: D is a poor choice of notation for the automaton given that there already is a language D. Might A (same font as current D) be an adequate choice?
\gabriele[inline]{Done. BTW. if anyone has changed the macros \cA, \cT, etc. mind that
                  now they are rendered as simple capital letters instead as with mathcal.
                  I like this, and it is more uniform than using caligraphic A for automata
                  and normal T for transducers... But we should check we are consistent!}%
}%

\begin{lem}\label{lem:D-complement}
Given a functional two-way transducer $\cT$, an  NFA $\cA$ recognizing
$D^\complement$ can be constructed in $2\expspace$.
Moreover, when $\cT$ is sweeping, the 
NFA $\cA$ can be constructed in $\expspace$.
\end{lem}

\begin{proof}
Consider an input word $u$. By Theorem~\ref{thm:main2} we know that 
$u\in D^\complement$ iff there exist a successful run $\rho$ of $\cT$ 
on $u$ and an inversion $\cI=(L_1,\ell_1,L_2,\ell_2)$ of $\rho$ such that
no positive number $p \le \bound$ is a period of the word
\[
  w_{\rho,\cI} ~=~ 
  \outb{\tr{\ell_1}} ~ \outb{\rho[\ell_1,\ell_2]} ~ \outb{\tr{\ell_2}}.
\]
The latter condition on $w_{\rho,\cI}$ can be rephrased as follows: 
there is a function $f:\{1,\dots,\bound\} \rightarrow \{1,\dots,|w_{\rho,\cI}|\}$
such that $w_{\rho,\cI}\big(f(p)\big) \neq w_{\rho,\cI}\big(f(p)+p\big)$ 
for all positive numbers $p\le\bound$.
In particular, each of the images of the latter function $f$, that is,
$f(1),\dots,f(\bound)$, can be encoded by a suitable marking of the 
crossing sequences of $\rho$. This shows that the run $\rho$, the 
inversion $\cI$, and the function $f$ described above can all be 
guessed within space $\cO(\bound)$: $\r$ is guessed on-the-fly, the
inversion is guessed by marking the anchor points, and for $f$ we only
store two symbols and a counter $\le\bound$, for each $1 \le i \le \bound$.
That is, any state of $\cA$ requires doubly
exponential space, resp.~simply exponential space, depending on whether $\cT$ is arbitrary 
two-way or sweeping. 
\end{proof}

As a consequence of the previous lemma, the emptiness problem for the language 
$\dom(\cT) \cap D^\complement$, and thus the one-way definability problem for $\cT$,
can be decided in $2\expspace$ or $\expspace$, depending on whether $\cT$ is
two-way or sweeping:

\begin{prop}\label{prop:complexity}
The problem of deciding whether a functional two-way transducer 
$\cT$ is one-way definable is in $2\expspace$. When $\cT$ is 
sweeping, the problem is in $\expspace$.
\end{prop}

\reviewOne[inline]{The transition from 9.2 to 9.3 is a tad abrupt. The kind of argument you make before Corollary 8.7 to explain that this makes your result more robust or tight than I expected in first approach could be useful here.
\gabriele[inline]{Done, see below}%
}%

\medskip
The last result of the section shows that functional two-way transducers 
are close to be the largest class for which a characterization 
of one-way definability is feasible: as soon as we consider 
arbitrary transducers (including non-functional ones),
the problem becomes undecidable.

\begin{prop}\label{prop:undecidability}
The one-way definability problem for \emph{non-functional} 
sweeping transducers is undecidable.
\end{prop}

\begin{proof}
The proof uses some ideas and variants of constructions provided in \cite{Ibarra78},
concerning the proof of undecidability of the equivalence problem for one-way
non-functional transducers.

We show a reduction from the Post Correspondence Problem (PCP).
A \emph{PCP instance} is described by two finite alphabets $\Sigma$ and $\Delta$
and two morphisms $f,g:\Sigma^*\then\Delta^*$. A \emph{solution} of such an instance
is any non-empty word $w\in\Sigma^+$ such that $f(w)=g(w)$. We recall that the problem
of testing whether a PCP instance has a solution is undecidable. 

Below, we fix a tuple $\tau=(\Sigma,\Delta,f,g)$ describing a PCP instance and we 
show how to reduce the problem of testing the {\sl non-existence of solutions} of 
$\tau$ to the problem of deciding {\sl one-way definability} of a relation computed 
by a sweeping transducer.
Roughly, the idea is to construct a relation $B_\tau$ between words over a suitable 
alphabet $\Gamma$ that encodes all the {\sl non-solutions} to the PCP instance 
$\tau$ (this is simpler than encoding solutions because the presence of errors 
can be easily checked). The goal is to have a relation $B_\tau$ that
(i) can be computed by a sweeping transducer and (ii) coincides with 
a trivial one-way definable relation when $\tau$ has no solution.

We begin by describing the encodings for the solutions of the PCP instance.
We assume that the two alphabets of the PCP instance, $\Sigma$ and $\Delta$, 
are disjoint and we use a fresh symbol $\#\nin \Sigma\cup\Delta$.
We define the new alphabet $\Gamma = \Sigma\cup\Delta\cup\{\#\}$ that will
serve both as input alphabet and as output alphabet for the transduction.
We call \emph{encoding} any pair of words over $\Gamma$ of the form
$(w\cdot u,w\cdot v)$, where $w\in\Sigma^+$, $u\in\Delta^*$, and $v\in\{\#\}^*$.
We will write the encodings as vectors to improve readability, e.g., as
\[
  \lbinom{w\cdot u}{w\cdot v} \ .
\]
We denote by $E_\tau$ the set of all encodings and we observe that $E_\tau$ 
is computable by a one-way transducer (note that this transducer needs
$\varepsilon$-transitions).
We then restrict our attention to the pairs in $E_\tau$ that are encodings
of valid solutions of the PCP instance. 
Formally, we call \emph{good encodings} the pairs in $E_\tau$ of the form 
\[
  \lbinom{w\cdot u}{w\cdot\#^{|u|}}
  \qquad\qquad\text{where } u = f(w) = g(w) \ .
\]
All the other pairs in $E_\tau$ are called \emph{bad encodings}. 
Of course, the relation that contains the good encodings is not computable 
by a transducer. On the other hand, we can show that the complement 
of this relation w.r.t.~$E_\tau$ is computable by a sweeping transducer. 
Let $B_\tau$ be the set of all bad encodings. 
Consider  $(w\cdot u,w\cdot \#^m)\in E_\tau$, with $w\in\Sigma^+$,
$u\in\Delta^*$, and $m\in\bbN$, and we observe that this pair belongs to 
$B_\tau$ if and only if one of the following conditions is satisfied:
\begin{enumerate}
  \item $m<|u|$,
        \label{enc1}
  \item $m>|u|$,
        \label{enc2}
  \item $u\neq f(w)$,
        \label{enc3}
  \item $u\neq g(w)$.
        \label{enc4}
\end{enumerate}
We explain how to construct a sweeping transducer $\cS_\tau$ that computes $B_\tau$.
Essentially, $\cS_\tau$ guesses which of the above conditions holds and processes 
the input accordingly. More precisely, if $\cS_\tau$ guesses that the first condition 
holds, then it performs a single left-to-right pass, first copying the prefix $w$ 
to the output and then producing a block of occurrences of the symbol $\#$ that is
shorter than the suffix $u$. This task can be easily performed while reading 
$u$: it suffices to emit at most one occurrence of $\#$ for each position in $u$, 
and at the same time guarantee that, for at least one such position, no occurrence 
of $\#$ is emitted. The second condition can be dealt with by a similar strategy:
first copy the prefix $w$, then output a block of $\#$ that is longer than
the suffix $u$. To deal with the third condition, the transducer $\cS_\tau$ 
has to perform two left-to-right passes, interleaved by a backward pass that 
brings the head back to the initial position.
During the first left-to-right pass, $\cS_\tau$ copies the prefix $w$ to the output. 
During the second left-to-right pass, it reads again the prefix $w$, but this time
he guesses a factorization of it of the form $w_1\:a\:w_2$. 
On reading $w_1$, $\cS_\tau$ will output $\#^{|f(w_1)|}$. 
After reading $w_1$, $\cS_\tau$ will store the symbol $a$ and move to the position
where the suffix $u$ begins. From there, it will guess a factorization of $u$
of the form $u_1\:u_2$, check that $u_2$ does not begin with $f(a)$, and
emit one occurrence of $\#$ for each position in $u_2$.
The number of occurrences of $\#$ produced in the output is thus
$m=|f(w_1)| + |u_2|$, and the fact that $u_2$ does not begin with $f(a)$
ensures that the factorizations of $w$ and $u$ do not match, i.e.
\[ m\neq|f(w)| \]
Note that the described behaviour does not immediately guarantee that $u\neq f(w)$.
Indeed, it may still happen that $u=f(w)$, but as a consequence $m\neq |u|$.
This case is already covered by the first and second condition, 
so the computation is still correct in the sense that it produces only 
bad encodings.
On the other hand, if $m$ happens to be the same as $|u|$,
then $|u| = m \neq |f(w)|$ and thus $u\neq f(w)$.
A similar behaviour can be used to deal with the fourth condition.

\smallskip
We have just shown that there is a sweeping non-functional transducer $\cS_\tau$
that computes the relation $B_\tau$ containing all the bad encodings.
Note that, if the PCP instance $\tau$ admits no solution, then all encodings 
are bad, i.e., $B_\tau=E_\tau$, and hence $B_\tau$ is one-way definable. 
It remains to show that when $\tau$ has a solution, $B_\tau$ is not one-way 
definable. Suppose that $\tau$ has solution $w\in\Sigma^+$ and let 
$\big(w\cdot u,\:w\cdot \#^{|u|}\big)$ be the corresponding good encoding, 
where $u=f(w)=g(w)$. 
Note that every exact repetition of $w$ is also a solution, and hence 
the pairs $\big(w^n\cdot u^n,\:w^n\cdot \#^{n\cdot|u|}\big)$ are also 
good encodings, for all $n\ge 1$.
 
Suppose, by way of contradiction, that there is a one-way transducer $\cT$ 
that computes the relation $B_\tau$. 
For every $n,m\in\bbN$, we define the encoding 
\[
  \alpha_{n,m} ~=~
  \lbinom{w^n\cdot u^m}{w^n\cdot \#^{m\cdot|u|}}
\]
and we observe that $\alpha_{n,m} \in B_\tau$ if and only if $n\neq m$
(recall that $w\neq\emptystr$ is the solution of the PCP instance $\tau$ and $u=f(w)=g(w)$).
Below, we consider bad encodings like the above ones,
where the parameter $n$ is supposed to be large enough.
Formally, we define the set $I$ of all pairs of indices $(n,m)\in\bbN^2$ 
such that (i) $n\neq m$ (this guarantees that $\alpha_{n,m}\in B_\tau$)
and (ii) $n$ is larger than the number $|Q|$ of states of $\cT$.

We consider some pair $(n,m)\in I$ and we choose a successful 
run $\rho_{n,m}$ of $\cT$ that witnesses the membership of 
$\alpha_{n,m}$ in $B_\tau$, namely, that reads the input 
$w^n\cdot u^m$ and produces the output $w^n\cdot \#^{m\cdot|u|}$.
We can split the run $\rho_{n,m}$ into a prefix $\olft\rho_{n,m}$ 
and a suffix $\ort\rho_{n,m}$ in such a way that $\olft\rho_{n,m}$ 
consumes the prefix $w^n$ and $\ort\rho_{n,m}$ consumes the remaining 
suffix $u^m$.
Since $n$ is larger than the number of state of $\cT$, we can find a
factor $\hat\rho_{n,m}$ of $\olft\rho_{n,m}$ that starts and ends with 
the same state and consumes a non-empty exact repetition of $w$, 
say $w^{n_1}$, for some $1\le n_1\le |Q|$.
We claim that the output produced by the factor $\hat\rho_{n,m}$ 
must coincide with the consumed part $w^{n_1}$ of the input.
Indeed, if this were not the case, then deleting the factor $\hat\rho_{n,m}$ 
from $\rho_{n,m}$ would result in a new successful run that reads
$w^{n-n_1}\cdot u^m$ and produces $w^{n-n_2}\cdot \#^{m\cdot|u|}$
as output, for some $n_2\neq n_1$. This however would contradict
the fact that, by definition of encoding, the possible outputs produced 
by $\cT$ on input $w^{n-n_1}\cdot u^m$ must agree on the prefix $w^{n-n_1}$. 
We also remark that, even if we do not annotate this explicitly, the number
$n_1$ depends on the choice of the pair $(n,m)\in I$. This number, however,
range over the fixed finite set $J = \big[1,|Q|\big]$.

We can now pump the factor $\hat\rho_{n,m}$ of the run $\rho_{n,m}$ any 
arbitrary number of times. In this way, we obtain new successful runs of 
$\cT$ that consume inputs of the form 
$w^{n+k\cdot n_1}\cdot u^m$ and produce outputs of the form
$w^{n+k\cdot n_1}\cdot \#^m$, for all $k\in\bbN$.
In particular, we know that $B_\tau$ contains all 
pairs of the form $\alpha_{n+k\cdot n_1,m}$.
Summing up, we can claim the following:

\begin{clm}
There is a function $h:I\then J$ such that, for all pairs $(n,m)\in I$,
\[
  \big\{ (n+k\cdot h(n,m),m) ~\big|~ k\in\bbN \big\} \:\subseteq\: I \ .
\]
\end{clm}

\noindent
We can now head towards a contradiction. Let $\tilde n$ 
be the maximum common multiple of the numbers $h(n,m)$, 
for all $(n,m)\in I$. Let $m=n+\tilde n$ and observe that 
$n\neq m$, whence $(n,m)\in I$. Since $\tilde n$ is a multiple
of $h(n,m)$, we derive from the above claim that the pair
$(n+\tilde n,m) = (m,m)$ also belongs to $I$.
However, this contradicts the definition of $I$, since we
observed earlier that $\alpha_{n,m}$ is a bad encoding if
and only if $n\neq m$.
We conclude that $B_\tau$ is not one-way definable 
when $\tau$ has a solution.
\end{proof}

\section{Conclusions}\label{sec:conclusions}

It was shown in \cite{fgrs13} that it is decidable whether 
a given two-way transducer can be implemented by some one-way
transducer. However, the provided algorithm has non-elementary complexity.

The main contribution of our paper is a new algorithm that solves 
the above question with elementary complexity, precisely in $2\expspace$.
The algorithm is based on a characterization of those transductions, 
given as two-way transducers, that can be realized by one-way 
transducers. The flavor of our characterization is 
different from that of \cite{fgrs13}. The approach of
\cite{fgrs13} is based on a variant of Rabin and Scott's construction 
\cite{RS59} of one-way automata, and on local modifications of
the two-way run. Our approach instead relies on the global notion
of \emph{inversion} and more involved combinatorial arguments.
The size of the one-way transducer that we obtain is triply exponential 
in the general case, and doubly exponential in the sweeping case, and 
the latter is optimal.
\felix{add last sentence. Maybe we can do more on future work ? The characterization of sweeping in term of rational relations maybe ?}%
\gabriele{I did not like the sentence in the end, so I moved it here and rephrased (and it was not future work!)}%
The approach described here was adapted to characterize functional two-way transducers 
that are equivalent to some sweeping transducer with either known or
unknown number of passes (see~\cite{bgmp16}, \cite{Bas17} for details).

\reviewOne[inline]{The conclusion gives no clue towards future works, either current, considered or at all possible.
\olivier[inline]{added the following paragraph}%
}%
Our procedure considers non-deterministic transducers,
both for the initial two-way transducer, and
for the equivalent one-way transducer, if it exists.
\gabriele{Improved and mentioned that determinism does not bring much to our problem}%
Deterministic two-way transducers are as expressive as
non-deterministic functional ones. This means that 
starting from deterministic two-way transducers would address 
the same problem in terms of transduction classes, but could 
in principle yield algorithms with better complexity.
\gabriele{Added this tiny remark}%
We also recall that Proposition \ref{prop:lower-bound} 
gives a tight lower bound for the size of a one-way transducer 
equivalent to a given deterministic sweeping transducer. This
latter result basically shows that there is no advantage in
considering one-way definability for deterministic variants 
of sweeping transducers, at least with respect to the size
of the possible equivalent one-way transducers.
 
\gabriele{slightly rephrased}%
A variant of the one-way definability problem asks whether a given
two-way transducer is equivalent to some \emph{deterministic} one-way transducer.
A decision procedure for this latter problem is obtained by combining 
our characterization with the classical  algorithm that determines whether a one-way 
transducer can be determinized \cite{cho77,bealcarton02,weberklemm95}.
In terms of complexity, it is conceivable that a better algorithm 
may exist for deciding definability by deterministic one-way transducers,
since in this case one can rely on structural properties 
 that characterize 
deterministic transducers. \anca{I deleted the last phrase (``on the
  contrary, our approach etc'') because at this point, the reader
  should have remembered that we do only guessing.}


\section*{Acknowledgment}
\noindent The authors wish to acknowledge fruitful discussions with
Emmanuel Filiot, Isma\"el Jecker and Sylvain Salvati. We also thank
the referees for their very careful reading and the suggestions for improvement.

\bibliographystyle{abbrv}
\bibliography{biblio}

\newpage
\appendix
\section{}\label{app:proof-component}

Here we give a fully detailed proof of Lemma \ref{lem:component},
for which we recall the statement below:

\medskip\noindent
{\bfseries Lemma \ref{lem:component}.}
{\em 
Let $C$ be a component of a loop $L=[x_1,x_2]$. 
The nodes of $C$ are precisely the levels in the interval $[\min(C),\max(C)]$.
Moreover, if $C$ is left-to-right (resp.~right-to-left), then $\max(C)$ 
is the smallest level $\ge \min(C)$ 
such that between $(x_1,\min(C))$ and $(x_2,\max(C))$ (resp.~$(x_2,\min(C))$ and $(x_1,\max(C))$)
there are equally many $\LL$-factors and $\RR$-factors intercepted by $L$.
%
}
\medskip

\begin{proof}
To ease the understanding the reader may refer to Figure~\ref{fig:edges}, 
that shows some factors intercepted by $L$ and the corresponding edges in the flow.

We begin the proof by partitioning the set of levels of the flow into 
suitable intervals as follows.
We observe that every loop $L=[x_1,x_2]$ intercepts equally many
$\LL$-factors and $\RR$-factors. This is so because the crossing
sequences at $x_1,x_2$ have the same length $h$. 
We also observe that the sources of the factors intercepted by $L$
are either of the form $(x_1,y)$, with $y$ even, or $(x_2,y)$, with $y$ odd.
For any location $\ell\in\{x_1,x_2\}\times\bbN$ that is the source
of an intercepted factor, we define $d_{\ell}$ to be the difference 
between the number of $\LL$-factors and the number of $\RR$-factors 
intercepted by $L$ that {\sl end} at a location {\sl strictly before} $\ell$.
Intuitively, $d_{\ell}=0$ when the prefix of the run up to location $\ell$ 
has visited equally many times the position $x_1$ and the position $x_2$. 
For the sake of brevity, we let $d_y=d_{(x_1,y)}$ for an even level $y$,
and $d_y=d_{(x_2,y)}$ for an odd level $y$. 
Note that $d_0=0$. We also let $d_{h+1}=0$, by convention.

We now consider the numbers $z$'s, with $0\le z\le h+1$, such that 
$d_z=0$, that is: $0 = z_0 < z_1 < \dots < z_k = h+1$.  
Using a simple induction, we prove that for all $i \le k$, 
the parity of $z_i$ is the same as the parity of its index $i$. 
The base case $i = 0$ is trivial, since $z_0 = 0$. For the inductive
case, suppose that $z_i$ is even (the case of $z_i$ odd is similar). 
We prove that $z_{i+1}$ is odd by a case distinction based on the 
type of factor intercepted by $L$ that starts at level $z_i$. 
If this factor is an $\LR$-factor, then it ends at the same level $z_i$,
and hence $d_{z_i+1} = d_{z_i} = 0$, which implies that $z_{i+1} = z_i +1$ is odd. 
Otherwise, if the factor is an $\LL$-factor, then for all levels $z$ strictly
between $z_i$ and $z_{i+1}$, we have $d_z > 0$, and since $d_{z_{i+1}} =0$, 
the last factor before $z_{i+1}$ must decrease $d_z$, that is, must be an $\RR$-factor. 
This implies that $(x_2,z_{i+1})$ is the source of an intercepted factor, 
and thus $z_{i+1}$ is odd.

The levels $0 = z_0 < z_1 < \dots < z_k = h+1$ induce a partition of
the set of nodes of the flow into intervals of the form $Z_i=[z_i,z_{i+1}-1]$.
To prove the lemma, it is suffices to show that the subgraph of the flow 
induced by each interval $Z_i$ is connected. Indeed, because the union of 
the previous intervals covers all the nodes of the flow, and because each 
node has one incoming and one outgoing edge, this will imply that the 
intervals coincide with the components of the flow.

\medskip
Now, let us fix an interval of the partition, which we denote by $Z$ 
to avoid clumsy notation. Hereafter, we will focus on the edges of subgraph
of the flow induced by $Z$ (we call it \emph{subgraph of $Z$} for short). 
We prove a few basic properties of these edges.
For the sake of brevity, we call $\LL$-edges the edges of the subgraph of $Z$
that correspond to the $\LL$-factors intercepted by $L$, and similarly for 
the $\RR$-edges, $\LR$-edges, and $\RL$-edges.

We make a series of assumption to simplify our reasoning.
First, we assume that the edges are ordered based on the occurrences of the 
corresponding factors in the run. For instance, we may say the first, 
second, etc.~$\LR$-edge (of the subgraph of $Z$) --- from now on, we 
tacitly assume that the edges are inside the subgraph of $Z$.
Second, we assume that the first edge of the subgraph of $Z$ starts 
at an even node, namely, it is an $\LL$-edge or an $\LR$-edge 
(if this were not the case, one could apply symmetric arguments to prove the lemma).
From this it follows that the subgraph contains $n$ $\LR$-edges interleaved by
$n-1$ $\RL$-edges, for some $n>0$. 
Third, we assume that $\min(Z)=0$, in order to avoid clumsy notations 
(otherwise, we need to add $\min(Z)$ to all the levels considered hereafter).

Now, we observe that, by definition of $Z$, there are equally 
many $\LL$-edges and $\RR$-edges: 
indeed, the difference between the number of $\LL$-edges and the number
of $\RR$-edges at the beginning and at the end of $Z$ is the same, namely,
$d_z = 0$ for both $z=\min(Z)$ and $z=\max(Z)$.
It is also easy to see that the $\LL$-edges and the $\RR$-edges 
are all of the form $y \rightarrow y+1$, for some level $y$.
We call these edges \emph{incremental edges}.

For the other edges, we denote by $\ort y_i$ (resp.~$\olft y_i$) 
the source level of the $i$-th $\LR$-edge (resp.~the $i$-th $\RL$-edge). 
Clearly, each $\ort y_i$ is even, and each $\olft y_i$ is odd,
and $i\le j$ implies $\ort y_i < \ort y_j$ and $\olft y_i < \olft y_j$.
Consider the location $(x_1,\ort y_i)$, which is the source of the 
$i$-th $\LR$-edge (e.g.~the edge in blue in the figure).  
The latest location at position $x_2$ that 
precedes $(x_1,\ort y_i)$ must be of the form 
$(x_2,\olft y_{i-1})$, provided that $i>1$.
This implies that, for all $1<i\le n$, the $i$-th $\LR$-edge 
is of the form $\ort y_i \rightarrow \olft y_{i-1} + 1$.
For $i=1$, we recall that $\min(Z)=0$ and observe that the 
first location at position $x_2$ that occurs after the location 
$(x_1,0)$ is $(x_2,0)$, and thus the
first $\LR$-edge has a similar form: $\ort y_1 \rightarrow \olft y_0 + 1$,
where $\olft y_0 = -1$ by convention.

Using symmetric arguments, we see that the $i$-th $\RL$-edge 
(e.g.~the one in red in the figure) is of the form 
$\olft y_i \rightarrow \ort y_i + 1$. In particular,
the last $\LR$-edge starts at the level $\ort y_n = \max(Z)$.

Summing up, we have just seen that the edges of the subgraph of $Z$ are of the following forms:
\begin{itemize}
  \item \rightward{$y \rightarrow y+1$}                      
        \hspace{25mm} (incremental edges),
  \item \rightward{$\ort y_i \rightarrow \olft y_{i-1} + 1$} 
        \hspace{25mm} ($i$-th $\LR$-edge, for $i=1,\dots,n$),
  \item \rightward{$\olft y_i \rightarrow \ort y_i + 1$}     
        \hspace{25mm} ($i$-th $\RL$-edge, for $i=1,\dots,n-1$).
\end{itemize}
In addition, we have $\ort y_i +1 = \olft y_i + 2d_{\olft y_i}$.
Since $d_z > 0$ for all $\min(Z) < z < \max(Z)$, this implies that 
$\ort y_i > \olft y_i$.

\medskip
The goal is to prove that the subgraph of $Z$ is strongly connected, 
namely, it contains a cycle that visits all its nodes. As a matter of fact, 
because components are also strongly connected subgraphs, and because every 
node in the flow has in-/out-degree $1$, this will imply that the considered 
subgraph coincides with a component $C$, thus implying that the nodes in $C$ 
form an interval.
Towards this goal, we will prove a series of claims that aim at
identifying suitable sets of nodes that are covered by paths in 
the subgraph of $Z$.
Formally, we say that a path \emph{covers} a set $Y$ 
if it visits all the nodes in $Y$, and possibly other nodes.
As usual, when we talk of edges or paths, we tacitly 
understand that they occur inside the subgraph of $Z$.
On the other hand, we do not need to assume $Y\subseteq Z$,
since this would follow from the fact that $Y$ is covered by
a path inside $Z$.
For example, the right hand-side of Figure~\ref{fig:edges} shows 
a path from $\ort y_i$ to $\ort y_i+1$ that covers the set 
$Y=\{\ort y_i,\ort y_i+1\} \cup [\olft y_{i-1}+1,\olft y_i]$.

The covered sets will be intervals of the form
\[
  Y_i ~=~ [\olft y_{i-1}+1,\olft y_i].
\]
Note that the sets $Y_i$ are well-defined for all $i=1,\dots,n-1$, but not 
for $i=n$ since $\olft y_n$ is not defined either (the subgraph of $Z$
contains only $n-1$ $\RL$-edges). 

\begin{clm}
For all $i=1,\dots,n-1$, there is a path from $\ort y_i$ to $\ort y_i+1$ 
that covers $Y_i$ (for short, we call it an \emph{incremental path}).
\end{clm}

\begin{proof}
We prove the claim by induction on $i$. The base case $i=1$ is rather easy. 
Indeed, we recall the convention that $\olft y_0+1 = \min(Z)=0$. 
In particular, the node $\olft y_0+1$ is the target of the first 
$\LR$-edge of the subgraph of $Z$.
Before this edge, according to the order induced by the run, 
we can only have $\LL$-edges of the form $y \rightarrow y+1$, 
with $y=0,2,\dots,\ort y_1-2$. Similarly, after the $\LR$-edge we have 
$\RR$-edges of the form $y \rightarrow y+1$, with $y=1,3,\dots,\olft y_1-2$.
Those incremental edges can be connected to form the path 
$\olft y_{i-1}+1 \rightarrow^* \olft y_1$ that covers 
the interval $[\olft y_0+1,\olft y_1]$ . 
By prepending to this path the $\LR$-edge $\ort y_1 \rightarrow \olft y_0 + 1$,
and by appending the $\RL$-edge $\olft y_1 \rightarrow \ort y_1 + 1$,
we get a path from $\ort y_1$ to $\ort y_1+1$ that covers 
the interval $[\olft y_0+1,\olft y_1]$. The latter interval is precisely 
the set $Y_1$.

For the inductive step, we fix $1<i<n$ and we construct the desired
path from $\ort y_i$ to $\ort y_i+1$. The initial edge of this path
is defined to be the $\LR$-edge $\ort y_i \rightarrow \olft y_{i-1} + 1$.
Similarly, the final edge of the path will be the $\RL$-edge 
$\olft y_i \rightarrow \ort y_i + 1$, which exists since $i<n$.
It remains to connect $\olft y_{i-1} + 1$ to $\olft y_i$.
For this, we consider the edges that depart from nodes strictly 
between $\olft y_{i-1}$ and $\olft y_i$.

Let $y$ be an arbitrary node in $[\olft y_{i-1}+1,\olft y_i-1]$. 
Clearly, $y$ cannot be of the form $\olft y_j$, for some $j$, 
because it is strictly between $\olft y_{i-1}$ and $\olft y_i$.
So $y$ cannot be the source of an $\RL$-edge.
Moreover, recall that the $\LL$-edges and the $\RR$-edges are the
of the form $y \rightarrow y+1$. As these incremental edges 
do not pose particular problems for the construction of the path, 
we focus mainly on the $\LR$-edges that depart from nodes inside 
$[\olft y_{i-1}+1,\olft y_i-1]$.

Let $\ort y_j \rightarrow \olft y_{j-1}+1$ be such an $\LR$-edge, 
for some $j$ such that $\ort y_j \in [\olft y_{i-1}+1,\olft y_i-1]$.
If we had $j\ge i$, then we would have 
$\ort y_j \ge \ort y_i > \olft y_i$, but this would contradict
the assumption that $\ort y_j \in [\olft y_{i-1}+1,\olft y_i-1]$.
So we know that $j<i$. 
This enables the use of the inductive hypothesis, which implies
the existence of an incremental path from $\ort y_j$ to $\ort y_j+1$ 
that covers the interval $Y_j$. 

Finally, by connecting the above paths using the 
incremental edges, and by adding the initial and 
final edges $\ort y_i \rightarrow \olft y_{i-1} + 1$ and
$\olft y_i \rightarrow \ort y_i + 1$, we obtain a path 
from $\ort y_i$ to $\ort y_i+1$. It is easy to see that this
path covers the interval $Y_i$.
\end{proof}

Next, we define 
\[
  Y ~=~ [\olft y_{n-1}+1,\ort y_n] ~\cup \bigcup_{1\le i<n} Y_i.
\]
We prove a claim similar to the previous one, 
but now aiming to cover $Y$ with a {\sl cycle}.
Towards the end of the proof we will argue that the set 
$Y$ coincides with the full interval $Z$, 
thus showing that there is a component $C$ 
whose set of notes is precisely $Z$.

\begin{clm}
There is a cycle that covers $Y$.
\end{clm}

\begin{proof}
It is convenient to construct our cycle starting from the last 
$\LR$-edge, that is, $\ort y_n \rightarrow \olft y_{n-1} + 1$, 
since this will cover the upper node $\ort y_n=\max(Z)$.
From there we continue to add edges and incremental paths, 
following an approach similar to the proof of the previous 
claim, until we reach the node $\ort y_n$ again.
More precisely, we consider the edges that depart from nodes strictly 
between $\olft y_{n-1}$ and $\ort y_n$. As there are only $n-1$ $\RL$-edges,
we know that every node in the interval $[\olft y_{n-1}+1,\ort y_n-1]$
must be source of an $\LL$-edge, an $\RR$-edge, or an $\LR$-edge.
As usual, incremental edges do not pose particular problems for 
the construction of the cycle, so we focus on the $\LR$-edges.
Let $\ort y_i \rightarrow \olft y_{i-1}+1$ be such an $\LR$-edge,
with $\ort y_i \in [\olft y_{n-1}+1,\ort y_n-1]$. Since $i < n$, we know 
from the previous claim that there is a path from $\ort y_i$ to $\ort y_i+1$
that covers $Y_i$.
We can thus build a cycle $\pi$ by connecting the above  
paths using the incremental edges and the $\LR$-edge 
$\ort y_n \rightarrow \olft y_{n-1} + 1$.

By construction, the cycle $\pi$ covers 
the interval $[\olft y_{n-1}+1,\ort y_n]$, and for every
$i<n$, if $\pi$ visits $\ort y_i$, then $\pi$ covers $Y_i$.
So to complete the proof --- namely, to show that $\pi$ covers the entire set $Y$ ---
it suffices to prove that $\pi$ visits each node $\ort y_i$, with $i<n$.

Suppose, by way of contradiction, that $\ort y_i$ is the node 
with the highest index $i<n$ that is not visited by $\pi$.
Recall that $\ort y_i > \olft y_i$. This shows that
\[
  \ort y_i ~\in~ [\olft y_i+1,\ort y_n] ~= \bigcup_{i\le j<n-1}[\olft y_j+1,\olft y_{j+1}] ~\cup~ [\olft y_{n-1}+1,\ort y_n].
\]
As we already proved that $\pi$ covers the interval $[\olft y_{n-1}+1,\ort y_n]$,
we know that $\ort y_i \in [\olft y_j+1,\olft y_{j+1}]$ for some $j$ with $i\le j<n-1$.
Now recall that $\ort y_i$ is the highest node that is not visited by $\pi$. This means
that $\ort y_{j+1}$ is visited by $\pi$. Moreover, since $j+1<n$, we know 
that $\pi$ uses the incremental path from $\ort y_{j+1}$ to $\ort y_{j+1}+1$, which
covers $Y_{j+1} = [\olft y_j+1,\olft y_{j+1}]$. But this contradicts the fact that $\ort y_i$ 
is not visited by $\pi$, since $\ort y_i \in [\olft y_j+1,\olft y_{j+1}]$.
\end{proof}

We know that the set $Y$ is covered by a cycle of the subgraph of $Z$, and that
$Z$ is an interval whose endpoints are consecutive levels $z<z'$, with $d_z=d_{z'}=0$.
For the homestretch, we prove that $Y=Z$. This will imply that the nodes of 
the cycle are precisely the nodes of the interval $Z$. 
Moreover, because the cycle must coincide with a component $C$ of the flow 
(recall that all the nodes have in-/out-degree $1$), this will show that 
the nodes of $C$ are precisely those of $Z$. 

To prove $Y=Z$ it suffices to recall its definition as the union of 
the interval $[\olft y_{n-1}+1,\ort y_n]$ with the sets $Y_i$, for all $i=1,\dots,n-1$.
Clearly, we have that $Y\subseteq Z$.
For the converse inclusion, we also recall that 
$\olft y_0+1 = 0 = \min(Z)$ and $\ort y_n=\max(Z)$. 
Consider an arbitrary level $z\in Z$. Clearly, we have either 
$z\le \olft y_i$, for some $1\le i<n$, or $z > \olft y_n$. 
In the former case, by choosing the smallest index $i$ such that $z \le \olft y_i$, 
we get $z\in [\olft y_{i-1}+1,\olft y_i]$, whence $z\in Y_i \subseteq Y$.
In the latter case, we immediately have $z\in Y$, by construction.
\end{proof}


\end{document}